\documentclass[journal]{IEEEtran}

\usepackage{microtype}
\usepackage{graphicx}
\usepackage{subfigure}
\usepackage{booktabs} 
\usepackage{longtable}

\usepackage{xcolor}
\newcommand{\lsc}[1]{{\color[rgb]{0,0,0}{{#1}}}}
\newcommand{\rev}[1]{{\color[rgb]{0,0,0}{{#1}}}}

\usepackage{hyperref}

\usepackage{amssymb}
\usepackage{amsmath}
\usepackage{amsthm}
\usepackage{algorithm}
\usepackage{algorithmic}

\usepackage{mathtools}
\DeclarePairedDelimiter\ceil{\lceil}{\rceil}
\DeclarePairedDelimiter\floor{\lfloor}{\rfloor}
\DeclareMathOperator\caret{\raisebox{1ex}{$\scriptstyle\wedge$}}

\usepackage[flushleft]{threeparttable}

\theoremstyle{plain}
\newtheorem{thm}{Theorem}[section]

\newtheorem{lem}[thm]{Lemma}

\theoremstyle{definition}

\theoremstyle{remark}
\newtheorem{remark}{Remark}[section]

\graphicspath{{figures/}}

\usepackage{upgreek}
\usepackage{bm,MnSymbol}
\newcommand{\mat}[1]{\mbox{\boldmath$#1$}} 
\newcommand{\set}[1]{{\mathcal #1}}    

\newcommand{\real}{\mathbb{R}}         

\renewcommand{\pi}{\uppi} 

\newcommand{\opof}[2]{\mathop{{\rm #1}}\left\{#2\right\}}         

\newcommand{\traceof}[1]{\opof{trace}{#1}}       
\newcommand{\expvof}[1]{\opof{\mathbb{E}}{#1}}   

\newcommand{\vecelem}[2]{\errmessage{DO NOT USE}}    

\newcommand{\normof}[2]{\left\|#1\right\|_{#2}}
\newcommand{\twonorm}[1]{\normof{#1}{}}                  

\newcommand{\trans}{{\rm T}}   

\usepackage{color}

\newcounter{numtodos}

\renewcommand{\mat}[1]{\mathbf{#1}}

\renewcommand{\opof}[2]{\mathop{{\rm #1}}\left(#2\right)}
\newcommand{\proxopof}[3]{{\rm #1}_{#2}\left(#3\right)}         

\DeclareMathOperator*{\argmax}{argmax}
\DeclareMathOperator*{\argmin}{argmin}

\newcommand{\indicator}[1]{\mathbf{1}\left(#1\right)}

\newcommand{\inp}[2]{\left\langle#1, #2\right\rangle}
\newcommand{\Covof}[1]{\mathop{{\rm Cov}}\left(#1\right)}         
\newcommand{\Corrof}[1]{\mathop{{\rm Corr}}\left(#1\right)}         

\renewcommand{\expvof}[1]{\mathop{{\rm \mathbb{E}}}\left[#1\right]}   
\renewcommand{\maxof}[1]{\opof{max}{#1}}           
\renewcommand{\minof}[1]{\opof{min}{#1}}           

\newcommand{\signof}[1]{\opof{sign}{#1}}           

\newcommand{\abs}[1]{\left|#1\right|}

\newcommand\coolover[2]{\mathrlap{\smash{\overbrace{
    \begin{matrix} #2 \end{matrix}}^{\mbox{$#1$}}}}#2}

\newcommand{\NA}{---}

\begin{document}

\title{Grouping effects of sparse CCA models in variable selection}
	
%
%
\author{Kefei Liu, Qi Long, Li Shen
	\thanks{Department of Biostatistics, Epidemiology and Informatics, University of Pennsylvania, Philadelphia, Pennsylvania, USA. Email:li.shen@pennmedicine.upenn.edu.}
}

\maketitle
\begin{abstract}
The sparse canonical correlation analysis (SCCA) is a bi-multivariate association model that finds sparse linear combinations of two sets of variables that are maximally correlated with each other. In addition to the standard SCCA model, a simplified SCCA criterion which maixmizes the cross-covariance between a pair of canonical variables instead of their cross-correlation, is widely used in the literature due to its computational simplicity. However, the behaviors/properties of the solutions of these two models remain unknown in theory. In this paper, we analyze the grouping effect of the standard and simplified SCCA models in variable selection. In high-dimensional settings, the variables often form groups with high within-group correlation and low between-group correlation. Our theoretical analysis shows that for grouped variable selection, the simplified SCCA jointly selects or deselects a group of variables together, while the standard SCCA randomly selects a few dominant variables from each relevant group of correlated variables. Empirical results on synthetic data and real imaging genetics data verify the finding of our theoretical analysis.
\end{abstract}


\section{Introduction}
Canonical correlation analysis (CCA) \cite{Hotelling1936,hardoon2004canonical} is a multivariate statistical method which investigates the associations between two sets of variables. It has found applications in statistics \cite{klami2013bayesian}, data mining and machine learning \cite{hardoon2004canonical,sun2010canonical}, functional magnetic resonance imaging \cite{worsley1997characterizing,friman2001detection}, genomics \cite{yamanishi2003extraction} and other fields \cite{via2005canonical}. Given two data sets $\mat{X} \in \real^{n \times p}$ and $\mat{Y} \in \real^{n \times q}$ measured on the same set of $n$ samples, CCA seeks linear combinations of the variables in $\mat{X}$ and those in $\mat{Y}$ that are maximally correlated with each other:
\begin{equation*}\label{equ:CCA_model}
\underset{\mat{u},\mat{v}}{\text{maximize}} \; \mat{u}^\trans \mat{X}^\trans \mat{Y} \mat{v} \quad
\text{s.t.} \quad   \mat{u}^\trans \mat{X}^\trans \mat{X} \mat{u} \leq 1, \mat{v}^\trans \mat{Y}^\trans \mat{Y} \mat{v} \leq 1,
\end{equation*}
where $\mat{X}$ and $\mat{Y}$ are column-centered to zero mean.

Compared with multivariate multiple regression, the CCA is ``symmetric" and more flexible in finding variables from both $\mat{X}$ and $\mat{Y}$ to predict each other well. However, in high dimensional setting ($n<p$) such as linking imaging to genomics \cite{hariri2003imaging,shen2020pieee}, the CCA breaks down because it has infinitely many solutions. In particular, the solution can have any support of cardinality greater than or equal to $n$, which means that the CCA can select an arbitrary set of $n$ or more variables. To handle that, the sparse CCA (SCCA)~\cite{waaijenborg2008quantifying,hardoon2011sparse,chu2013sparse,chi2013isbi,suo2017sparse} utilizes the L1 sparsity regularization to select a subset of variables, which can improve the interpretability, stability as well as performance in variable selection.

A main drawback of the SCCA is that it is computationally expensive.
To reduce the computational load, a common practice is to replace the covariance matrices $\mat{X}^\trans \mat{X}$ and $\mat{Y}^\trans \mat{Y}$ in the L2 constraints with diagonal matrices~\cite{Parkhomenko2009,witten2009a,witten2009b,chen2012_aistat,chen2013structure}. The resulting simplified SCCA model allows a closed-form solution for solving each subproblem (update of $\mat{u}$ with $\mat{v}$ fixed or vice versa) and is thus computationally more efficient.

However, the fundamental difference between the standard and simplified SCCA in variable selection remains unclear, particularly in the theoretical properties of their solutions.
In~\cite{witten2009a,chen2013structure}, the use of the simplified SCCA model is justified based only on the empirical observation that ``in high-dimensional classification problems~\cite{Dudoit2002,Tibshirani2003}, treating the covariance matrix as diagonal can yield good results".
In this paper, we attempt to close this gap by investigating the properties of the solutions of the standard and simplified SCCA models. 

Our main contributions are summarized as follows.
\begin{itemize}
  \item The behaviors of the standard and simplified SCCA models in grouped variable selection is theoretically characterized. In high-dimension small sample-size problems, the variables often form groups of various sizes with high within-group correlation and low between-group correlation. It shows that the simplified SCCA jointly selects or deselects a group of correlated variables together, while the standard SCCA tends to select a few dominant variables from each relevant group of correlated variables. This finding could be used by practitioners using SCCA, allowing them to select the proper method for their tasks.

  \item The Lemma 2.2 of~\cite{witten2009a} is extended from $c \in [\sqrt{|\mathcal{S}|}, \infty)$ to $c \in (0, \infty)$, where $\set{S} = \left\{i: i \in \argmax_j \abs{a_j} \right\}$. The Lemma 2.2 of~\cite{witten2009a}, which solves $\underset{\mat{u}}{\text{maximize}} \; \mat{a}^\trans \mat{u}$ subject to $\twonorm{\mat{u}}^2 \leq 1, \normof{\mat{u}}{1} \leq c$, is a key component of the simplified SCCA algorithm used to solve the subproblems at each iteration of the alternating optimization algorithm. However, the lemma fails to provide a solution to the above problem for $c \in (0, \sqrt{|\mathcal{S}|})$.

  \item Greedy algorithms to sequentially compute multiple canonical components for standard and simplified SCCA are derived and presented. To the best of our knowledge, these algorithms are new.
\end{itemize}

\emph{Notation}: Scalars are denoted as italic letters, column vectors as boldface lowercase letters, and matrices as boldface capitals. The $j$-th column vector of a matrix $\mat{X}$ is denoted as $\mat{x}_j$. The superscript $^\trans$ stands for the transpose. The $\twonorm{\mat{u}}$ and $\normof{\mat{u}}{1}$ denote the Euclidean norm and $\ell_1$ norm of a vector $\mat{u}$, respectively. \rev{The $\sigma_{\rm max}\left(\mat{A}\right)$ and $\lambda_{\rm max}\left(\mat{A}\right)$ denote the largest singular value and largest eigenvalue of a matrix $\mat{A}$, respectively.} For a set $\set{S}$, its cardinality is denoted as $\abs{\set{S}}$. \rev{The soft-thresholding operator is defined as 
	\begin{equation*}
	\opof{S}{a, \Delta} = 
	\begin{cases}
	a-\Delta, & a>\Delta \\
	a+\Delta, & a<-\Delta \\
	0, & -\Delta \leq a \leq \Delta,
	\end{cases}
	\end{equation*}
	where $\Delta$ is a non-negative constant.}

\section{Sparse CCA model}
\label{sec:model}

Assume that $\mat{X}$ and $\mat{Y}$ are column-centered to zero mean. SCCA aims to find a linear combination of variables in $\mat{X}$ and $\mat{Y}$ to maximize their correlation~\cite{suo2017sparse,chu2013sparse}:
\begin{equation}\label{equ:SCCA_model}
\begin{aligned}
& \underset{\mat{u},\mat{v}}{\text{maximize}} & & \mat{u}^\trans \mat{X}^\trans \mat{Y} \mat{v} \\
& \text{subject to}                           & & \mat{u}^\trans \mat{X}^\trans \mat{X} \mat{u} \leq 1, \normof{\mat{u}}{1} \leq c_1 \\
&                                             & & \mat{v}^\trans \mat{Y}^\trans \mat{Y} \mat{v} \leq 1, \normof{\mat{v}}{1} \leq c_2,
\end{aligned}
\end{equation}
where $\mat{X} \mat{u}$ and $\mat{Y} \mat{v}$ are the canonical variables, $\mat{u}$ and $\mat{v}$ are canonical loadings/weights measuring the contribution of each feature in the identified association, and $c_1 >0, c_2 >0$ are the regularization parameters that control the sparsity of the solution.

\lsc{The problem \eqref{equ:SCCA_model} is not convenient to solve due to the quadratic constraints.} To save the computational cost, it is a common practice to treat the covariance matrices $\mat{X}^\trans \mat{X}$ and $\mat{Y}^\trans \mat{Y}$ as diagonal~\cite{Parkhomenko2009,witten2009a,witten2009b,chen2012_aistat,chen2013structure,chi2013isbi,Fang2016}. 
This yields the following simplified formulation of SCCA:
\begin{equation}\label{equ:simplified_SCCA_model}
\begin{aligned}
& \underset{\mat{u},\mat{v}}{\text{maximize}} & & \mat{u}^\trans \mat{X}^\trans \mat{Y} \mat{v} \\
& \text{subject to}                           & & \twonorm{\mat{u}}^2 \leq 1, \normof{\mat{u}}{1} \leq c_1 \\
&                                             & & \twonorm{\mat{v}}^2 \leq 1, \normof{\mat{v}}{1} \leq c_2,
\end{aligned}
\end{equation}
where $c_1, c_2 >0$. 

In Section \ref{sec:algorithm} and Supplementary Materials Section \ref{sec_in_supp:multiple_components}, we will describe algorithms to fit the two models, as well as explain how to obtain multiple canonical components.


\section{Grouping effect analysis}
\label{sec:properties_optimal_solution_SCCA_model}

\lsc{In high-dimensional problems such as imaging genomics, grouped variables are common and how to properly select them is an important research problem \cite{shen2020pieee,mikewest2001predicting,zou2005regularization,thompson2010curropinneurol}. For a sparse CCA model, we say it exhibits the grouping effect if it jointly selects or deselects each group of highly correlated variables together.}

To gain initial insights, we start with the simplest case with all $p$ $X$ variables fully correlated with each other.

\begin{lem}\label{lema:grouping_effect_analysis_SCCA_models_particular_case}
Let $\mat{x}_1=\mat{x}_2=\dots=\mat{x}_p$ have unit L2 norm.

The optimal solution $\mat{u}^*$ to problem \eqref{equ:SCCA_model} is
\begin{description}
  \item[(i)] any point on the segment of the line $u_1+u_2+\dots+u_p = 1$ that is inside the L1 ball:
\begin{align}\label{equ:standard_SCCA_model_particular_case_solution_u}
\rev{u_1^* + u_2^* + \dots + u_p^*} &= 1 \\
\normof{\mat{u}^*}{1} &\leq c_1
\end{align}
when $c_1 \geq 1$, and

  \item[(ii)] any $u_1^* \geq 0, u_2^* \geq 0, \dots, u_p^* \geq 0$ that satisfy:
\begin{align}\label{equ:standard_SCCA_model_particular_case_solution_u_Case2}
\rev{u_1^* + u_2^* + \dots + u_p^*} = c_1
\end{align}
when $0 < c_1 < 1$.
\end{description}

The optimal solution $\mat{u}^*$ to problem \eqref{equ:simplified_SCCA_model} is:
\begin{description}
  \item[(i)] $\rev{u_1^* = u_2^* = \dots = u_p^*} = \frac{1}{\sqrt{p}}$ when $c_1 \geq \sqrt{p}$, and
  \item[(ii)] any $u_1^* \geq 0, u_2^* \geq 0, \dots, u_p^* \geq 0$ that satisfy:
\begin{align}\label{equ:simplified_SCCA_model_particular_case_solution_u_Case2}
\rev{u_1^* + u_2^* + \dots + u_p^*} &= c_1 \\
\rev{{u_1^*}^2} + {u_2^*}^2 + \dots + {u_p^*}^2 &\leq 1 
\end{align}
when $1 \leq c_1 < \sqrt{p}$.
\end{description}

\end{lem}

\begin{proof}
We first prove the result for problem \eqref{equ:SCCA_model}, i.e., the SCCA model.

When $\mat{x}_1=\mat{x}_2=\dots=\mat{x}_p \triangleq \mat{x}$, the problem \eqref{equ:SCCA_model} reduces to
\begin{equation}\label{equ:standard_SCCA_model_particular_case}
\begin{aligned}
& \underset{\mat{u},\mat{v}}{\text{maximize}} & & (u_1+u_2+\dots+u_p) \mat{x}^\trans \mat{Y} \mat{v} \\
& \text{subject to} & & \abs{u_1+u_2+\dots+u_p} \leq 1, \normof{\mat{u}}{1} \leq c_1 \\
&                   & & \mat{v}^\trans \mat{Y}^\trans \mat{Y} \mat{v} \leq 1, \normof{\mat{v}}{1} \leq c_2,
\end{aligned}
\end{equation}
where $c_1 \geq 1$, $c_2 \geq 1$.

Note that the optimal solution to problem \eqref{equ:standard_SCCA_model_particular_case} is not unique because the objective function remains the same after we reverse the signs of both $\mat{u}$ and $\mat{v}$. To resolve this, we assume $u_1+u_2+\dots+u_p \geq 0$.

Note also that the optimal value of problem \eqref{equ:standard_SCCA_model_particular_case} is larger than zero when $c_1>0, c_2>0$.

As a result, $\mat{u}$ and $\mat{v}$ can be independently optimized:
\begin{align}
&
\begin{aligned}
\mat{u}^* &= \underset{\mat{u}}{\argmax} & & (u_1+u_2+\dots+u_p) \\
                      & \text{subject to} & & \abs{u_1+u_2+\dots+u_p} \leq 1, \normof{\mat{u}}{1} \leq c_1 \\
\end{aligned}\label{equ:standard_SCCA_model_particular_case_u} \\[6pt]
&
\begin{aligned}
\mat{v}^* & = \underset{\mat{v}}{\argmax} & & \mat{x}^\trans \mat{Y} \mat{v} \\
& \text{subject to} & & \mat{v}^\trans \mat{Y}^\trans \mat{Y} \mat{v} \leq 1, \normof{\mat{v}}{1} \leq c_2.
\end{aligned}\label{equ:standard_SCCA_model_particular_case_v}
\end{align}

Solving \eqref{equ:standard_SCCA_model_particular_case_u} yields the optimal solution $\mat{u}^*$ shown in \eqref{equ:standard_SCCA_model_particular_case_solution_u}-\eqref{equ:standard_SCCA_model_particular_case_solution_u_Case2}.

We next prove the result regarding problem \eqref{equ:simplified_SCCA_model}, i.e., the simplified SCCA model.

When $\mat{x}_1=\mat{x}_2=\dots=\mat{x}_p \triangleq \mat{x}$, the problem \eqref{equ:simplified_SCCA_model} reduces to
\begin{equation}\label{equ:simplified_SCCA_model_particular_case}
\begin{aligned}
& \underset{\mat{u},\mat{v}}{\text{maximize}} & & (u_1+u_2+\dots+u_p) \mat{x}^\trans \mat{Y} \mat{v} \\
& \text{subject to} & & \twonorm{\mat{u}}^2 \leq 1, \normof{\mat{u}}{1} \leq c_1 \\
&                   & & \twonorm{\mat{v}}^2 \leq 1, \normof{\mat{v}}{1} \leq c_2,
\end{aligned}
\end{equation}
where $c_1 \geq 1$, $c_2 \geq 1$.

To resolve sign ambiguity, we assume $u_1+u_2+\dots+u_p \geq 0$. Therefore, $\mat{u}$ and $\mat{v}$ can be independently optimized:
\begin{equation}\label{equ:simplified_SCCA_model_particular_case_u}
\begin{aligned}
\mat{u}^* & = \underset{\mat{u}}{\argmax} & & (u_1+u_2+\dots+u_p) \\
& \text{subject to} & & \twonorm{\mat{u}}^2 \leq 1, \normof{\mat{u}}{1} \leq c_1 \\
\end{aligned}
\end{equation}
\vspace{6pt}
\begin{equation}\label{equ:simplified_SCCA_model_particular_case_v}
\begin{aligned}
\mat{v}^* & = \underset{\mat{v}}{\argmax} & & \mat{x}^\trans \mat{Y} \mat{v} \\
& \text{subject to} & & \twonorm{\mat{v}}^2 \leq 1, \normof{\mat{v}}{1} \leq c_2.
\end{aligned}
\end{equation}

Solving \eqref{equ:simplified_SCCA_model_particular_case_u} yields the optimal solution shown in Lemma \ref{lema:grouping_effect_analysis_SCCA_models_particular_case} (simplified SCCA part).
\end{proof}

\lsc{We then provide a formal proof of the grouping effects in variable selection for the simplified SCCA.}
%

\begin{thm}\label{thm:grouping_effect}
Given data $\left(\mat{X},\mat{Y}\right)$, with columns standardized to zero mean and unit norm, and regularization parameters $\left(c_1,c_2\right)$. Let $\left(\mat{u}^*,\mat{v}^*\right)$ be an optimal solution to problem \eqref{equ:simplified_SCCA_model}. 
Assume at $\left(\mat{u}^*,\mat{v}^*\right)$ the L2 inequality constraint on $\mat{u}$ is strongly active. We have:
\begin{itemize}
  \item when $u_i^* u_j^*>0$
  \begin{align}\label{equ:ui-uj_bound}
         &\abs{u_i^* - u_j^*} \\
    \leq\; &\frac{1}{\alpha_1} \minof{\rev{\sigma_{\rm max}}\left(\mat{Y}\right), c_2 \sqrt{\sum_{\ell=1}^n \max_{1 \leq j \leq q} y_{\ell j}^2}} \sqrt{\left(1-r_{ij}\right)/2} \nonumber
  \end{align}
    \item when $u_i^* u_j^*<0$
\begin{align}\label{equ:ui+uj_bound}
         &\abs{u_i^* + u_j^*} \\
    \leq\; &\frac{1}{\alpha_1} \minof{\rev{\sigma_{\rm max}}\left(\mat{Y}\right), c_2 \sqrt{\sum_{\ell=1}^n \max_{1 \leq j \leq q} y_{\ell j}^2}} \sqrt{\left(1+r_{ij}\right)/2}, \nonumber
  \end{align}
\end{itemize}
where $r_{ij} = \mat{x}_i^\trans \mat{x}_j \in [-1,1]$ is the Pearson correlation coefficient between $\mat{x}_i$ and $\mat{x}_j$, and $\alpha_1>0$ is a constant that \lsc{only} depends on $\left(\mat{X},\mat{Y}, c_1,c_2\right)$.

Likewise, if at $\left(\mat{u}^*,\mat{v}^*\right)$ the L2 inequality constraint on $\mat{v}$ is strongly active, we have
\begin{itemize}
  \item when $v_i^* v_j^*>0$
  \begin{align}\label{equ:vi-vj_bound}
    &\abs{v_i^* - v_j^*} \\
    \leq\; &\frac{1}{\alpha_2} \minof{\rev{\sigma_{\rm max}}\left(\mat{X}\right), c_1 \sqrt{\sum_{\ell=1}^n \max_{1 \leq i \leq p} x_{\ell i}^2}} \sqrt{(1-r_{ij}')/2} \nonumber
  \end{align}
    \item when $v_i^* v_j^*<0$
\begin{align}\label{equ:vi+vj_bound}
    &\abs{v_i^* + v_j^*} \\
    \leq\; &\frac{1}{\alpha_2} \minof{\rev{\sigma_{\rm max}}\left(\mat{X}\right), c_1 \sqrt{\sum_{\ell=1}^n \max_{1 \leq i \leq p} x_{\ell i}^2}} \sqrt{\left(1+r_{ij}'\right)/2}, \nonumber
  \end{align}
\end{itemize}
where $r_{ij}' = \mat{y}_i^\trans \mat{y}_j \in [-1,1]$ is the Pearson correlation coefficient between $\mat{y}_i$ and $\mat{y}_j$, and $\alpha_2>0$ is a constant that \lsc{only} depends on $\left(\mat{X},\mat{Y}, c_1,c_2\right)$.
\end{thm}
\begin{proof}
Since each subproblem (solve for $\mat{u}$ with $\mat{v}$ fixed or solve for $\mat{v}$ with $\mat{u}$ fixed) is a convex optimization problem with differentiable objective and constraint functions (The L1 inequality constraint can be written as $2^p$ linear inequality constraints), and is strictly feasible (Slater's condition holds), the KKT conditions provide necessary and sufficient conditions for optimality \cite{Boyd2004}.

The KKT conditions for the optimality of $\mat{u}^*$ consist of the following conditions:
\begin{equation}\label{equ:simplified_SCCA_model_KKT_conditions_stationarity_u}
    2\alpha_1 \mat{u}^* + \lambda_1 \mat{s} = \mat{X}^\trans \mat{Y} \mat{v}^*,
\end{equation}
where $s_i=\signof{u_i^*}$ if $u_i^* \neq 0$; otherwise, $s_i \in [-1,1]$.
\begin{align}
    \alpha_1 \geq 0, \quad \twonorm{\mat{u}^*}^2 \leq 1, \quad \alpha_1 \left(\twonorm{\mat{u}^*}^2 - 1\right) = 0
    \label{equ:simplified_SCCA_model_KKT_conditions_L2_constraint_u} \\
    \lambda_1 \geq 0, \quad \normof{\mat{u}^*}{1} \leq c_1, \quad \lambda_1 \left(\normof{\mat{u}^*}{1} - c_1\right) = 0.
    \label{equ:simplified_SCCA_model_KKT_conditions_L1_constraint_u}
\end{align}

If $u_i^* u_j^*>0$, then both $u_i^*$ and $u_j^*$ are non-zero with $\signof{u_i^*}=\signof{u_j^*}$. From \eqref{equ:simplified_SCCA_model_KKT_conditions_stationarity_u}, it follows that
\begin{align}
    2\alpha_1 u_i^* + \lambda_1 \signof{u_i^*} &= \mat{x}_i^\trans \mat{Y} \mat{v}^* \label{equ:simplified_SCCA_model_KKT_conditions_stationarity_i}\\
    2\alpha_1 u_j^* + \lambda_1 \signof{u_j^*} &= \mat{x}_j^\trans \mat{Y} \mat{v}^*. \label{equ:simplified_SCCA_model_KKT_conditions_stationarity_j}
\end{align}

Subtracting \eqref{equ:simplified_SCCA_model_KKT_conditions_stationarity_j} from \eqref{equ:simplified_SCCA_model_KKT_conditions_stationarity_i} gives
\begin{equation}
    2\alpha_1 \left(u_i^* - u_j^*\right) = \left(\mat{x}_i-\mat{x}_j\right)^\trans \mat{Y} \mat{v}^*.
\end{equation}

Therefore, we have
\begin{equation}\label{equ:ui-uj_bound1}
\begin{aligned}
\abs{u_i^* - u_j^*} &= \frac{1}{2\alpha_1} \abs{\left(\mat{x}_i-\mat{x}_j\right)^\trans \mat{Y} \mat{v}^*} \\
                 &\leq \frac{1}{2\alpha_1} \twonorm{\mat{x}_i-\mat{x}_j} \twonorm{\mat{Y} \mat{v}^*}. 
\end{aligned}
\end{equation}

Since $\mat{X}$ is column standardized, we have 
\begin{equation}\label{equ:samplePearson_correlation_coefficient}
    \twonorm{\mat{x}_i-\mat{x}_j} = \sqrt{\twonorm{\mat{x}_i}^2 + \twonorm{\mat{x}_j}^2 - 2\mat{x}_i^\trans \mat{x}_j} = \sqrt{2\left(1-r_{ij}\right)},
\end{equation}
where $r_{ij} = \mat{x}_i^\trans \mat{x}_j$ is the sample Pearson correlation coefficient between $\mat{x}_i$ and $\mat{x}_j$.

In the domain of problem \eqref{equ:simplified_SCCA_model}, it holds that 
	\begin{equation}\label{equ:ub11}
	\twonorm{\mat{Y} \mat{v}^*} \leq \sigma_{\rm max}\left(\mat{Y}\right) \twonorm{\mat{v}^*} \leq \sigma_{\rm max}\left(\mat{Y}\right)
	\end{equation}
	and
\begin{equation}\label{equ:ub12}
\begin{aligned}
	  \twonorm{\mat{Y} \mat{v}^*} &= \sqrt{\sum_{\ell=1}^n \left(\sum_{j=1}^q y_{\ell j} v_j^*\right)^2} \\
\leq &\sqrt{\sum_{\ell=1}^n \max_{1 \leq j \leq q} y_{\ell j}^2 \left(\sum_{j=1}^q \abs{v_j^*}\right)^2} \\
   = &\sqrt{\sum_{\ell=1}^n \max_{1 \leq j \leq q} y_{\ell j}^2} \normof{\mat{v}^*}{1} \\
\leq &c_2 \sqrt{\sum_{\ell=1}^n \max_{1 \leq j \leq q} y_{\ell j}^2},
\end{aligned}
\end{equation}
where in \eqref{equ:ub11} and \eqref{equ:ub12} we have used the L2 and L1 constraints in problem \eqref{equ:simplified_SCCA_model}, respectively.

Substituting \eqref{equ:samplePearson_correlation_coefficient}-\eqref{equ:ub12} into \eqref{equ:ui-uj_bound1}, we arrive at
\begin{align}
       &\abs{u_i^* - u_j^*} \nonumber \\
\leq\; &\frac{1}{\alpha_1} \minof{\sigma_{\rm max}\left(\mat{Y}\right), c_2 \sqrt{\sum_{\ell=1}^n \max_{1 \leq j \leq q} y_{\ell j}^2}} \sqrt{\left(1-r_{ij}\right)/2}.
\end{align}

Since the L2 inequality constraint on $\mat{u}$ is strongly active at $\left(\mat{u}^*,\mat{v}^*\right)$, we have $\alpha_1 > 0$. Specifically, combining conditions \eqref{equ:simplified_SCCA_model_KKT_conditions_stationarity_u}-\eqref{equ:simplified_SCCA_model_KKT_conditions_L1_constraint_u} yields 
\begin{equation}
    \alpha_1 = \frac{1}{2}\twonorm{\rev{\opof{S}{\mat{X}^\trans \mat{Y} \mat{v}^*,\lambda_1}}},
\end{equation}
where $\lambda_1 = 0$ if this results in $\frac{\normof{\mat{X}^\trans \mat{Y} \mat{v}^*}{1}}{\twonorm{\mat{X}^\trans \mat{Y} \mat{v}^*}} \leq c_1$; otherwise, $\lambda_1$ is the smallest positive number for which it satisfies $\frac{\normof{\rev{\opof{S}{\mat{X}^\trans \mat{Y} \mat{v}^*,\lambda_1}}}{1}}{\twonorm{\rev{\opof{S}{\mat{X}^\trans \mat{Y} \mat{v}^*,\lambda_1}}}} = c_1$. Thus we obtain \eqref{equ:ui-uj_bound}.

Using a similar line of argumentation, we can prove \eqref{equ:ui+uj_bound} and \eqref{equ:vi-vj_bound}-\eqref{equ:vi+vj_bound}.
\end{proof}

\begin{figure*}[tb]
\begin{center}
\subfigure[]{\label{fig:a}\includegraphics[width=.25\textwidth]{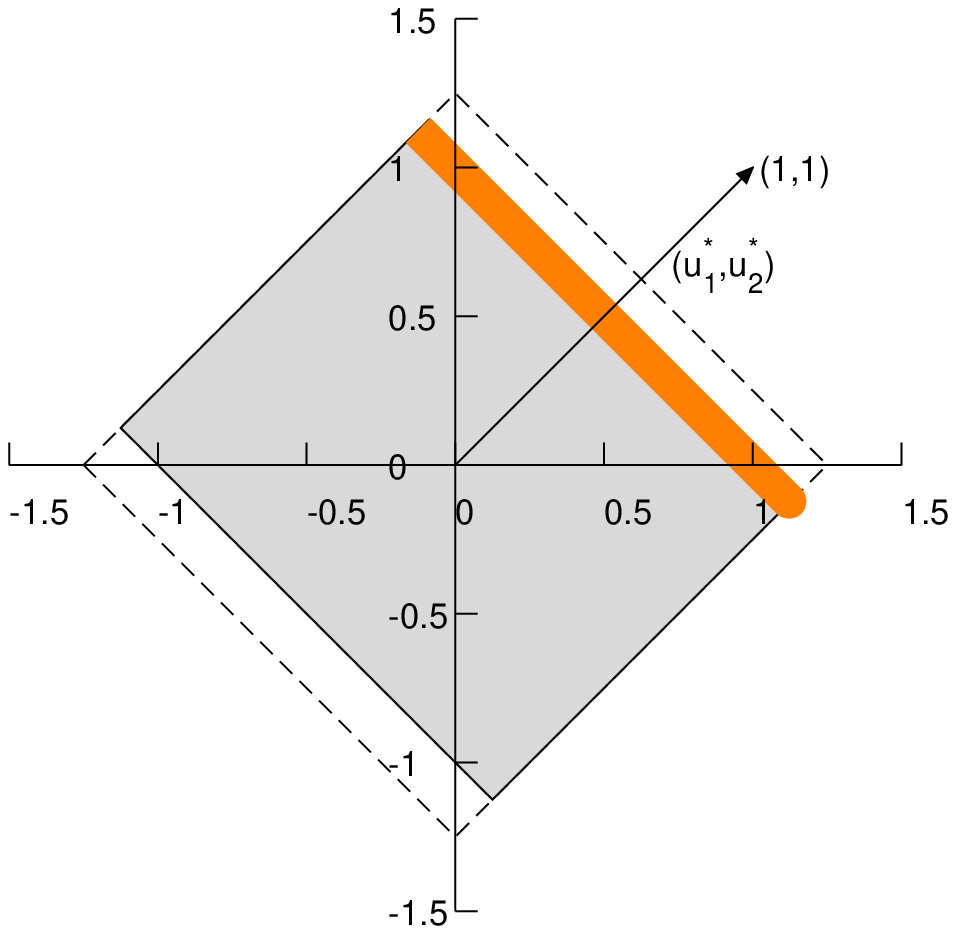}}
\subfigure[]{\label{fig:b}\includegraphics[width=.25\textwidth]{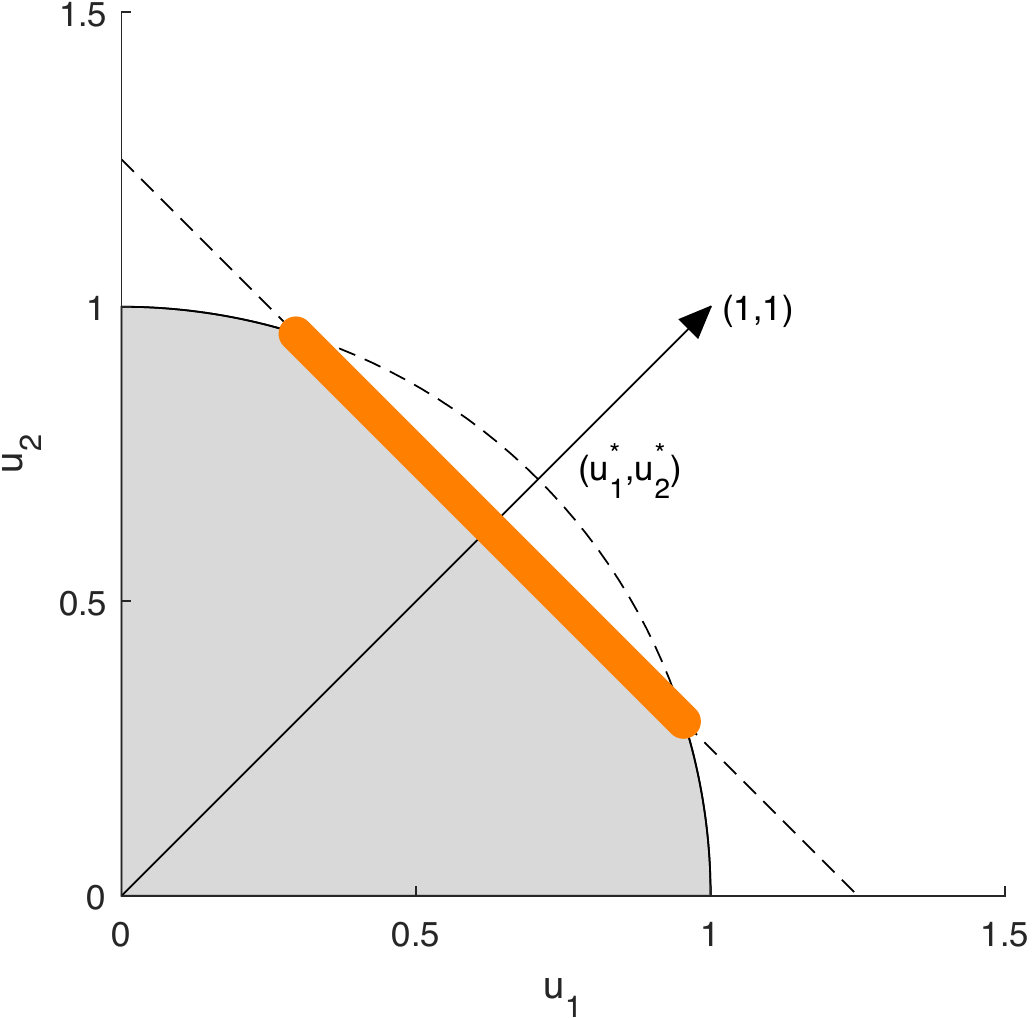}}
\subfigure[]{\label{fig:c}\includegraphics[width=.25\textwidth]{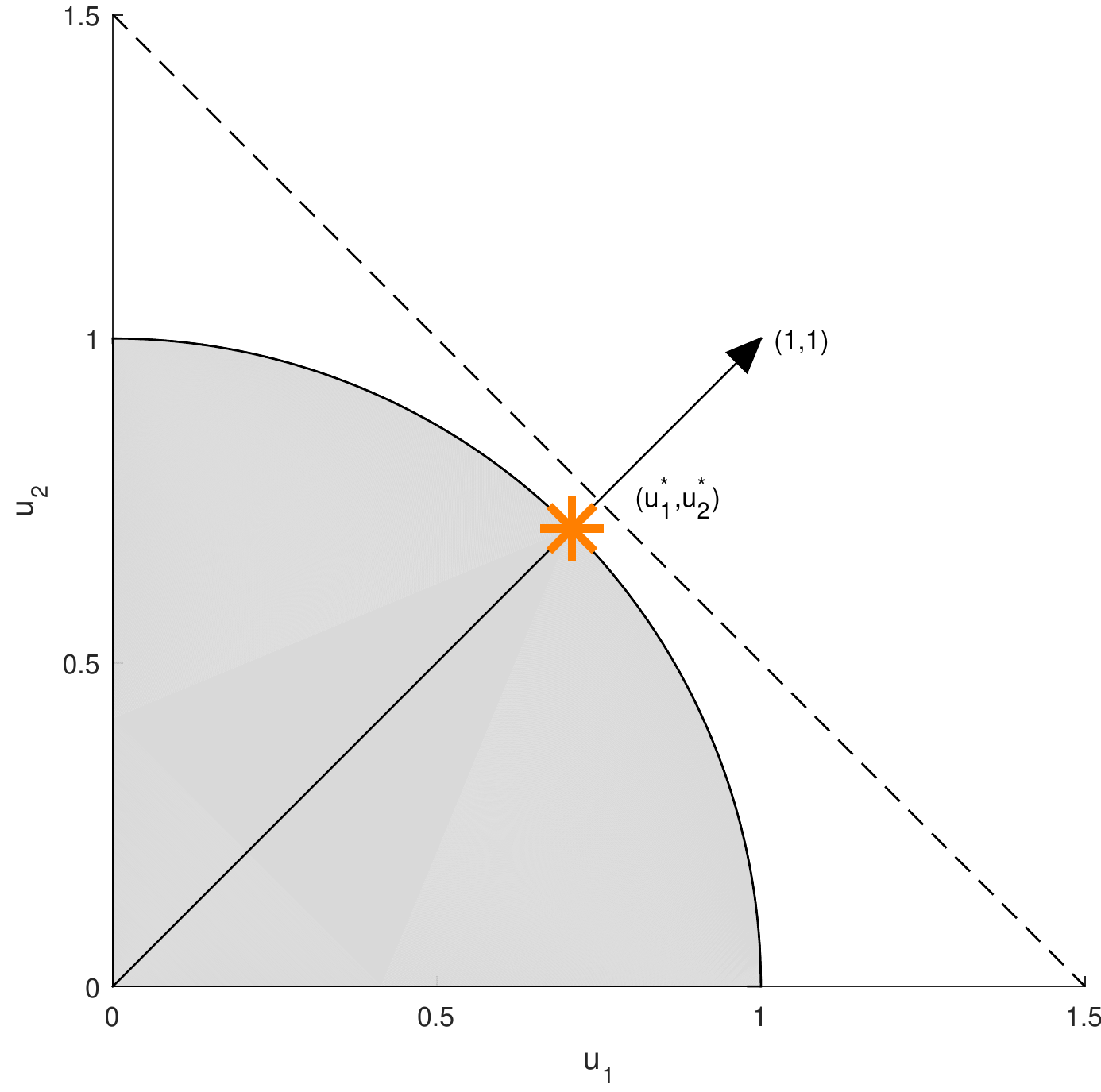}}
\caption{The optimal solution set $\mat{u}^*$ with $p=2$ identical variables. (a) the SCCA problem with $c_1 = 1.25$; (b) simplified SCCA with $c_1 = 1.25$; (c) simplified SCCA with $c_1 = 1.5$. The feasible set of points are shown lightly shaded. The optimal points are highlighted in orange.}
\label{fig:SCCA_illustration_solution_u_p2}
\end{center}
\end{figure*}

Fig.~\ref{fig:SCCA_illustration_solution_u_p2} illustrates the optimal solution $\mat{u}^*$ to problems \eqref{equ:SCCA_model} and \eqref{equ:simplified_SCCA_model} with $p=2$ identical $X$ variables. We see that for SCCA (Fig. \ref{fig:a}), the optimal solution set is a line segment that cross the axes (i.e., includes sparse solutions). While for simplified SCCA (Figs. \ref{fig:b}-\ref{fig:c}), the optimal solution set does not intersect with the axes (i.e., does not include sparse or nearly sparse solutions); in particular, when the L2 constraint on $\mat{u}$ is strongly active at the optimal solution, i.e., when $c_1 \geq \sqrt{2}$, the optimal solution set contains a single point with equal coordinates: $(u_1^*,u_2^*)=\left(\frac{\sqrt{2}}{2},\frac{\sqrt{2}}{2}\right)$.

\section{Optimization Algorithms}
\label{sec:algorithm}
Both problems \eqref{equ:SCCA_model} and \eqref{equ:simplified_SCCA_model} are bi-convex, i.e., convex in $\mathbf{u}$ with $\mathbf{v}$ fixed and in $\mathbf{v}$ with $\mathbf{u}$ fixed, but not jointly convex in $\mathbf{u}$ and $\mathbf{v}$. A standard method to solve the SCCA models is alternating optimization~\cite{Bezdek2002}: it first updates $\mat{u}$ while holding $\mat{v}$ fixed and then updates $\mat{v}$ while holding $\mat{u}$ fixed, and repeats this process until convergence.

\subsection{SCCA model \eqref{equ:SCCA_model}}
\label{subsec:algorithm_SCCA_model}

The SCCA model fitting algorithm is shown in Algorithm \ref{alg:SCCA}.

\begin{algorithm}
	\caption{SCCA algorithm}
	\label{alg:SCCA}
	\begin{algorithmic}[1]
		\STATE Initialize $\mat{v}$;
		\REPEAT
		
		\STATE Update $\mat{u}$ with $\mat{v}$ fixed:
\begin{equation}\label{equ:SCCA_model_solve_for_u}
\begin{aligned}
& \underset{\mat{u}}{\text{maximize}} & & \mat{u}^\trans \mat{X}^\trans \mat{Y} \mat{v} \\
& \text{subject to}                           & & \twonorm{\mat{X} \mat{u}}^2 \leq 1, \normof{\mat{u}}{1} \leq c_1
\end{aligned}
\end{equation}
		
		\STATE Update $\mat{v}$ with $\mat{u}$ fixed:;	
		\begin{equation}\label{equ:SCCA_model_solve_for_v}
\begin{aligned}
& \underset{\mat{v}}{\text{maximize}} & & \mat{u}^\trans \mat{X}^\trans \mat{Y} \mat{v} \\
& \text{subject to}                           & & \mat{v}^\trans \mat{Y}^\trans \mat{Y} \mat{v} \leq 1, \normof{\mat{v}}{1} \leq c_2
\end{aligned}
\end{equation}
		\UNTIL{convergence.}
	\end{algorithmic}
\end{algorithm}

Both problems \eqref{equ:SCCA_model_solve_for_u} and \eqref{equ:SCCA_model_solve_for_v} are convex optimization problems, and in~\cite{suo2017sparse} the linearized alternating direction method of multipliers (ADMM)~\cite{Boyd2011} algorithm has been proposed to solve each of them.
Since in~\cite{suo2017sparse} it uses a slightly different formulation (\lsc{therein the L1 regularizer appears in the objective function}), we have \lsc{presented a new} linearized ADMM algorithm to solve problem \eqref{equ:SCCA_model_solve_for_u} in Supplementary Materials \ref{sec_in_supp:ADMM}.

\subsection{Simplified SCCA model \eqref{equ:simplified_SCCA_model}}
\label{sec:algorithm}

	
We first introduce the following lemma, which will be used \lsc{as a building block} in the simplified SCCA algorithm.

\begin{lem}\label{lema:solution_QCLP}
	Consider the quadratically constrained linear program (QCLP) optimization problem:
	\begin{equation}\label{equ:QCLP}
	\underset{\mat{u}}{\text{maximize}} \; \mat{a}^\trans \mat{u}  \quad  \text{subject to } \twonorm{\mat{u}}^2 \leq 1, \normof{\mat{u}}{1} \leq c,
	\end{equation}
	where $c>0$ is a constant.

Define $\set{S} = \left\{i: i \in \argmax_j \abs{a_j} \right\}$. The optimal solution $\mat{u}^*$ to \eqref{equ:QCLP} is as below.
\begin{itemize}
  \item Case 1~\footnote{In Case 1, the solution is generally not unique. Specifically, the optimal solution has the following form:
\begin{equation*}
	[\mat{u}^*]_i =
	\begin{cases}
	w_i \signof{a_i} , & i \in \set{S} \\
	0, & i \notin \set{S}
	\end{cases}
\end{equation*}
where $w_i$, $i \in \set{S}$, can be any non-negative numbers that satisfy
$\sum_{i \in \set{S}} w_i^2 \leq 1, \quad \sum_{i \in \set{S}} w_i = c$.
The presented solution is the solution that minimizes $\sum_{i \in \set{S}} w_i^2$.}: $c < \sqrt{\abs{\set{S}}}$
\begin{equation}\label{equ:solution_QCLP_Case1}
	[\mat{u}^*]_i =
	\begin{cases}
	\frac{c}{\abs{\set{S}}} \signof{a_i} , & i \in \set{S} \\
	0, & i \notin \set{S}
	\end{cases}
\end{equation}

  \item Case 2: $c \geq \sqrt{\abs{\set{S}}}$
	\begin{equation}\label{equ:solution_QCLP_Case3}
	\mat{u}^* = \frac{\rev{\opof{S}{\mat{a},\Delta}}}{\twonorm{\rev{\opof{S}{\mat{a},\Delta}}}}
	\end{equation}
	where $\Delta=0$ if this results $\normof{\mat{u}^*}{1} \leq c$; otherwise, $\Delta > 0$ satisfies $\normof{\mat{u}^*}{1} = c$. Here the soft-thresholding $\rev{\opof{S}{\mat{a},\Delta}}$ is applied to $\mat{a}$ coordinate-wise.
\end{itemize}
\end{lem}

The above lemma extends Lemma 2.2 of~\cite{witten2009a} from $c \in [\sqrt{\abs{\set{S}}}, \infty)$ to $c \in (0, \infty)$. See Supplementary Materials Section \ref{sec_in_supp:lemma2_2} for the proof of Lemma \ref{lema:solution_QCLP} and how it extends Lemma 2.2 of~\cite{witten2009a}.

For the simplified SCCA in \eqref{equ:simplified_SCCA_model}, each subproblem (solving $\mat{u}$ with $\mat{v}$ fixed or solving $\mat{v}$ with $\mat{u}$ fixed) is a QCLP problem of form \eqref{equ:QCLP}, which results in Algorithm \ref{alg:simplified_SCCA}.

\begin{algorithm}
	\caption{Simplified SCCA algorithm}
	\label{alg:simplified_SCCA}
	\begin{algorithmic}[1]
		\STATE Initialize $\mat{v}$;
		\REPEAT
		\STATE Update $\mat{u}$ according to Lemma \ref{lema:solution_QCLP}, with $\mat{a} = \mat{X}^\trans \mat{Y} \mat{v}$ and $c=c_1$;
		\STATE Update $\mat{v}$  according to Lemma \ref{lema:solution_QCLP}, with $\mat{a} = \mat{Y}^\trans \mat{X} \mat{u}$ and $c=c_2$;
		\UNTIL{convergence.}
	\end{algorithmic}
\end{algorithm}

Note that by repeatedly applying Algorithms \ref{alg:SCCA} and \ref{alg:simplified_SCCA}, we can obtain multiple canonical components, as described in Section \ref{sec_in_supp:multiple_components} in Supplemental Materials.

\begin{figure*}[tb]
  \begin{minipage}[b]{.48\textwidth}
  \centering
    \includegraphics[width=\linewidth]{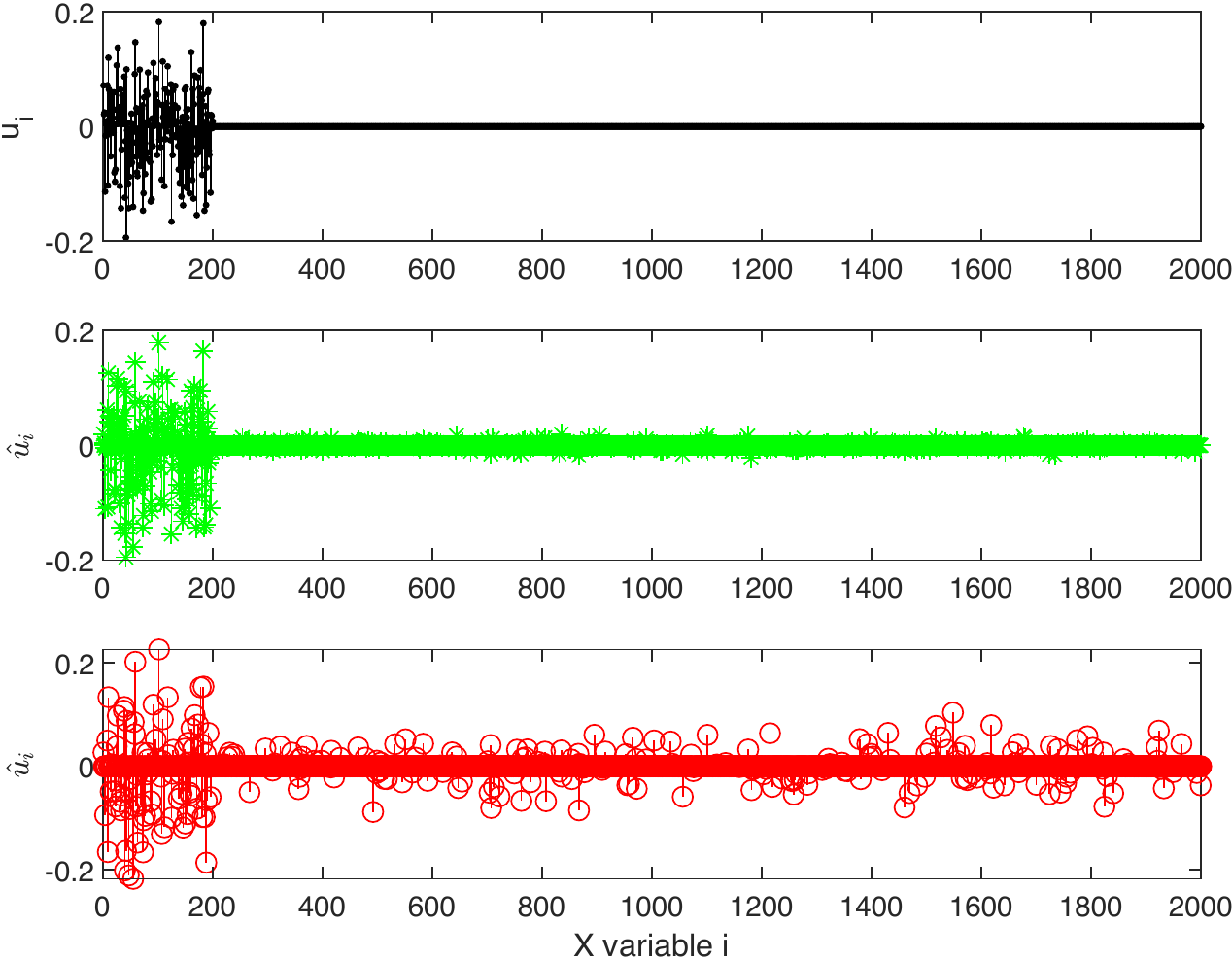}
  \end{minipage}
  \hspace{0.5cm}
  \begin{minipage}[b]{.48\textwidth}
  \centering
    \includegraphics[width=\linewidth]{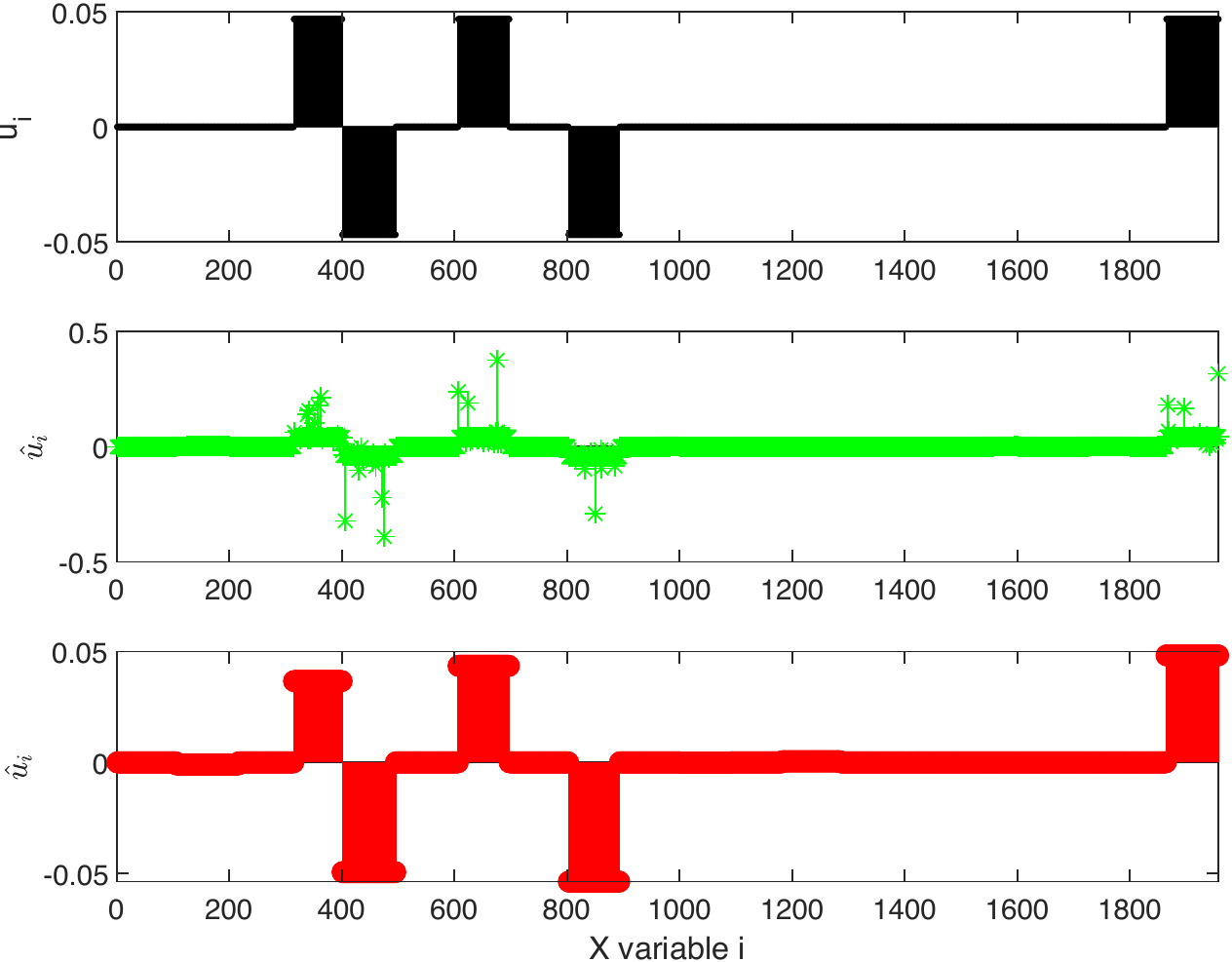}
  \end{minipage} \\
  \begin{minipage}[b]{.48\textwidth}
  \centering
    \includegraphics[width=\linewidth]{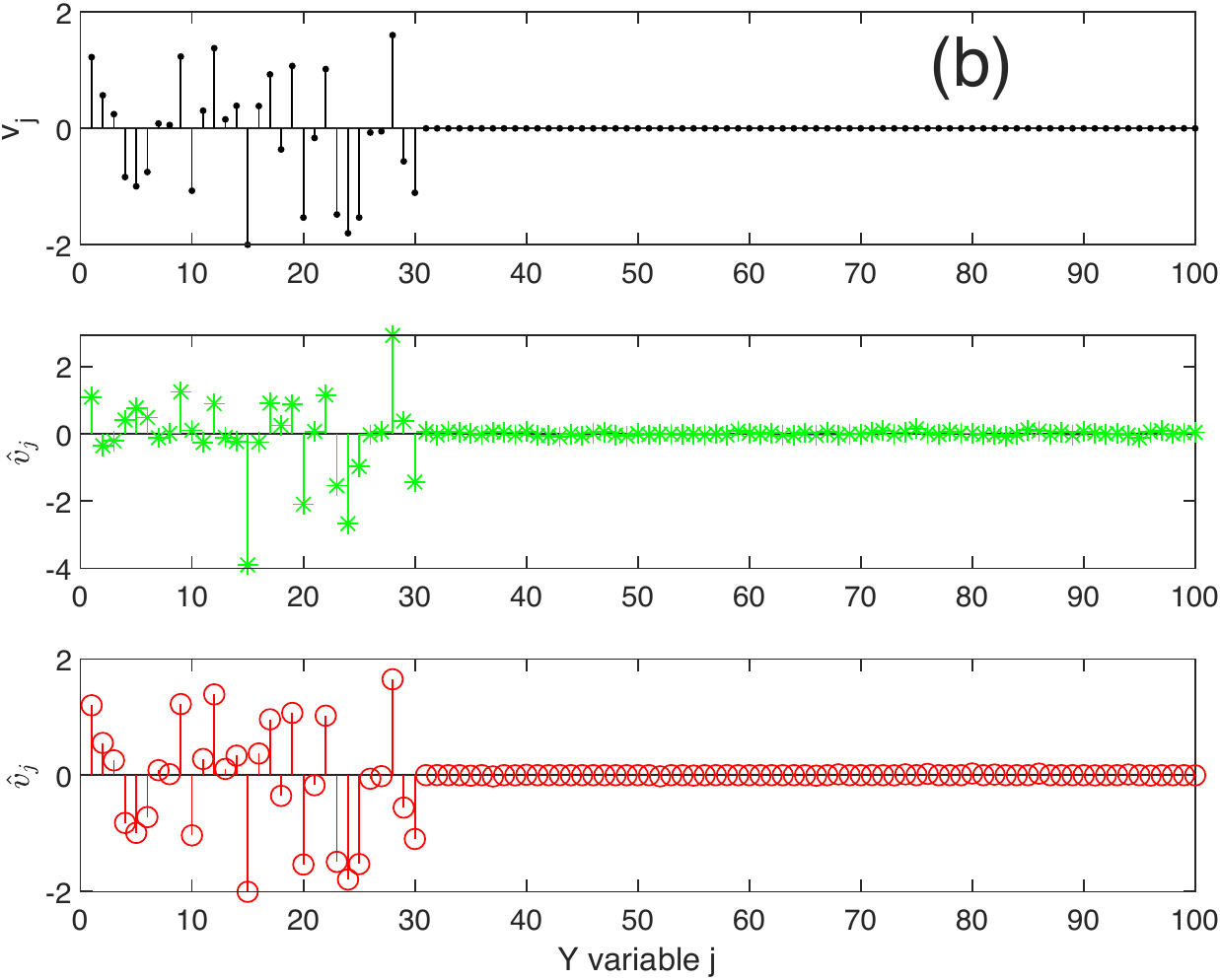}
  \end{minipage}
  \hspace{0.5cm}
  \begin{minipage}[b]{.48\textwidth}
  \centering
    \includegraphics[width=\linewidth]{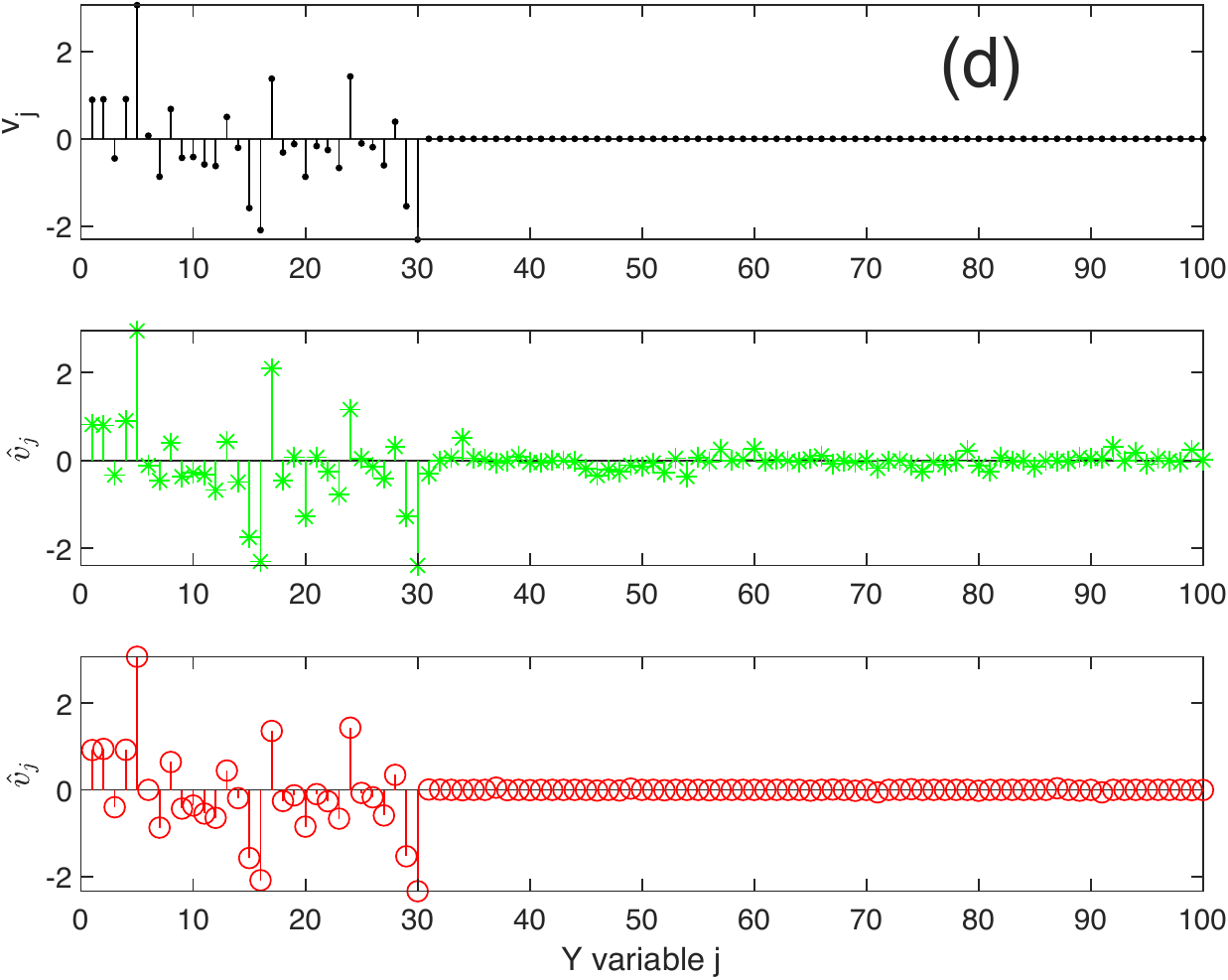}
  \end{minipage}
      \vspace{-0.3in}
      \caption{The actual and estimated canonical weight vectors $\mat{u}$ and  $\mat{v}$ for (a-b) Experiment 1 and (c-d) Experiment 2. In each subfigure, the top row shows the actual weights used in the generative model, and the bottom two rows show the weights estimated by SCCA and simplified SCCA on all training data, respectively. \lsc{To facilitate comparison, the estimated weight vectors are scaled to have the same Euclidean norm as the actual weight vector.}}
    \label{fig:synthetic_data_uv_esti}

\end{figure*}

\vfill

\section{Experimental results and discussion}
\label{sec:experiment}

We perform comparative study of the two SCCA models using both synthetic data and real imaging genetics data.

\subsection{Simulation study on synthetic data}
\label{subsec:simulation}
Assume the data $\mat{X} \in \real^{n \times p}$ and $\mat{Y} \in \real^{n \times q}$ collect $n$ i.i.d. observations/samples of random vectors $\mat{x} \in \real^{p \times 1}$ and $Y \in \real^{q \times 1}$ (with slight abuse of notation), respectively, with $n=1000, p\approx 2000, q=100$. \lsc{We consider two simulation setups, one with uncorrelated variables and the other with grouped variables.} For simplicity, we focus on the simulation and analysis of $X$ variables only.

\begin{table*}[tb]
\begin{center}
\begin{threeparttable}
\caption{Performance comparison on canonical correlation coefficients on synthetic data.}
\label{tab:Training_testing_CC_synthetic_data}
\vskip 0.15in
\begin{small}
\begin{sc}
\begin{tabular}{lccccccr}
\toprule
            &                                            & \multicolumn{4}{c}{Training} & \multicolumn{2}{c}{Testing} \\
                                                          \cmidrule(lr){3-6}\cmidrule(lr){7-8}
  Model     & $\left(c_1^{\rm opt},c_2^{\rm opt}\right)$ & Cov@Val$^*$ & Corr@Val$^*$ & Cov$^{**}$ & Corr$^{**}$ & Cov & Corr \\
\midrule
          \multicolumn{8}{c}{Experimental setup 1} \\
SCCA		 &  (10.379, 1.218)	 &  \NA		 &    0.843	 &    4.750	 &    0.995	 &    4.395	 &    0.962	 \\
Simp SCCA	 &  (11.763, 4.513)	 &    3.304	 &  \NA		 &    7.396	 &    0.893	 &    4.104	 &    0.749	 \\
          \multicolumn{8}{c}{Experimental setup 2} \\
SCCA		 &  (2.516, 0.145)	 &  \NA		 &    0.986	 &  384.013	 &    0.991	 &  406.562	 &    0.990	 \\
Simp SCCA	 &  (21.354, 4.154)	 &  461.295	 &  \NA		 &  522.153	 &    0.983	 &  545.016	 &    0.985	 \\
\bottomrule
\end{tabular}
\end{sc}
\end{small}
\begin{tablenotes}
      \footnotesize
\item[*] Cov@Val/Corr@Val: canonical covariance/correlation on the validation data during the training (model selection) stage. The reported value is the maximum canonical covariance/correlation over all candidate $\left(c_1,c_2\right)$ (i.e. at the optimal regularization parameters $\left(c_1^{\rm opt},c_2^{\rm opt}\right)$).
\item[**] Cov/Corr: canonical covariance/correlation when the optimal model is fit to combined training and validation data.
\end{tablenotes}
\vskip -0.1in
\end{threeparttable}
\end{center}
\end{table*}

\subsubsection{Setup 1: uncorrelated variables}
\label{subsubsec:uncorrelated_variables}
The random vector $\mat{x}$ is modeled as standard normal: $\mat{x} \sim \mathcal{N}(0,\, \mat{I}_p)$. Define $z$ as $z = \mat{c}^\trans \mat{x}$, where $\mat{c} \in \real^{p \times 1}$ is a sparse vector.
The \rev{random vector $\mat{y}$} is the modeled as
\begin{equation}
\mat{y} = \mat{d} z + \sigma \mat{n} \label{equ:uncorrelated_variables_random_vector_Y}
\end{equation}
where $\mat{d} \in \real^{q \times 1}$ is a sparse vector, $\mat{n} \sim \mathcal{N}(0,\, \mat{I}_q)$ models random noise. We set $\sigma^2$ to have signal-to-noise ratio of 1.

\subsubsection{Setup 2: grouped variables}
\label{subsubsec:grouped_variables}
We assume that \rev{the variables in $\mat{x}$} form $G=20$ non-overlapping groups:
\vskip 1pt
\begin{equation*}\label{equ:X}
\mat{x} =
\begin{bmatrix}
\coolover{p_1}{x_1 & \cdots & x_1} & \coolover{p_2}{x_2 & \cdots & x_2} & \cdots & \coolover{p_G}{x_G & \cdots & x_G}
\end{bmatrix}^\trans
\end{equation*}
The group sizes $p_g$ are drawn independently from a Poisson distribution with mean 100. The total number of variables in $\mat{x}$ is $p=\sum_{g=1}^G p_g $. For $g=1,2,\dots,G$, $x_g \sim \mathcal{N}(0,\,1)$.

Define $\mat{c} \in \real^{p \times 1}$ as a sparse vector collecting the weights of variables in $\mat{x}$. We assume that the elements of $\mat{c}$ are grouped in the same way as $\mat{x}$.
Five of $G=20$ groups of variables in $\mat{x}$ are randomly selected and their weights are set to 1 (\lsc{alternate in sign group-wise for visual clarity}), while the remaining groups of variables in $\mat{x}$ are not correlated/informative and their weights are set to 0.
The $\mat{c}$ is shown in the top row of \rev{Fig.~\ref{fig:synthetic_data_uv_esti}(c)}.
Define a random variable $z$ as $z = \mat{c}^\trans \mat{X}$. The \rev{random vector $\mat{y}$} is modeled in the same way as described in Section \ref{subsubsec:uncorrelated_variables}.

\subsubsection{Hyperparameter tuning \& performance estimation}
\label{subsubsec:Tuning_parameter_selection+generalization_performance_estimation}

To tune the hyperparameters $\left(c_1,c_2\right)$, we partition the data into training (50\%), validation (25\%), and testing (25\%) sets. After fitting the SCCA model on the training data, the canonical correlation on the validation data is estimated over a two-dimensional grid in log-linear scale: $2.\caret{\left(\floor{\log_2 c_{1,{\rm min}}}:\ceil{\log_2 c_{1,{\rm max}}}\right)} \times 2.\caret{\left(\floor{\log_2 c_{2,{\rm min}}}:\ceil{\log_2 c_{2,{\rm max}}}\right)}$, where $c_{\ell,{\rm min}}$ and $c_{\ell,{\rm max}}$, $\ell=1,2$, are the minimum and maximum value of $c_\ell$, respectively. The $c_1$ and $c_2$ yielding the maximum validation canonical correlation is selected. Then, we train the model with the selected regularization parameters on the full training data (training+validation) and report the canonical correlation on the testing set as the performance. For the simplified SCCA, the same procedure is used except that the canonical covariance is used as the metric for hyperparameter tuning. More detailed description of the procedure to select $c_1, c_2$ and to assess performance, including how to determine $c_{\ell,{\rm min}}$ and $c_{\ell,{\rm max}}$, $\ell=1,2$, is provided in Supplementary Materials \ref{subsec_in_supp:Tuning_parameter_selection+generalization_performance_estimation}.

\begin{table*}[tb]
\begin{center}
\begin{threeparttable}
\caption{Performance comparison on canonical correlation coefficients on real data.}
\label{tab:Training_testing_CC_real_data}
\vskip 0.15in
\begin{small}
\begin{sc}
\begin{tabular}{lccccccr}
\toprule

            &                                            & \multicolumn{4}{c}{Training} & \multicolumn{2}{c}{Testing} \\
                                                          \cmidrule(lr){3-6}\cmidrule(lr){7-8}
  Fold index     & $\left(c_1^{\rm opt},c_2^{\rm opt}\right)$ & Cov@Val$^*$ & Corr@Val$^*$ & Cov$^{**}$ & Corr$^{**}$ & Cov & Corr \\
\midrule
          \multicolumn{8}{c}{SCCA} \\
Fold 1	 &  (2, 4)	 &  \NA		 &   0.4880	 &   1.1583	 &   0.6775	 &   0.7114	 &   0.4451	 \\
Fold 2	 &  (1, 4)	 &  \NA		 &   0.4641	 &   0.5738	 &   0.5654	 &   0.4585	 &   0.4853	 \\
Fold 3	 &  (2, 2)	 &  \NA		 &   0.4480	 &   1.4923	 &   0.6001	 &   1.2068	 &   0.4826	 \\
Fold 4	 &  (2, 4)	 &  \NA		 &   0.4369	 &   1.0060	 &   0.6379	 &   0.8623	 &   0.5612	 \\
\hline
Full data	 &  (2, 4)	 &  \NA		 &  \NA		 &   1.1274	 &   0.0.6331	 &  \NA		 &  \NA		 \\
          \multicolumn{8}{c}{Simplified SCCA} \\
Fold 1	 &  (4, 16)	 &   6.1471	 &  \NA		 &   7.6180	 &   0.4551	 &   5.1317	 &   0.3222	 \\
Fold 2	 &  (4, 16)	 &   5.3965	 &  \NA		 &   7.1126	 &   0.4222	 &   6.8125	 &   0.4245	 \\
Fold 3	 &  (4, 16)	 &   5.8975	 &  \NA		 &   7.1964	 &   0.4346	 &   6.3124	 &   0.3931	 \\
Fold 4	 &  (4, 16)	 &   5.7150	 &  \NA		 &   7.3270	 &   0.4326	 &   5.2280	 &   0.3205	 \\
\hline
Full data	 &  (4, 16)	 &  \NA		 &  \NA		 &   7.1326	 &   0.4248	 &  \NA		 &  \NA		 \\
\bottomrule
\end{tabular}
\end{sc}
\end{small}
\begin{tablenotes}
      \footnotesize
\item[*] Cov@Val/Corr@Val: mean canonical covariance/correlation for the left-out folds in the inner cross-validation to select the regularization parameters. The reported value is the maximum mean canonical covariance/correlation over all candidate $\left(c_1,c_2\right)$ (i.e. at the optimal regularization parameters $\left(c_1^{\rm opt},c_2^{\rm opt}\right)$).
\item[**] Cov/Corr: mean canonical covariance/correlation when the optimal model is fit to the whole training data.
\end{tablenotes}
\end{threeparttable}
\end{center}
\vskip -0.1in
\end{table*}

\subsubsection{Simulation study results}
\label{subsubsec:Simulation_study_results}

Fig.~\ref{fig:synthetic_data_uv_esti} shows the canonical weight vectors estimated by SCCA and simplified SCCA on the entire training data. In Supplementary Materials Tables \ref{tab:X_variable_selection_performance_synthetic_data}-\ref{tab:Y_variable_selection_performance_synthetic_data}, we also summarize the variable selection performance in terms of recall, precision, F1 score, accuracy (ACC), balanced accuracy (bACC), Matthews correlation coefficient (MCC), precision-recall area under curve (PR AUC), and relative absolute error (RAE). The canonical correlation/covariance on the training and testing sets are reported in Table\ref{tab:Training_testing_CC_synthetic_data}.

Referring to Experimental setup 1 where \rev{the variables in $\mat{x}$} are uncorrelated, the standard SCCA consistently outperforms the simplified SCCA in both selection of variables in $\mat{x}$ and identification of strong canonical correlation.

Referring to Experimental setup 2 where \rev{the variables in $\mat{x}$} form in groups with full correlation within each group, the simplified SCCA always assigns the same weights to each group of variables in $\mat{x}$. However, for the standard SCCA, the weights of variables in $\mat{x}$ in the same group is randomly assigned, which leads to a few variables with large weights while remaining variables with weights \lsc{close to zero}. \lsc{Despite that the simplified SCCA can falsely detect variables group-wise,} it outperforms standard SCCA in selection of variables in $\mat{x}$. Note that, compared to standard SCCA, the simplified SCCA has slightly lower canonical correlation but much higher canonical covariance. \lsc{This is not surprising because} in the standard SCCA the objective is to maximize the canonical correlation while the simplified SCCA maximizes the canonical covariance.

Regarding the selection of \rev{variables in $\mat{y}$}, the simplified SCCA performs better than standard SCCA in both Experimental setups. \lsc{This is as expected} considering that \rev{the variables in $\mat{y}$} in \eqref{equ:uncorrelated_variables_random_vector_Y} are highly correlated.

\subsection{Application to real imaging genetic data}
\label{subsec:experiment_real_data}

We applied the two SCCA models to a real imaging genetics data set to compare their performances.
The genotyping and baseline AV-45 PET data of 757 non-Hispanic Caucasian subjects (age 72.26$\pm$7.17), including 183 healthy control (HC, 94 female), 75 significant memory concern (SMC, 46 female), 218 early mild cognitive impairment (EMCI, 105 female), 184 late MCI (LMCI, 88 female), and 97 Alzheimer's disease (AD, 43 female) participants, were downloaded from the Alzheimer's Disease Neuroimaging Initiative (ADNI) database \cite{weiner2017alzheimer}. One aim of ADNI has been to test whether serial magnetic resonance imaging (MRI), positron emission tomography (PET), other biological markers, and clinical and neuropsychological assessment can be combined to measure the progression of MCI and early AD. For up-to-date information, see \url{www.adni-info.org}.

The \lsc{AV-45 scans} were aligned to each participant's same visit MRI scan and normalized to the Montreal Neurological Institute (MNI) space. Region-of-interest (ROI) level AV-45 measurements were further extracted based on the MarsBaR AAL atlas. We focused on the analysis of 1,542 single nucleotide polymorphisms (SNPs) from \lsc{27 AD risk genes} and AV-45 imaging measures from 116 ROIs. Using the regression weights derived from the HC participants, the genotype and imaging measures were preadjusted for removing the effects of age, gender, education, and handedness.

The two SCCA models were applied to the ADNI data to identify bi-multivariate imaging genetics associations. We employed the nested five-fold cross-validation (which is an extension of the procedure described in Section \ref{subsubsec:Tuning_parameter_selection+generalization_performance_estimation}) to choose the regularization parameters and report the performance.
The genetic and imaging feature selection results are reported in Figures \ref{fig:genetic_weight_esti}-\ref{fig:imaging_weight_esti} and Tables \ref{tab:snps_selected}-\ref{tab:QTs_selected} in Supplementary Materials Section \ref{sec_in_supp:experiment_real_data}, while the canonical correlation/covariance performance is reported in Table \ref{tab:Training_testing_CC_real_data}.

For genetic feature selection (Figure \ref{fig:genetic_weight_esti} and Table \ref{tab:snps_selected}), both SCCA models select top AD risk genes such as \emph{APOE}, \emph{PICALM} and \emph{ABCA7}. However, in each gene, the simplified SCCA selects a cluster of SNPs while the standard SCCA only selects \lsc{one or very few SNPs which dominate}. Together with the correlation among the SNPs within each gene (Fig.~\ref{fig:corr_SNP_AV45} middle), it verifies that the simplified SCCA has the grouping effects in feature selection while the standard SCCA does not.

For imaging feature selection (Figure \ref{fig:imaging_weight_esti} and Table \ref{tab:QTs_selected}), although high correlation is prevalent among the 116 imaging features (Fig.~\ref{fig:corr_SNP_AV45} right), the standard SCCA only selects about 20 features while the simplified SCCA selects more than 60 features, which confirms that the simplified SCCA is prone to selecting correlated feature together.

\section{Conclusion}
The sparse canonical correlation analysis (SCCA) is a bi-multivariate model that maximizes the multivariate correlation between two sets of variables. Since SCCA is computationally expensive, a simplified SCCA model which maximizes the multivariate covariance, has been widely used as its surrogate. The fundamental properties of the solutions of these two models remain unknown. Through theoretical analysis, we show that these two models behave differently regarding the grouping effects in variable selection. The simplified SCCA jointly selects or \rev{deselects} a group of correlated variables together, while the standard SCCA randomly selects one or few representatives from a group of correlated variables. Empirical results on both synthetic and real data confirm our theoretical finding. \lsc{This result can guide users to choose the right SCCA model in practice.}

\newpage
\clearpage

\vfill

\bibliographystyle{IEEEtran}
\bibliography{CCA}

\begin{thebibliography}{10}
\providecommand{\url}[1]{#1}
\csname url@samestyle\endcsname
\providecommand{\newblock}{\relax}
\providecommand{\bibinfo}[2]{#2}
\providecommand{\BIBentrySTDinterwordspacing}{\spaceskip=0pt\relax}
\providecommand{\BIBentryALTinterwordstretchfactor}{4}
\providecommand{\BIBentryALTinterwordspacing}{\spaceskip=\fontdimen2\font plus
\BIBentryALTinterwordstretchfactor\fontdimen3\font minus
  \fontdimen4\font\relax}
\providecommand{\BIBforeignlanguage}[2]{{%
\expandafter\ifx\csname l@#1\endcsname\relax
\typeout{** WARNING: IEEEtran.bst: No hyphenation pattern has been}%
\typeout{** loaded for the language `#1'. Using the pattern for}%
\typeout{** the default language instead.}%
\else
\language=\csname l@#1\endcsname
\fi
#2}}
\providecommand{\BIBdecl}{\relax}
\BIBdecl

\bibitem{Hotelling1936}
H.~Hotelling, ``Relations between two sets of variates,'' \emph{Biometrika},
  vol.~28, pp. 321--377, 1936.

\bibitem{hardoon2004canonical}
D.~R. Hardoon, S.~Szedmak, and J.~Shawe-Taylor, ``Canonical correlation
  analysis: An overview with application to learning methods,'' \emph{Neural
  computation}, vol.~16, no.~12, pp. 2639--2664, 2004.

\bibitem{klami2013bayesian}
A.~Klami, S.~Virtanen \emph{et~al.}, ``Bayesian canonical correlation
  analysis,'' \emph{J. Mach. Learn. Res.}, vol.~14, no. Apr, pp. 965--1003,
  2013.

\bibitem{sun2010canonical}
L.~Sun, S.~Ji, and J.~Ye, ``Canonical correlation analysis for multilabel
  classification: A least-squares formulation, extensions, and analysis,''
  \emph{IEEE Trans Pattern Anal Mach Intell}, vol.~33, no.~1, pp. 194--200,
  2010.

\bibitem{worsley1997characterizing}
K.~J. Worsley, J.-B. Poline, K.~J. Friston, and A.~Evans, ``Characterizing the
  response of {PET} and {fMRI} data using multivariate linear models,''
  \emph{Neuroimage}, vol.~6, no.~4, pp. 305--319, 1997.

\bibitem{friman2001detection}
O.~Friman, J.~Cedefamn \emph{et~al.}, ``Detection of neural activity in
  functional {MRI} using canonical correlation analysis,'' \emph{Magnetic
  Resonance in Medicine}, vol.~45, no.~2, pp. 323--330, 2001.

\bibitem{yamanishi2003extraction}
Y.~Yamanishi, J.-P. Vert \emph{et~al.}, ``Extraction of correlated gene
  clusters from multiple genomic data by generalized kernel canonical
  correlation analysis,'' \emph{Bioinformatics}, vol.~19, no. suppl\_1, pp.
  i323--i330, 2003.

\bibitem{via2005canonical}
J.~Via, I.~Santamaria, and J.~P{\'e}rez, ``Canonical correlation analysis
  ({CCA}) algorithms for multiple data sets: Application to blind {SIMO}
  equalization,'' in \emph{IEEE European Signal Proc. Conf.}\hskip 1em plus
  0.5em minus 0.4em\relax IEEE, 2005, pp. 1--4.

\bibitem{hariri2003imaging}
A.~R. Hariri and D.~R. Weinberger, ``Imaging genomics,'' \emph{British medical
  bulletin}, vol.~65, no.~1, pp. 259--270, 2003.

\bibitem{shen2020pieee}
L.~{Shen} and P.~M. {Thompson}, ``Brain imaging genomics: Integrated analysis
  and machine learning,'' \emph{Proceedings of the IEEE}, vol. 108, no.~1, pp.
  125--162, Jan 2020.

\bibitem{waaijenborg2008quantifying}
S.~Waaijenborg, P.~C.~V. de~Witt~Hamer, and A.~H. Zwinderman, ``Quantifying the
  association between gene expressions and {DNA}-markers by penalized canonical
  correlation analysis,'' \emph{Statistical applications in genetics and
  molecular biology}, vol.~7, no.~1, 2008.

\bibitem{hardoon2011sparse}
D.~R. Hardoon and J.~Shawe-Taylor, ``Sparse canonical correlation analysis,''
  \emph{Machine Learning}, vol.~83, no.~3, pp. 331--353, 2011.

\bibitem{chu2013sparse}
D.~Chu, L.-Z. Liao, M.~K. Ng, and X.~Zhang, ``Sparse canonical correlation
  analysis: New formulation and algorithm,'' \emph{IEEE Trans Pattern Anal Mach
  Intell}, vol.~35, no.~12, pp. 3050--3065, 2013.

\bibitem{chi2013isbi}
E.~C. Chi, G.~I. Allen \emph{et~al.}, ``Imaging genetics via sparse canonical
  correlation analysis,'' in \emph{IEEE 10th Int Sym on Biomedical Imaging
  (ISBI)}, San Francisco, CA, 2013, pp. 740--743.

\bibitem{suo2017sparse}
X.~Suo, V.~Minden, B.~Nelson, R.~Tibshirani, and M.~Saunders, ``Sparse
  canonical correlation analysis,'' \emph{arXiv preprint arXiv:1705.10865},
  2017.

\bibitem{Parkhomenko2009}
E.~Parkhomenko, D.~Tritchler, and J.~Beyene, ``Sparse canonical correlation
  analysis with application to genomic data integration,'' \emph{Statistical
  Applications in Genetics and Molecular Biology}, vol.~8, pp. 1--34, 2009.

\bibitem{witten2009a}
D.~Witten, R.~Tibshirani, and T.~Hastie, ``A penalized matrix decomposition,
  with applications to sparse principal components and canonical correlation
  analysis,'' \emph{Biostatistics}, vol.~10, no.~3, pp. 515--34, 2009.

\bibitem{witten2009b}
D.~M. Witten and R.~J. Tibshirani, ``Extensions of sparse canonical correlation
  analysis with applications to genomic data,'' \emph{Stat Appl Genet Mol
  Biol}, vol.~8, no.~1, pp. 1--27, 2009.

\bibitem{chen2012_aistat}
X.~Chen, H.~Liu, and J.~G. Carbonell, ``Structured sparse canonical correlation
  analysis,'' in \emph{International Conference on Artificial Intelligence and
  Statistics}, vol.~12, La Palma, Canary Islands, 2012, pp. 199--207.

\bibitem{chen2013structure}
J.~Chen, F.~D. Bushman, J.~D. Lewis, G.~D. Wu, and H.~Li,
  ``Structure-constrained sparse canonical correlation analysis with an
  application to microbiome data analysis,'' \emph{Biostatistics}, vol.~14,
  no.~2, pp. 244--258, 2013.

\bibitem{Dudoit2002}
S.~Dudoit, J.~Fridlyand, and T.~P. Speed, ``Comparison of discrimination
  methods for the classification of tumors using gene expression data,''
  \emph{J. Am. Stat. Assoc.}, vol.~97, no. 457, pp. 77--87, 2002.

\bibitem{Tibshirani2003}
R.~Tibshirani, T.~Hastie, B.~Narasimhan, and G.~Chu, ``Class prediction by
  nearest shrunken centroids, with applications to {{DNA}} microarrays,''
  \emph{Statistical Science}, pp. 104--117, 2003.

\bibitem{Fang2016}
J.~Fang, D.~Lin \emph{et~al.}, ``Joint sparse canonical correlation analysis
  for detecting differential imaging genetics modules,'' \emph{Bioinformatics},
  vol.~32, no.~22, pp. 3480--3488, 2016.

\bibitem{mikewest2001predicting}
C.~B. MikeWest, H.~Dressman \emph{et~al.}, ``Predicting the clinical status of
  human breast cancer using gene expression profiles,'' \emph{PNAS}, 2001.

\bibitem{zou2005regularization}
H.~Zou and T.~Hastie, ``Regularization and variable selection via the elastic
  net,'' \emph{Journal of the royal statistical society: series B (statistical
  methodology)}, vol.~67, no.~2, pp. 301--320, 2005.

\bibitem{thompson2010curropinneurol}
P.~M. Thompson, N.~G. Martin, and M.~J. Wright, ``Imaging genomics,''
  \emph{Curr Opin Neurol}, vol.~23, no.~4, pp. 368--73, 2010.

\bibitem{Boyd2004}
S.~Boyd and L.~Vandenberghe, \emph{Convex optimization}.\hskip 1em plus 0.5em
  minus 0.4em\relax Cambridge university press, 2004.

\bibitem{Bezdek2002}
J.~C. Bezdek and R.~J. Hathaway, ``Some notes on alternating optimization,'' in
  \emph{AFSS International Conference on Fuzzy Systems}.\hskip 1em plus 0.5em
  minus 0.4em\relax Berlin, Heidelberg: Springer, 2002, pp. 288--300.

\bibitem{Boyd2011}
S.~Boyd, N.~Parikh, E.~Chu, B.~Peleato, and J.~Eckstein, ``Distributed
  optimization and statistical learning via the alternating direction method of
  multipliers,'' \emph{Foundations and Trends{\textregistered} in Machine
  Learning}, vol.~3, no.~1, pp. 1--122, 2011.

\bibitem{weiner2017alzheimer}
M.~W. Weiner, D.~P. Veitch \emph{et~al.}, ``The {Alzheimer}'s disease
  neuroimaging initiative 3: Continued innovation for clinical trial
  improvement,'' \emph{{Alzheimer}'s \& Dementia}, vol.~13, no.~5, pp.
  561--571, 2017.

\bibitem{eckart1936}
C.~Eckart and G.~Young, ``The approximation of one matrix by another of lower
  rank,'' \emph{Psychometrika}, vol.~1, no.~3, pp. 211--218, 1936.

\end{thebibliography}


\vfill

\newpage
\clearpage
\pagenumbering{arabic}

\setcounter{figure}{0}
\setcounter{table}{0}
\renewcommand{\thetable}{S\arabic{table}}
\renewcommand{\thefigure}{S\arabic{figure}}

\onecolumn

\section*{\centering \Large
Study on the grouping effects of two sparse CCA models in variable selection \\
\vspace{0.05in}
Supplementary materials}
\vspace{0.05in}

\appendix

\section{Linearized ADMM to solve subproblem \eqref{equ:SCCA_model_solve_for_u}}
\label{sec_in_supp:ADMM}

We present how to solve problem \eqref{equ:SCCA_model_solve_for_u} using the linearized alternating direction method of multipliers (ADMM)~\cite{Boyd2011,suo2017sparse}. The problem \eqref{equ:SCCA_model_solve_for_v} can be solved in a similar manner.

First, we write problem \eqref{equ:SCCA_model_solve_for_u} in the form:
\begin{equation}\label{equ:SCCA_model_solve_for_u_unconstrained_form}
\underset{\mat{u}}{\text{minimize}} \; -\mat{u}^\trans \mat{X}^\trans \mat{Y} \mat{v}  + \indicator{\twonorm{\mat{X} \mat{u}}^2 \leq 1} + \indicator{\normof{\mat{u}}{1} \leq c_1},
\end{equation}
where $\indicator{\cdot}$ is the indicator function defined as
\begin{equation*}
	\indicator{\mat{x} \in \mathbb{A}} =
	\begin{cases}
	0, & \mat{x} \in \mathbb{A} \\
    \infty, & \mat{x} \notin \mathbb{A}
	\end{cases}
\end{equation*}

To apply the ADMM, the problem \eqref{equ:SCCA_model_solve_for_u_unconstrained_form} is reformulated as
\begin{equation}\label{equ:SCCA_model_solve_for_u_unconstrained_form_ADMM}
\begin{aligned}
& \underset{\mat{u}}{\text{minimize}} & & -\mat{u}^\trans \mat{X}^\trans \mat{Y} \mat{v} + \indicator{\twonorm{\mat{z}}^2 \leq 1} + \indicator{\normof{\mat{u}}{1} \leq c_1} \\
& \text{subject to} & & \mat{X} \mat{u} = \mat{z}
\end{aligned}
\end{equation}

The augmented Lagrangian of problem \eqref{equ:SCCA_model_solve_for_u_unconstrained_form_ADMM} is
\begin{equation}\label{equ:SCCA_model_solve_for_u_augmented_Lagrangian}
\mathcal{L}_\rho\left(\mat{u}, \mat{z}, \mat{\lambda}\right) = -\mat{u}^\trans \mat{X}^\trans \mat{Y} \mat{v}  + \indicator{\twonorm{\mat{z}}^2 \leq 1} + \indicator{\normof{\mat{u}}{1} \leq c_1} + \inp{\mat{\lambda}}{\mat{X}\mat{u}-\mat{z}} + \frac{\rho}{2} \twonorm{\mat{X} \mat{u} - \mat{z}}^2.
\end{equation}

ADMM consists of the iterations:
\begin{align}
\mat{u}_{\ell+1} &= \argmin_{\mat{u}} \mathcal{L}_\rho\left(\mat{u}, \mat{z}_\ell, \mat{\lambda}_\ell\right) \\
\mat{z}_{\ell+1} &= \argmin_{\mat{z}} \mathcal{L}_\rho\left(\mat{u}_{\ell+1}, \mat{z}, \mat{\lambda}_\ell\right) \\
\mat{\lambda}_{\ell+1} &= \mat{\lambda}_\ell + \rho \left(\mat{X} \mat{u}_{\ell+1} - \mat{z}_{\ell+1}\right)
\end{align}

That is
\begin{align}
\mat{u}_{\ell+1} &= \argmin_{\mat{u}} -\mat{u}^\trans \mat{X}^\trans \mat{Y} \mat{v}  + \indicator{\normof{\mat{u}}{1} \leq c_1} + \inp{\mat{\lambda}_\ell}{\mat{X} \mat{u} - \mat{z}_\ell} + \frac{\rho}{2} \twonorm{\mat{X} \mat{u} - \mat{z}_\ell}^2 \label{equ:ADMM_update_u} \\
\mat{z}_{\ell+1} &= \argmin_{\mat{z}} \indicator{\twonorm{\mat{z}}^2 \leq 1} + \inp{\mat{\lambda}_\ell}{\mat{X} \mat{u}_{\ell+1} - \mat{z}} + \frac{\rho}{2} \twonorm{\mat{X} \mat{u}_{\ell+1} - \mat{z}}^2 \label{equ:ADMM_update_z} \\
\mat{\lambda}_{\ell+1} &= \mat{\lambda}_\ell + \rho \left(\mat{X} \mat{u}_{\ell+1} - \mat{z}_{\ell+1}\right) \label{equ:ADMM_update_y}
\end{align}

The problem \eqref{equ:ADMM_update_u} is not easy to solve due to the term $\frac{1}{2} \twonorm{\mat{X} \mat{u} - \mat{z}_\ell}^2 \triangleq f(\mat{u})$. To handle this, we construct a quadratic approximation of $f(\mat{u})$ near the estimate $\mat{u}_\ell$ of $\mat{u}$ in the previous iteration $\ell$:
\begin{align}
F(\mat{u}) &\triangleq f(\mat{u}_\ell) + \inp{\nabla f(\mat{u}_\ell)}{\mat{u}-\mat{u}_\ell} + \frac{L_X}{2} \twonorm{\mat{u} - \mat{u}_\ell}^2 \nonumber \\
           & = \frac{1}{2} \twonorm{\mat{X} \mat{u}_\ell - \mat{z}_\ell}^2 + \inp{\mat{X}^\trans \left(\mat{X} \mat{u}_\ell - \mat{z}_\ell\right)}{\mat{u}-\mat{u}_\ell} + \frac{L_X}{2} \twonorm{\mat{u} - \mat{u}_\ell}^2
\end{align}
where $L_X=\lambda_{\rm max}\left(\mat{X}^\trans \mat{X}\right)$, where $\lambda_{\rm max}\left(\cdot\right)$ is the largest eigenvalue of its argument/input. Note that we have $F(\mat{u}) \geq f(\mat{u})$ for any $\mat{u} \in \real^{p \times 1}$ and $F(\mat{u}_\ell)=f(\mat{u}_\ell)$.

In the linearized ADMM, it solves the approximate version of problem \eqref{equ:ADMM_update_u} with the term $f(\mat{u})=\frac{1}{2}\twonorm{\mat{X} \mat{u} - \mat{z}_\ell}^2$ replaced by $F(\mat{u})$:
\begin{align}\label{equ:linearized_ADMM_update_formula_u}
    \mat{u}_{\ell+1} &= \argmin_{\mat{u}} -\mat{u}^\trans \mat{X}^\trans \mat{Y} \mat{v}  + \indicator{\normof{\mat{u}}{1} \leq c_1} + \inp{\mat{\lambda}_\ell}{\mat{X} \mat{u} - \mat{z}_\ell} + \rho F(\mat{u}) \nonumber \\
                  &= \argmin_{\mat{u}} \frac{\rho L_X}{2} \twonorm{\mat{u} - \mat{u}_\ell + \frac{1}{L_X} \mat{X}^\trans \left(\mat{X} \mat{u}_\ell - \mat{z}_\ell  + \frac{1}{\rho} \mat{\lambda}_\ell - \frac{1}{\rho} \mat{Y} \mat{v}\right)} + \indicator{\normof{\mat{u}}{1} \leq c_1} \nonumber \\
                  &= \proxopof{prox}{L1}{\mat{u}_\ell - \frac{1}{L_X} \mat{X}^\trans \left(\mat{X} \mat{u}_\ell - \mat{z}_\ell + \frac{1}{\rho} \mat{\lambda}_\ell - \frac{1}{\rho} \mat{Y} \mat{v} \right); c_1}
\end{align}
where the proximal operator $\proxopof{prox}{L1}{\cdot; \cdot}$ is defined as
\begin{equation}\label{equ:def_prox}
	\proxopof{prox}{L1}{\mat{a};c} = \argmin_{\mat{x}} \; \frac{1}{2} \twonorm{\mat{x} - \mat{a}}^2 + \indicator{\normof{\mat{x}}{1} \leq c} =
	\begin{cases}
	\mat{a}, & \normof{\mat{a}}{1} \leq c \\
	\opof{S}{\mat{a},\Delta}, & \normof{\mat{a}}{1} > c
	\end{cases}
\end{equation}
where $\Delta$ is a positive constant that satisfies $\normof{\opof{S}{\mat{a},\Delta}}{1} = c$.

The update formula of problem \eqref{equ:ADMM_update_z} is
\begin{equation}\label{equ:ADMM_update_formula_z}
    \mat{z}_{\ell+1} = 	
    \begin{cases}
	\mat{X} \mat{u}_{\ell+1} + \mat{\lambda}_\ell / \rho, & \twonorm{\mat{X} \mat{u}_{\ell+1} + \mat{\lambda}_\ell / \rho} \leq 1 \\
	\frac{\mat{X} \mat{u}_{\ell+1} + \mat{\lambda}_\ell / \rho}{\twonorm{\mat{X} \mat{u}_{\ell+1} + \mat{\lambda}_\ell / \rho}}, & \twonorm{\mat{X} \mat{u}_{\ell+1} + \mat{\lambda}_\ell / \rho} > 1
	\end{cases}
\end{equation}

Taken together, the updates at each ADMM iteration are
\begin{align}
\mat{u}_{\ell+1} &= \proxopof{prox}{L1}{\mat{u}_\ell - \frac{1}{L_X} \mat{X}^\trans \left(\mat{X} \mat{u}_\ell - \mat{z}_\ell + \mat{\xi}_\ell - \frac{1}{\rho} \mat{Y} \mat{v}\right); c_1} \label{equ:ADMM_update_formula_u} \\
\mat{z}_{\ell+1} &= 	
    \begin{cases}
	\mat{X} \mat{u}_{\ell+1} + \mat{\xi}_\ell, & \twonorm{\mat{X} \mat{u}_{\ell+1} + \mat{\xi}_\ell} \leq 1 \\
	\frac{\mat{X} \mat{u}_{\ell+1} + \mat{\xi}_\ell}{\twonorm{\mat{X} \mat{u}_{\ell+1} + \mat{\xi}_\ell}}, & \twonorm{\mat{X} \mat{u}_{\ell+1} + \mat{\xi}_\ell} > 1
	\end{cases} \label{equ:ADMM_update_formula_z} \\
\mat{\xi}_{\ell+1} &= \mat{\xi}_\ell + \mat{X} \mat{u}_{\ell+1} - \mat{z}_{\ell+1} \label{equ:ADMM_update_formula_y}
\end{align}
where $\mat{\xi}_\ell = \mat{\lambda}_\ell / \rho$.

The Linearized ADMM algorithm for fitting the SCCA model is summarized in Algorithm~\ref{alg:LinearizedADMM}.

\begin{algorithm*}
	\caption{Sparse CCA fitting algorithm: Linearized ADMM}
	\label{alg:LinearizedADMM}
	\begin{algorithmic}[1]
		\REQUIRE
        \par \noindent $\mat{X} \in \real^{n \times p}$, $\mat{Y} \in \real^{n \times q}$, with column-wise zero empirical mean (the sample mean of each column has been shifted to zero);
		\par \indent Regularization parameters $c_1$ and $c_2$.
		\STATE Calculate Lipschitz constants: $L_X=\lambda_{\rm max}\left(\mat{X}^\trans \mat{X}\right)$, $L_Y=\lambda_{\rm max}\left(\mat{Y}^\trans \mat{Y}\right)$;
		\STATE Initialization: $\mat{u}^{(0)} \in \real^{p \times 1}, \mat{v}^{(0)} \in \real^{q \times 1}$;
        \STATE Set the penalty parameter to $\rho_1=\rho_2=1$ \cite{Boyd2011};
		\STATE $k=0$; 
        \REPEAT	
		\STATE \underline{Update $\mat{u}$:}
         \STATE $\mat{a} = \mat{Y} \mat{v}^{(k)}$
        \STATE Input: $\mat{u}_0=\mat{u}^{(k)} \in \real^{p \times 1}$, $\mat{z}_0 = \mat{\xi}_0 = \mat{0} \in \real^{n \times 1}$;
        \FOR{$\ell=0,1,2,\dots$} 
        \STATE
\begin{align*}
\mat{u}_{\ell+1} &= \proxopof{prox}{L1}{\mat{u}_\ell - \frac{1}{L_X} \mat{X}^\trans \left(\mat{X} \mat{u}_\ell - \mat{z}_\ell + \mat{\xi}_\ell - \frac{1}{\rho_1} \mat{a}\right); c_1} \\
\mat{z}_{\ell+1} &= 	
    \begin{cases}
	\mat{X} \mat{u}_{\ell+1} + \mat{\xi}_\ell, & \twonorm{\mat{X} \mat{u}_{\ell+1} + \mat{\xi}_\ell} \leq 1 \\
	\frac{\mat{X} \mat{u}_{\ell+1} + \mat{\xi}_\ell}{\twonorm{\mat{X} \mat{u}_{\ell+1} + \mat{\xi}_\ell}}, & \twonorm{\mat{X} \mat{u}_{\ell+1} + \mat{\xi}_\ell} > 1
	\end{cases} \\
\mat{\xi}_{\ell+1} &= \mat{\xi}_\ell + \mat{X} \mat{u}_{\ell+1} - \mat{z}_{\ell+1}
\end{align*}
	\ENDFOR
	\STATE Output: $\mat{u}^{(k+1)} = \mat{u}_{\ell+1}$.

\

		\STATE \underline{Update $\mat{v}$:}
         \STATE $\mat{b} = \mat{X} \mat{u}^{(k+1)}$
        \STATE Initialization: $\mat{v}_0=\mat{v}^{(k)} \in \real^{q \times 1}$, $\mat{\zeta}_0 = \mat{\psi}_0 = \mat{0} \in \real^{n \times 1}$; 
        \FOR{$\ell=0,1,2,\dots$} 
        \STATE
\begin{align*}
\mat{v}_{\ell+1} &= \proxopof{prox}{L1}{\mat{v}_\ell - \frac{1}{L_Y} \mat{Y}^\trans \left(\mat{Y} \mat{v}_\ell - \mat{\zeta}_\ell + \mat{\psi}_\ell - \frac{1}{\rho_2} \mat{b}\right); c_2} \\
\mat{\zeta}_{\ell+1} &= 	
    \begin{cases}
	\mat{Y} \mat{v}_{\ell+1} + \mat{\psi}_\ell, & \twonorm{\mat{Y} \mat{v}_{\ell+1} + \mat{\psi}_\ell} \leq 1 \\
	\frac{\mat{Y} \mat{v}_{\ell+1} + \mat{\psi}_\ell}{\twonorm{\mat{Y} \mat{v}_{\ell+1} + \mat{\psi}_\ell}}, & \twonorm{\mat{Y} \mat{v}_{\ell+1} + \mat{\psi}_\ell} > 1
	\end{cases} \\
\mat{\psi}_{\ell+1} &= \mat{\psi}_\ell + \mat{Y} \mat{v}_{\ell+1} - \mat{\zeta}_{\ell+1}
\end{align*}
	\ENDFOR
\STATE Output: $\mat{v}^{(k+1)}=\mat{v}_{\ell+1}$.

\STATE $k \gets k+1$.
\UNTIL{convergence.} 
	\end{algorithmic}
\end{algorithm*}

\section{Proof of Lemma \ref{lema:solution_QCLP} and how it extends Lemma 2.2 of~\cite{witten2009a}}
\label{sec_in_supp:lemma2_2}

\subsection{Proof of Lemma \ref{lema:solution_QCLP}}
\label{subsec_in_supp:proof_lemma2_2}

\begin{proof}
Since the problem \eqref{equ:QCLP} is a convex optimization problem with differentiable objective and constraint functions (Note that the L1 inequality constraint can be written as $2^p$ linear inequality constraints), and is strictly feasible (Slater's condition holds), the KKT conditions provide necessary and sufficient conditions for optimality \cite{Boyd2004}.

The Lagrangian function is
\begin{equation*}
    \mathcal{L} \left(\mat{u}, \alpha, \Delta\right) = - \mat{a}^\trans \mat{u} + \frac{\alpha}{2} \left(\twonorm{\mat{u}}^2 - 1\right) + \Delta \left(\normof{\mat{u}}{1} - c\right)
\end{equation*}
where $\alpha$ and $\Delta$ are the Lagrange multipliers (dual variables) for the L2 and L1 constraints, respectively.

Setting the differential of $\mathcal{L} \left(\mat{u}, \alpha, \Delta\right)$ with respect to $\mat{u}$ equal to zero yields
\begin{equation}\label{equ:QCLP_KKT_conditions_stationarity}
    \alpha \mat{u} + \Delta \mat{s} = \mat{a}
\end{equation}
where \rev{$\mat{s}$ is the subgradient of $\normof{\mat{u}}{1}$ with respect to $\mat{u}$, with $s_i=\signof{u_i}$ if $u_i \neq 0$ and $s_i \in [-1,1]$ otherwise.}

The KKT conditions for optimality consist of \eqref{equ:QCLP_KKT_conditions_stationarity} and
\begin{equation}\label{equ:QCLP_KKT_conditions_L2_constraint}
    \alpha \geq 0, \quad \twonorm{\mat{u}}^2 \leq 1, \quad \alpha \left(\twonorm{\mat{u}}^2 - 1\right) = 0
\end{equation}
\begin{equation}\label{equ:QCLP_KKT_conditions_L1_constraint}
    \Delta \geq 0, \quad \normof{\mat{u}}{1} \leq c, \quad \Delta \left(\normof{\mat{u}}{1} - c\right) = 0
\end{equation}

\begin{itemize}
  \item Case 1: $\alpha=0$, $\Delta>0$.

The KKT conditions \eqref{equ:QCLP_KKT_conditions_stationarity}-\eqref{equ:QCLP_KKT_conditions_L1_constraint} are simplified as
\begin{align}
    \Delta \mat{s} & = \mat{a}, \; \Delta > 0 \label{equ:QCLP_KKT_conditions_stationarity_Case1} \\
    \twonorm{\mat{u}}^2 & \leq 1 \label{equ:QCLP_KKT_conditions_L2_constraint_Case1} \\
    \normof{\mat{u}}{1} &= c \label{equ:QCLP_KKT_conditions_L1_constraint_Case1}
\end{align}

From \eqref{equ:QCLP_KKT_conditions_stationarity_Case1}, it follows that $\Delta = \max_{1 \leq i \leq p} \abs{a_i}$ and $u_i = 0$ for any $i \notin \set{S}$, where $\set{S} = \left\{i \mid \; \abs{a_{i}} = \Delta \right\}$.

Therefore, an optimal solution can be written in the following form:
\begin{equation}\label{equ:general_solution_QCLP_Case1}
	[\mat{u}^*]_i =
	\begin{cases}
	w_i \signof{a_i} , & i \in \set{S} \\
	0, & i \notin \set{S}
	\end{cases}
\end{equation}
with $w_i\geq 0$ satisfying
\begin{equation*}
\sum_{i \in \set{S}} w_i^2 \leq 1, \quad \sum_{i \in \set{S}} w_i = c
\end{equation*}

When $c \leq \sqrt{\abs{\set{S}}}$, the set of solutions defined above is non-empty, and among them the solution with minimum Euclidean norm is shown in \eqref{equ:solution_QCLP_Case1}.

  \item Case 2: $\alpha>0$

\begin{itemize}
  \item Case 2.1: $\Delta=0$
\end{itemize}
  The KKT conditions \eqref{equ:QCLP_KKT_conditions_stationarity}-\eqref{equ:QCLP_KKT_conditions_L1_constraint} are simplified as
\begin{align}
    \alpha \mat{u} & = \mat{a}, \; \alpha > 0 \label{equ:QCLP_KKT_conditions_stationarity_Case2} \\
    \twonorm{\mat{u}}^2 &= 1 \label{equ:QCLP_KKT_conditions_L2_constraint_Case2} \\
    \normof{\mat{u}}{1} &\leq c \label{equ:QCLP_KKT_conditions_L1_constraint_Case2}
\end{align}

From \eqref{equ:QCLP_KKT_conditions_stationarity_Case2}-\eqref{equ:QCLP_KKT_conditions_L2_constraint_Case2}, it follows that $\mat{u} = \frac{\mat{a}}{\twonorm{\mat{a}}}$.

When $c \geq \frac{\normof{\mat{a}}{1}}{\twonorm{\mat{a}}}$, the above $\mat{u}$ also satisfies \eqref{equ:QCLP_KKT_conditions_L1_constraint_Case2} and is therefore the optimal solution.

  \item Case 2.2: $\Delta>0$
  The conditions \eqref{equ:QCLP_KKT_conditions_L2_constraint}-\eqref{equ:QCLP_KKT_conditions_L1_constraint} become
\begin{equation}\label{equ:QCLP_KKT_conditions_L2_constraint_Case3}
    \alpha > 0, \quad \twonorm{\mat{u}}^2 = 1
\end{equation}
\begin{equation}\label{equ:QCLP_KKT_conditions_L1_constraint_Case3}
    \Delta > 0, \quad \normof{\mat{u}}{1} = c
\end{equation}

Combining conditions \eqref{equ:QCLP_KKT_conditions_stationarity} and \eqref{equ:QCLP_KKT_conditions_L2_constraint_Case3}-\eqref{equ:QCLP_KKT_conditions_L1_constraint_Case3}, we obtain the optimal solution shown in
Eq. \eqref{equ:solution_QCLP_Case3}. This corresponds to the range $\sqrt{\abs{\set{S}}} \leq c < \frac{\normof{\mat{a}}{1}}{\twonorm{\mat{a}}}$.

\end{itemize}

\end{proof}

\subsection{How does Lemma \ref{lema:solution_QCLP} extend Lemma 2.2 of~\cite{witten2009a}?}
\label{subsec_in_supp:lemma2_2_completed}

Fig.~\ref{fig:illustration_Lemma_2_2_fails} illustrates two particular scenarios of Case 1 in Lemma \ref{lema:solution_QCLP} for $p=2$ where Lemma 2.2 in~\cite{witten2009a} fails.
\begin{figure}[!htb]
\begin{center}
\subfigure{\label{fig:a}\includegraphics[width=.4\columnwidth]{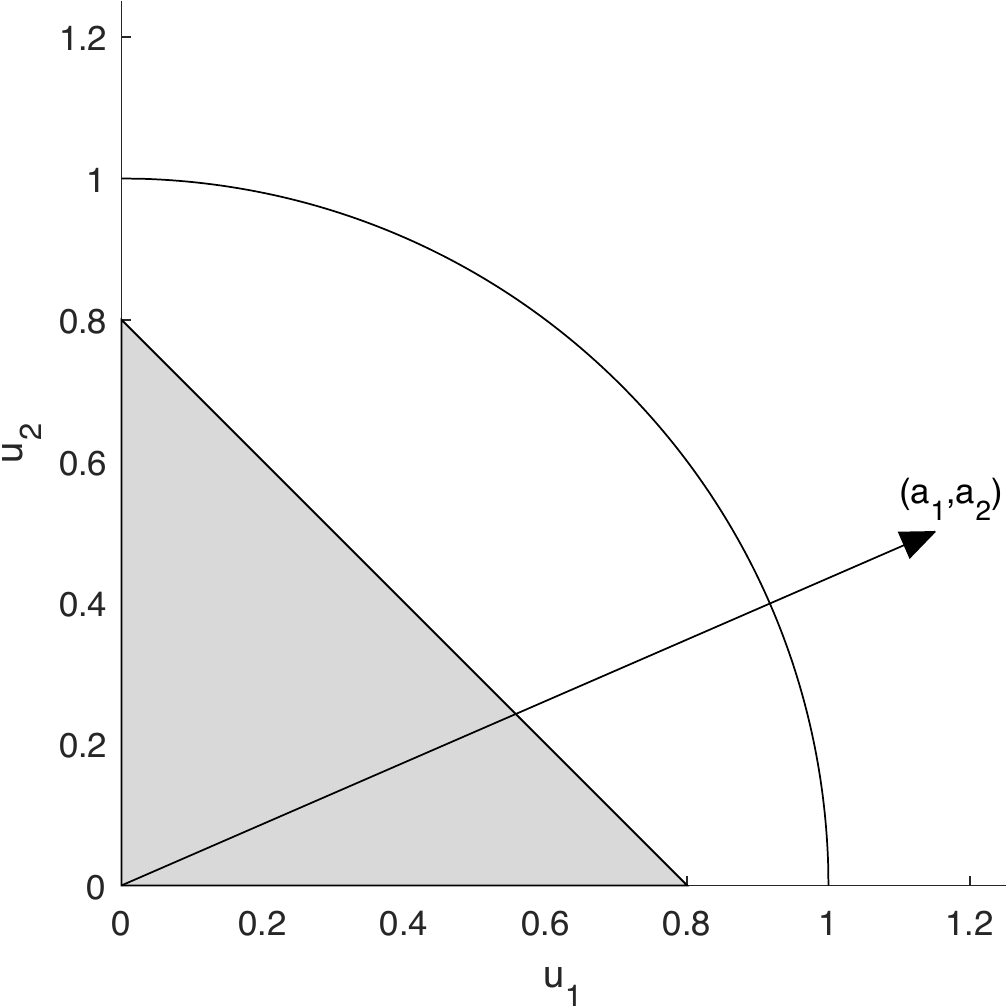}}
\hspace{1cm}
\subfigure{\label{fig:b}\includegraphics[width=.4\columnwidth]{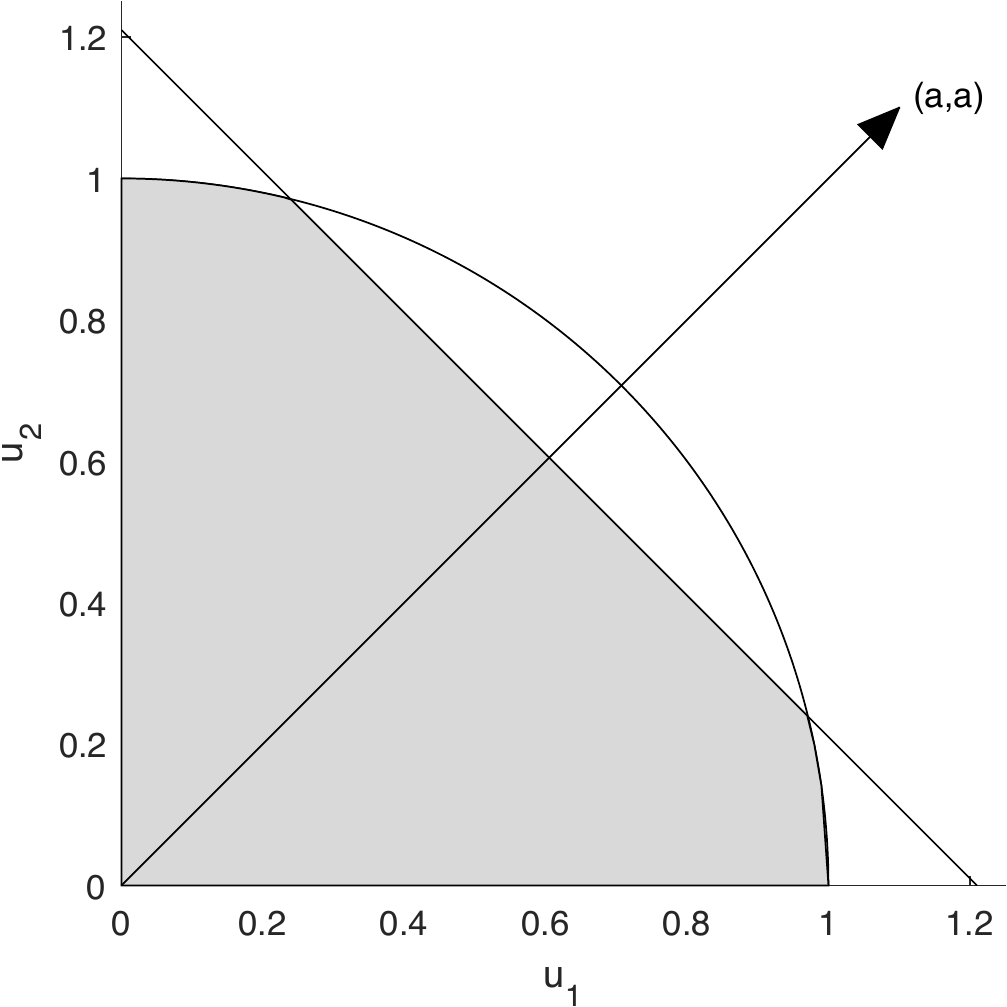}}
\caption{Two particular scenarios of Case 1 in Lemma \ref{lema:solution_QCLP} where Lemma 2.2 in~\cite{witten2009a} does not consider in the optimization problem. The dimension is $p=2$. The shaded area shows the domain [feasible/constraint set/region] (defined by the L2 and L1 constraints) of the objective function of problem \eqref{equ:QCLP}.
(a) $c = 0.8 < 1$;
(b) $c = 1.25 < \sqrt{2}$ and $a_1=a_2=a$.
\lsc{In both cases, the optimal solution in (a) (the point $\mat{u}*=[0.8; 0]$) and the optimal solution in (b) (any point on the chord of the circle) do not have a form that is shown in Lemma 2.2 in~\cite{witten2009a}.}}
\label{fig:illustration_Lemma_2_2_fails}
\end{center}
\end{figure}

\lsc{Essentially, the expression presented in Lemma 2.2 of~\cite{witten2009a} is the solution to
	\begin{equation}\label{equ:QCLP_L2equality_constraint}
	\underset{\mat{u}}{\text{maximize}} \; \mat{a}^\trans \mat{u}  \quad  \text{subject to } \twonorm{\mat{u}}^2 = 1, \normof{\mat{u}}{1} \leq c
	\end{equation}
while in the problem \eqref{equ:QCLP} that Lemma \ref{lema:solution_QCLP} is solving, the L2 equality constraint is replaced by the L2 inequality constraint, resulting in a convex problem.

Note that in order for problem \eqref{equ:QCLP_L2equality_constraint} to have an optimal solution of the form presented in Lemma 2.2 of~\cite{witten2009a}, $c$ must be larger than or equal to $\sqrt{\abs{\set{S}}}$, where $\set{S}$ is a set defined as $\set{S} = \left\{i: i \in \argmax_j \abs{a_j} \right\}$ (see the proof of Lemma \ref{lema:solution_QCLP}). Otherwise,
\begin{itemize}
  \item when $0 \leq c < 1$, problem \eqref{equ:QCLP_L2equality_constraint} is infeasible (there are no feasible points that satisfy the constraints).
  \item when $1 < c < \sqrt{\abs{\set{S}}}$, the optimal solution to problem \eqref{equ:QCLP_L2equality_constraint} is
\begin{subequations}\label{equ:solution_QCLP_L2equality_constraint_Case1}
\begin{equation}
	[\mat{u}^*]_i =
	\begin{cases}
	w_i \signof{a_i} , & i \in \set{S} \\
	0, & i \notin \set{S}
	\end{cases}
\end{equation}
where $w_i \geq 0$, $i \in \set{S}$, satisfy
\begin{equation}
\sum_{i \in \set{S}} w_i^2 = 1, \quad \sum_{i \in \set{S}} w_i = c
\end{equation}
\end{subequations}
\end{itemize}
Note that solution \eqref{equ:solution_QCLP_L2equality_constraint_Case1} cannot be written in the form shown in Lemma 2.2 in~\cite{witten2009a}.

By contrast, problem \eqref{equ:QCLP} has an optimal solution for every $c \geq 0$.}

\section{How to obtain multiple canonical components of SCCA model sequentially?}
\label{sec_in_supp:multiple_components}
In this section, we will show how to apply the single-canonical-component SCCA algorithms to sequentially compute multiple canonical components of standard and simplified SCCA models. \lsc{Note that except that Algorithm \ref{alg:simplified_SCCA_for_multiple_canonical_components_deflated_sample_cov_matrix_XY} was described in~\cite{witten2009a}, all algorithms (Algorithms \ref{alg:standard_SCCA_for_multiple_canonical_components_deflated_sample_cov_matrix_XY}-\ref{alg:standard_SCCA_for_multiple_canonical_components_deflated_data} for standard the SCCA model and Algorithm \ref{alg:simplified_SCCA_for_multiple_canonical_components_deflated_data} for the simplified SCCA model) and their theoretical justifications in Sections \ref{subsubsec_in_supp:proof_cal_multiple_canonical_components_for_standard_SCCA} and \ref{subsubsec_in_supp:proof_cal_multiple_canonical_components_for_simplified_SCCA} are new, to the best our knowledge.}

\subsection{Sequential calculation of multiple canonical components of standard SCCA}
\label{subsec_in_supp:proof_cal_multiple_canonical_components_for_standard_SCCA}

The SCCA model for computing $R$ canonical components is
\begin{equation}\label{equ:standard_SCCA_model_multiple_components}
\begin{aligned}
& \underset{\mat{U},\mat{V}}{\text{maximize}} & & \traceof{\mat{U}^\trans \mat{\hat{\Sigma}}_{\mat{xy}} \mat{V}} \\
&                           \text{subject to} & & \mat{U}^\trans \mat{\hat{\Sigma}}_{\mat{xx}} \mat{U} = \mat{I}_R, \normof{\mat{u}_r}{1} \leq c_{1r}, r=1,2,\dots,R\\
&                                             & & \mat{V}^\trans \mat{\hat{\Sigma}}_{\mat{yy}} \mat{V} = \mat{I}_R, \normof{\mat{v}_r}{1} \leq c_{2r}, r=1,2,\dots,R
\end{aligned}
\end{equation}
where
\begin{align}
    \mat{\hat{\Sigma}}_{\mat{xy}} &= \frac{1}{n-1} \mat{X}^\trans \mat{Y} \label{equ:sample_cov_matrix_XY} \\
    \mat{\hat{\Sigma}}_{\mat{xx}} &= \frac{1}{n-1} \mat{X}^\trans \mat{X} \label{equ:sample_cov_matrix_X} \\
    \mat{\hat{\Sigma}}_{\mat{yy}} &= \frac{1}{n-1} \mat{Y}^\trans \mat{Y} \label{equ:sample_cov_matrix_Y}
\end{align}
are the sample cross-covariance between \rev{random vectors $\mat{x}$ and $\mat{y}$}, sample auto-covariance matrix within \rev{random vector $\mat{x}$} and sample auto-covariance matrix within \rev{random vector $\mat{y}$}, respectively. Here we assume that the columns of $\mat{X}$ and $\mat{Y}$ have been centered to zero mean.

For clarity, we first present two algorithms (Algorithms \ref{alg:standard_SCCA_for_multiple_canonical_components_deflated_sample_cov_matrix_XY} and \ref{alg:standard_SCCA_for_multiple_canonical_components_deflated_data}) to sequentially compute multiple canonical components of SCCA: one is based on deflation of the cross-covariance matrix, and the other one is based on deflation of the data matrices. Then we provide theoretical explanations of both algorithms in the subsequent sections.

\begin{algorithm}
	\caption{Sequential computation of $R$ canonical components of SCCA via deflation of the cross-covariance matrix.}
	\label{alg:standard_SCCA_for_multiple_canonical_components_deflated_sample_cov_matrix_XY}
	\begin{algorithmic}[1]
\STATE Let $\mat{\hat{\Sigma}}_{\mat{xy}}^0 = \frac{1}{n-1} \mat{X}^\trans \mat{Y} \in \real^{p \times q}$, $\mat{\hat{\Sigma}}_{\mat{xx}} = \frac{1}{n-1} \mat{X}^\trans \mat{X} \in \real^{p \times p}$ and $\mat{\hat{\Sigma}}_{\mat{yy}} = \frac{1}{n-1} \mat{Y}^\trans \mat{Y} \in \real^{q \times q}$.
\FOR{$r=1,2,\dots,R$} 
\STATE Find the $r$-th pair of canonical weight vectors $\mat{\hat{u}}_r$ and $\mat{\hat{v}}_r$ by applying the single-canonical-component SCCA algorithm to $\left(\mat{\hat{\Sigma}}_{\mat{xy}}^{r-1}, \mat{\hat{\Sigma}}_{\mat{xx}}, \mat{\hat{\Sigma}}_{\mat{yy}}\right)$:
\begin{equation*}
\begin{aligned}
& \underset{\mat{u}_r,\mat{v}_r}{\text{maximize}} & & \mat{u}_r^\trans \mat{\hat{\Sigma}}_{\mat{xy}}^{r-1} \mat{v}_r \\
&                           \text{subject to} & & \mat{u}_r^\trans \mat{\hat{\Sigma}}_{\mat{xx}} \mat{u}_r \leq 1, \normof{\mat{u}_r}{1} \leq c_{1r} \\
&                                             & & \mat{v}_r^\trans \mat{\hat{\Sigma}}_{\mat{yy}} \mat{v}_r \leq 1, \normof{\mat{v}_r}{1} \leq c_{2r}
\end{aligned}
\end{equation*}
\STATE $\mat{\hat{\Sigma}}_{\mat{xy}}^r \gets \mat{\hat{\Sigma}}_{\mat{xy}}^{r-1} - \mat{\hat{\Sigma}}_{\mat{xx}} \hat{d}_r \mat{\hat{u}}_r \mat{\hat{v}}_r^\trans \mat{\hat{\Sigma}}_{\mat{yy}}$, where $\hat{d}_r = \frac{\mat{\hat{u}}_r^\trans \mat{\hat{\Sigma}}_{\mat{xy}}^{r-1} \mat{\hat{v}}_r}{\mat{\hat{u}}_r^\trans \mat{\hat{\Sigma}}_{\mat{xx}} \mat{\hat{u}}_r \cdot \mat{\hat{v}}_r^\trans \mat{\hat{\Sigma}}_{\mat{yy}} \mat{\hat{v}}_r}$.
\ENDFOR
	\end{algorithmic}
\end{algorithm}

\begin{algorithm}
	\caption{Sequential computation of $R$ canonical components of SCCA via deflation of the data matrices.}
	\label{alg:standard_SCCA_for_multiple_canonical_components_deflated_data}
	\begin{algorithmic}[1]
\STATE Let $\mat{X}^0 = \mat{X} \in \real^{n \times p}$, $\mat{Y}^0 = \mat{Y} \in \real^{n \times q}$.
\FOR{$r=1,2,\dots,R$} 
\STATE Find the $r$-th pair of canonical weight vectors $\left(\mat{\hat{u}}_r,\mat{\hat{v}}_r\right)$ by applying Algorithm \ref{alg:LinearizedADMM} to solve
\begin{equation*}
\begin{aligned}
& \underset{\mat{u},\mat{v}}{\text{maximize}} & & \frac{1}{n-1} \mat{u}_r^\trans {\mat{X}^{r-1}}^\trans \mat{Y}^{r-1} \mat{v}_r \\
& \text{subject to}                           & & \frac{1}{n-1} \mat{u}_r^\trans \mat{X}^\trans \mat{X} \mat{u}_r \leq 1, \normof{\mat{u}_r}{1} \leq c_{1r} \\
&                                             & & \frac{1}{n-1} \mat{v}_r^\trans \mat{Y}^\trans \mat{Y} \mat{v}_r \leq 1, \normof{\mat{v}_r}{1} \leq c_{2r} \\
\end{aligned}
\end{equation*}
\STATE Calculate the residual data:
\begin{align}
\mat{X}^r &\gets \mat{X}^{r-1} - \mat{X}^{r-1} \frac{\mat{\hat{u}}_r \mat{\hat{u}}_r^\trans \mat{X}^\trans \mat{X}}{\mat{\hat{u}}_r^\trans \mat{X}^\trans \mat{X} \mat{\hat{u}}_r} \label{equ:standard_SCCA_in_finite_sample_setting_deflated_data_matrix_X_iter_r} \\
\mat{Y}^r &\gets \mat{Y}^{r-1} - \mat{Y}^{r-1} \frac{\mat{\hat{v}}_r \mat{\hat{v}}_r^\trans \mat{Y}^\trans \mat{Y}}{\mat{\hat{v}}_r^\trans \mat{Y}^\trans \mat{Y} \mat{\hat{v}}_r} \label{equ:standard_SCCA_in_finite_sample_setting_deflated_data_matrix_Y_iter_r}
\end{align}
\ENDFOR
	\end{algorithmic}
\end{algorithm}

\begin{remark}
The deflated data in Eqs. \eqref{equ:standard_SCCA_in_finite_sample_setting_deflated_data_matrix_X_iter_r}-\eqref{equ:standard_SCCA_in_finite_sample_setting_deflated_data_matrix_Y_iter_r} can also be interpreted as the residual matrix of linear least squares regression: $\underset{\mat{z} \in \real^n}{\text{minimize}} \; \normof{\mat{X}^{r-1} \left(\mat{X}^\trans \mat{X}\right)^{-1/2} - \mat{z} \cdot \left[\left(\mat{X}^\trans \mat{X}\right)^{1/2} \mat{\hat{u}}_r\right]^\trans}{\rm F}^2$ and $\underset{\mat{\zeta} \in \real^n}{\text{minimize}} \; \normof{\mat{Y}^{r-1} \left(\mat{Y}^\trans \mat{Y}\right)^{-1/2} - \mat{\zeta} \cdot \left[\left(\mat{Y}^\trans \mat{Y}\right)^{1/2} \mat{\hat{v}}_r\right]^\trans}{\rm F}^2$, respectively.
\end{remark}

\clearpage

\subsubsection{Sequential calculation of multiple SCCA canonical components in the large-sample-size asymptotic regime}
\label{subsubsec_in_supp:proof_cal_multiple_canonical_components_for_standard_SCCA}

To compute $R$ canonical components sequentially/greedily, we consider the asymptotic regime of $n\to\infty$ in which case model \eqref{equ:standard_SCCA_model_multiple_components} becomes
\begin{equation}\label{equ:standard_SCCA_model_multiple_components_in_asymptotic_regime}
\begin{aligned}
& \underset{\mat{U},\mat{V}}{\text{maximize}} & & \traceof{\mat{U}^\trans \mat{\Sigma}_{\mat{xy}} \mat{V}} \\
& \text{subject to}                           & & \mat{U}^\trans \mat{\Sigma}_{\mat{xx}} \mat{U} = \mat{I}_R \\
&                                             & & \mat{V}^\trans \mat{\Sigma}_{\mat{yy}} \mat{V} = \mat{I}_R \\
\end{aligned}
\end{equation}
where $\mat{\Sigma}_{\mat{xy}}$, $\mat{\Sigma}_{\mat{xx}}$ and $\mat{\Sigma}_{\mat{yy}}$ are the population cross-covariance matrix between \rev{random vectors $\mat{x}$ and $\mat{y}$}, population auto-covariance matrix within \rev{random vector $\mat{x}$} and population auto-covariance matrix within \rev{random vector $\mat{y}$}, respectively. \lsc{Note that in model \eqref{equ:standard_SCCA_model_multiple_components_in_asymptotic_regime} we have dropped the L1 regularizers: since we have infinite amount of data available for use, the L1 regularizations are no longer necessary.}

The Lagrangian function of problem \eqref{equ:standard_SCCA_model_multiple_components_in_asymptotic_regime} is defined as
\begin{equation*}
    \mathcal{L} \left(\mat{U}, \mat{V}, \mat{\Psi}, \mat{\Phi}\right) = - \mat{U}^\trans \mat{\Sigma}_{\mat{xy}} \mat{V} + \inp{\mat{\Psi}}{\mat{U}^\trans \mat{\Sigma}_{\mat{xx}} \mat{U} - \mat{I}_R} + \inp{\mat{\Phi}}{\mat{V}^\trans \mat{\Sigma}_{\mat{yy}} \mat{V} - \mat{I}_R}
\end{equation*}
where 
$\mat{\Psi} \in \real^{R \times R}$ is a symmetric matrix of Lagrange multipliers for the $R(R+1)/2$ constraints on $\mat{U}$ in problem \eqref{equ:standard_SCCA_model_multiple_components_in_asymptotic_regime}, and $\mat{\Phi} \in \real^{R \times R}$ is a symmetric matrix of Lagrange multipliers for the $R(R+1)/2$ constraints on $\mat{V}$.

Denote the optimal primal and dual solutions of problem \eqref{equ:standard_SCCA_model_multiple_components_in_asymptotic_regime} as $\left(\mat{\hat{U}},\mat{\hat{V}}\right)$ and $\left(\mat{\hat{\Psi}},\mat{\hat{\Phi}}\right)$, respectively. According to the KKT conditions, we have
\begin{align}
2 \mat{\Sigma}_{\mat{xx}} \mat{\hat{U}} \mat{\hat{\Psi}} &= \mat{\Sigma}_{\mat{xy}} \mat{\hat{V}} \label{equ:standard_SCCA_in_asymptotic_regime_KKT_condition_stationarity_U} \\
2 \mat{\Sigma}_{\mat{yy}} \mat{\hat{V}} \mat{\hat{\Phi}} &= \mat{\Sigma}_{\mat{xy}}^\trans \mat{\hat{U}} \label{equ:standard_SCCA_in_asymptotic_regime_KKT_condition_stationarity_V}
\end{align}

Combining Eqs. \eqref{equ:standard_SCCA_in_asymptotic_regime_KKT_condition_stationarity_U}-\eqref{equ:standard_SCCA_in_asymptotic_regime_KKT_condition_stationarity_V} with the quadratic constraints in problem \eqref{equ:standard_SCCA_model_multiple_components_in_asymptotic_regime} yields 
\begin{align*}
    2\mat{\hat{\Psi}} &= \mat{\hat{U}}^\trans \mat{\Sigma}_{\mat{xy}} \mat{\hat{V}} \\
    2\mat{\hat{\Phi}} &= \mat{\hat{V}}^\trans \mat{\Sigma}_{\mat{xy}}^\trans \mat{\hat{U}}
\end{align*}

Note that problem \eqref{equ:standard_SCCA_model_multiple_components_in_asymptotic_regime} does not have a unique solution due to the rotational ambiguity: if $\left(\mat{\hat{U}},\mat{\hat{V}}\right)$ is an optimal solution of problem \eqref{equ:standard_SCCA_model_multiple_components_in_asymptotic_regime}, then $\left(\mat{\hat{\hat{U}}},\mat{\hat{\hat{V}}}\right) = \left(\mat{\hat{U}} \mat{Q},\mat{\hat{V}} \mat{Q}\right)$ for any orthogonal matrix $\mat{Q}  \in \real^{R \times R}$ is also an optimal solution.
\lsc{Since $\mat{\hat{\Psi}}$ and thus $\mat{\hat{U}}^\trans \mat{\Sigma}_{\mat{xy}} \mat{\hat{V}}$ is a symmetric matrix}, we can choose the optimal solution $\left(\mat{\hat{U}},\mat{\hat{V}}\right)$ for which $\mat{\hat{U}}^\trans \mat{\Sigma}_{\mat{xy}} \mat{\hat{V}}$ is a diagonal matrix. As a result, 
\begin{equation*}
    2\mat{\hat{\Psi}}=2\mat{\hat{\Phi}} \eqqcolon \mat{D}
\end{equation*}
is a diagonal matrix. Assuming both $\mat{\Sigma}_{\mat{xx}}$ and $\mat{\Sigma}_{\mat{yy}}$ are nonsingular, Eqs. \eqref{equ:standard_SCCA_in_asymptotic_regime_KKT_condition_stationarity_U}-\eqref{equ:standard_SCCA_in_asymptotic_regime_KKT_condition_stationarity_V} can be rewritten as
\begin{align}
\mat{\Sigma}_{\mat{xx}}^{1/2} \mat{\hat{U}} \mat{D} &= \mat{\Sigma}_{\mat{xx}}^{-1/2} \mat{\Sigma}_{\mat{xy}} \mat{\Sigma}_{\mat{yy}}^{-1/2}  \cdot \mat{\Sigma}_{\mat{yy}}^{1/2}\mat{\hat{V}} \label{equ:standard_SCCA_in_asymptotic_regime_KKT_condition_stationarity_U_2} \\
\mat{\Sigma}_{\mat{yy}}^{1/2} \mat{\hat{V}} \mat{D} &= \mat{\Sigma}_{\mat{yy}}^{-1/2} \mat{\Sigma}_{\mat{xy}}^\trans \mat{\Sigma}_{\mat{xx}}^{-1/2} \cdot \mat{\Sigma}_{\mat{xx}}^{1/2}\mat{\hat{U}} \label{equ:standard_SCCA_in_asymptotic_regime_KKT_condition_stationarity_V_2}
\end{align}
Note that the objective of problem \eqref{equ:standard_SCCA_model_multiple_components_in_asymptotic_regime} is to maximize $\traceof{\mat{D}}$ under the constraints that $\mat{\Sigma}_{\mat{xx}}^{1/2}\mat{U}$ and $\mat{\Sigma}_{\mat{yy}}^{1/2}\mat{V}$ both have orthonormal columns. It follows that 
$\mat{D}$ contains the $R$ largest singular values of $\mat{\Sigma}_{\mat{xx}}^{-1/2} \mat{\Sigma}_{\mat{xy}} \mat{\Sigma}_{\mat{yy}}^{-1/2}$, and
$\mat{\hat{E}} = \mat{\Sigma}_{\mat{xx}}^{1/2} \mat{\hat{U}}$ and $\mat{\hat{F}} = \mat{\Sigma}_{\mat{yy}}^{1/2} \mat{\hat{V}}$ contain the corresponding $R$ left and right singular vectors, respectively. \lsc{According to the Eckart-Young-Mirsky theorem \cite{eckart1936}, the columns of $\mat{\hat{U}}$ and $\mat{\hat{V}}$ can be obtained by successive rank-one SVDs of the residual covariance matrix.} Specifically, let $\mat{S}^0 = \mat{\Sigma}_{\mat{xx}}^{-1/2} \mat{\Sigma}_{\mat{xy}} \mat{\Sigma}_{\mat{yy}}^{-1/2} \in \real^{p \times q}$.
For $r=1,2,\dots,R$, we have
\begin{align}
    \left(\hat{d}_r, \mat{\hat{u}}_r,\mat{\hat{v}}_r\right) &= \argmin_{\substack{d_r, \mat{u}_r,\mat{v}_r \\ \twonorm{\mat{\Sigma}_{\mat{xx}}^{1/2} \mat{u}_r}=1 \\ \twonorm{\mat{\Sigma}_{\mat{yy}}^{1/2} \mat{v}_r}=1}} \normof{\mat{S}^{r-1} - \mat{\Sigma}_{\mat{xx}}^{1/2} d_r \mat{u}_r \mat{v}_r^\trans \mat{\Sigma}_{\mat{yy}}^{1/2}}{\rm F}^2 \label{equ:standard_SCCA_in_asymptotic_regime_estimate_duv} \\
    \mat{S}^r &= \mat{S}^{r-1} - \mat{\Sigma}_{\mat{xx}}^{1/2} \hat{d}_r \mat{\hat{u}}_r \mat{\hat{v}}_r^\trans \mat{\Sigma}_{\mat{yy}}^{1/2} \label{equ:standard_SCCA_in_asymptotic_regime_update_normalized_cov}
\end{align}
Suppose we have obtained the estimate of the $r$-th pair of canonical weight vectors $\left(\mat{\hat{u}}_r,\mat{\hat{v}}_r\right)$. 
We then estimate $d_r$ as
\begin{align*}
    \hat{d}_r = \argmin_{d_r} \normof{\mat{S}^{r-1} - \mat{\Sigma}_{\mat{xx}}^{1/2} d_r \mat{\hat{u}}_r \mat{\hat{v}}_r^\trans \mat{\Sigma}_{\mat{yy}}^{1/2}}{\rm F}^2 = \frac{\mat{\hat{u}}_r^\trans \mat{\Sigma}_{\mat{xx}}^{1/2} \mat{S}^{r-1} \mat{\Sigma}_{\mat{yy}}^{1/2} \mat{\hat{v}}_r}{\mat{\hat{u}}_r^\trans \mat{\Sigma}_{\mat{xx}} \mat{\hat{u}}_r \cdot \mat{\hat{v}}_r^\trans \mat{\Sigma}_{\mat{yy}} \mat{\hat{v}}_r}
\end{align*}

Taken all together, to compute multiple canonical components sequentially in the large-sample-size asymptotic regime, the residual covariance matrix is updated as below:
\begin{align}
    \mat{S}^0 &= \mat{\Sigma}_{\mat{xx}}^{-1/2} \mat{\Sigma}_{\mat{xy}} \mat{\Sigma}_{\mat{yy}}^{-1/2} \label{equ:standard_SCCA_in_asymptotic_regime_deflated_popu_cov_matrix_XY_normalized_iter0} \\
    \mat{S}^r &= \mat{S}^{r-1} - \frac{\mat{\Sigma}_{\mat{xx}}^{1/2} \mat{\hat{u}}_r \mat{\hat{u}}_r^\trans \mat{\Sigma}_{\mat{xx}}^{1/2} \mat{S}^{r-1} \mat{\Sigma}_{\mat{yy}}^{1/2} \mat{\hat{v}}_r \mat{\hat{v}}_r^\trans \mat{\Sigma}_{\mat{yy}}^{1/2}}{\mat{\hat{u}}_r^\trans \mat{\Sigma}_{\mat{xx}} \mat{\hat{u}}_r \cdot \mat{\hat{v}}_r^\trans \mat{\Sigma}_{\mat{yy}} \mat{\hat{v}}_r}, \; r=1,2,\dots,R \label{equ:standard_SCCA_in_asymptotic_regime_deflated_popu_cov_matrix_XY_normalized_iter_r}
\end{align}
or equivalently
\begin{align}
    \mat{\Sigma}_{\mat{xy}}^0 &= \mat{\Sigma}_{\mat{xy}} \label{equ:standard_SCCA_in_asymptotic_regime_deflated_popu_cov_matrix_XY_iter0} \\
    \mat{\Sigma}_{\mat{xy}}^r &= \mat{\Sigma}_{\mat{xy}}^{r-1} - \frac{\mat{\Sigma}_{\mat{xx}} \mat{\hat{u}}_r \mat{\hat{u}}_r^\trans \mat{\Sigma}_{\mat{xy}}^{r-1} \mat{\hat{v}}_r \mat{\hat{v}}_r^\trans \mat{\Sigma}_{\mat{yy}}}{\mat{\hat{u}}_r^\trans \mat{\Sigma}_{\mat{xx}} \mat{\hat{u}}_r \cdot \mat{\hat{v}}_r^\trans \mat{\Sigma}_{\mat{yy}} \mat{\hat{v}}_r}, \; r=1,2,\dots,R \label{equ:standard_SCCA_in_asymptotic_regime_deflated_popu_cov_matrix_XY_iter_r}
\end{align}
which results in Algorithm \ref{alg:standard_SCCA_for_multiple_canonical_components_in_asymptotic_regime_deflated_popu_cov_matrix_XY}.

Let $\mat{x} \in \real^{p \times 1}$ and $\mat{y} \in \real^{q \times 1}$ be random vectors generating the $\mat{X} \in \real^{n \times p}$ and $\mat{Y} \in \real^{n \times q}$, respectively. For notational simplicity, assume $\expvof{\mat{x}}=\mat{0}$, $\expvof{\mat{y}}=\mat{0}$. It can be shown that the residual covariance matrix update formulas \eqref{equ:standard_SCCA_in_asymptotic_regime_deflated_popu_cov_matrix_XY_normalized_iter0}-\eqref{equ:standard_SCCA_in_asymptotic_regime_deflated_popu_cov_matrix_XY_normalized_iter_r} can be rewritten in terms of \rev{random vectors $\mat{x}$ and $\mat{y}$} as
\begin{align}
    \mat{x}^0 &= \mat{x}, \; \mat{y}^0 = \mat{y} \label{equ:standard_SCCA_in_asymptotic_regime_deflated_random_vectors_XY_iter0} \\
    \mat{x}^r &= \mat{\Sigma}_{\mat{xx}}^{1/2} \left(\mat{I}_p - \frac{\mat{\Sigma}_{\mat{xx}}^{1/2} \mat{\hat{u}}_r \mat{\hat{u}}_r^\trans \mat{\Sigma}_{\mat{xx}}^{1/2}}{\mat{\hat{u}}_r^\trans \mat{\Sigma}_{\mat{xx}} \mat{\hat{u}}_r}\right) \mat{\Sigma}_{\mat{xx}}^{-1/2} \mat{x}^{r-1}
         = \mat{x}^{r-1} - \frac{\mat{\Sigma}_{\mat{xx}} \mat{\hat{u}}_r \mat{\hat{u}}_r^\trans}{\mat{\hat{u}}_r^\trans \mat{\Sigma}_{\mat{xx}} \mat{\hat{u}}_r} \mat{x}^{r-1}
         \label{equ:standard_SCCA_in_asymptotic_regime_deflated_random_vector_X_iter_r} \\
    \mat{y}^r &= \mat{\Sigma}_{\mat{yy}}^{1/2} \left(\mat{I}_q - \frac{\mat{\Sigma}_{\mat{yy}}^{1/2} \mat{\hat{v}}_r \mat{\hat{v}}_r^\trans \mat{\Sigma}_{\mat{yy}}^{1/2}}{\mat{\hat{v}}_r^\trans \mat{\Sigma}_{\mat{yy}} \mat{\hat{v}}_r}\right) \mat{\Sigma}_{\mat{yy}}^{-1/2} \mat{y}^{r-1}
         = \mat{y}^{r-1} - \frac{\mat{\Sigma}_{\mat{yy}} \mat{\hat{v}}_r \mat{\hat{v}}_r^\trans}{\mat{\hat{v}}_r^\trans \mat{\Sigma}_{\mat{yy}} \mat{\hat{v}}_r} \mat{y}^{r-1}
    \label{equ:standard_SCCA_in_asymptotic_regime_deflated_random_vector_Y_iter_r}
\end{align}
which results in Algorithm \ref{alg:standard_SCCA_for_multiple_canonical_components_in_asymptotic_regime_deflated_random vectors}.

\clearpage

\begin{algorithm}
	\caption{Sequential computation of $R$ canonical components of SCCA in asymptotic regime via deflation of the population cross-covariance matrix.}
	\label{alg:standard_SCCA_for_multiple_canonical_components_in_asymptotic_regime_deflated_popu_cov_matrix_XY}
	\begin{algorithmic}[1]
\STATE $\mat{\Sigma}_{\mat{xy}}^0 = \expvof{\mat{x} \mat{y}^\trans}$, $\mat{\Sigma}_{\mat{xx}} = \expvof{\mat{x} \mat{x}^\trans}$ and $\mat{\Sigma}_{\mat{yy}} = \expvof{\mat{y} \mat{y}^\trans}$.
\FOR{$r=1,2,\dots,R$} 
\STATE Find the estimate of the $r$-th pair of canonical weight vectors $\mat{\hat{u}}_r$ and $\mat{\hat{v}}_r$:
\begin{equation*}
\begin{aligned}
& \underset{\mat{u}_r,\mat{v}_r}{\text{maximize}} & & \mat{u}_r^\trans \mat{\Sigma}_{\mat{xy}}^{r-1} \mat{v}_r \\
&                           \text{subject to} & & \mat{u}_r^\trans \mat{\Sigma}_{\mat{xx}} \mat{u}_r = 1 \\
&                                             & & \mat{v}_r^\trans \mat{\Sigma}_{\mat{yy}} \mat{v}_r = 1 \end{aligned}
\end{equation*}
\STATE $\mat{\Sigma}_{\mat{xy}}^r \gets \mat{\Sigma}_{\mat{xy}}^{r-1} - \mat{\Sigma}_{\mat{xx}} \hat{d}_r \mat{\hat{u}}_r \mat{\hat{v}}_r^\trans \mat{\Sigma}_{\mat{yy}}$, where $\hat{d}_r = \frac{\mat{\hat{u}}_r^\trans \mat{\Sigma}_{\mat{xy}}^{r-1} \mat{\hat{v}}_r}{\mat{\hat{u}}_r^\trans \mat{\Sigma}_{\mat{xx}} \mat{\hat{u}}_r \cdot \mat{\hat{v}}_r^\trans \mat{\Sigma}_{\mat{yy}} \mat{\hat{v}}_r}$.
\ENDFOR
	\end{algorithmic}
\end{algorithm}

\begin{algorithm}
	\caption{Sequential computation of $R$ canonical components of SCCA in asymptotic regime via deflation of random vectors.}
	\label{alg:standard_SCCA_for_multiple_canonical_components_in_asymptotic_regime_deflated_random vectors}
	\begin{algorithmic}[1]
\STATE Let $\mat{x}^0 = \mat{x} \in \real^{p \times 1}$, $\mat{y}^0 = \mat{y} \in \real^{q \times 1}$.
\FOR{$r=1,2,\dots,R$} 
\STATE Find the estimate of the $r$-th pair of canonical weight vectors $\mat{\hat{u}}_r$ and $\mat{\hat{v}}_r$:
\begin{equation*}
\begin{aligned}
& \underset{\mat{u}_r,\mat{v}_r}{\text{maximize}} & & \mat{u}_r^\trans \expvof{\mat{x}^{r-1}{\mat{y}^{r-1}}^\trans} \mat{v}_r \\
&                           \text{subject to} & & \mat{u}_r^\trans \expvof{\mat{x} \mat{x}^\trans} \mat{u}_r = 1 \\
&                                             & & \mat{v}_r^\trans \expvof{\mat{y} \mat{y}^\trans} \mat{v}_r = 1
\end{aligned}
\end{equation*}
\STATE Calculate the residual random vectors:
\begin{align*}
    \mat{x}^r &\gets \mat{x}^{r-1} - \frac{\mat{\Sigma}_{\mat{xx}} \mat{\hat{u}}_r \mat{\hat{u}}_r^\trans}{\mat{\hat{u}}_r^\trans \mat{\Sigma}_{\mat{xx}} \mat{\hat{u}}_r} \mat{x}^{r-1} \\ 
    \mat{y}^r &\gets \mat{y}^{r-1} - \frac{\mat{\Sigma}_{\mat{yy}} \mat{\hat{v}}_r \mat{\hat{v}}_r^\trans}{\mat{\hat{v}}_r^\trans \mat{\Sigma}_{\mat{yy}} \mat{\hat{v}}_r} \mat{y}^{r-1}  
\end{align*}
\ENDFOR
	\end{algorithmic}
\end{algorithm}


The Algorithms \ref{alg:standard_SCCA_for_multiple_canonical_components_deflated_sample_cov_matrix_XY} and \ref{alg:standard_SCCA_for_multiple_canonical_components_deflated_data} are implementations of Algorithms \ref{alg:standard_SCCA_for_multiple_canonical_components_in_asymptotic_regime_deflated_popu_cov_matrix_XY} and \ref{alg:standard_SCCA_for_multiple_canonical_components_in_asymptotic_regime_deflated_random vectors} in finite-sample settings, respectively.

\clearpage

\subsection{Sequential calculation of multiple canonical components of simplified SCCA}
\label{subsec_in_supp:proof_cal_multiple_canonical_components_for_simplified_SCCA}

The simplified SCCA model for computing $R$ canonical components is
\begin{equation}\label{equ:simplified_SCCA_model_multiple_components}
\begin{aligned}
& \underset{\mat{U},\mat{V}}{\text{maximize}} & & \traceof{\mat{U}^\trans \mat{\hat{\Sigma}}_{\mat{xy}} \mat{V}} \\
&                           \text{subject to} & & \mat{U}^\trans \mat{U} = \mat{I}_R, \normof{\mat{u}_r}{1} \leq c_{1r}, r=1,2,\dots,R\\
&                                             & & \mat{V}^\trans \mat{V} = \mat{I}_R, \normof{\mat{v}_r}{1} \leq c_{2r}, r=1,2,\dots,R
\end{aligned}
\end{equation}
where $\mat{\hat{\Sigma}}_{\mat{xy}}$ is the sample cross-covariance matrix between \rev{random vectors $\mat{x}$ and $\mat{y}$}.

For clarity, we first present two algorithms (Algorithms \ref{alg:simplified_SCCA_for_multiple_canonical_components_deflated_sample_cov_matrix_XY} and \ref{alg:simplified_SCCA_for_multiple_canonical_components_deflated_data}) to sequentially compute multiple canonical components of simplified SCCA: one is based on deflation of the cross-covariance matrix, and the other one is based on deflation of the data matrices. Then we provide theoretical explanations of both algorithms in the subsequent sections.

\begin{algorithm}
	\caption{Sequential computation of $R$ canonical components of simplified SCCA via deflation of the cross-covariance matrix.}
	\label{alg:simplified_SCCA_for_multiple_canonical_components_deflated_sample_cov_matrix_XY}
	\begin{algorithmic}[1]
\STATE Let $\mat{\hat{\Sigma}}_{\mat{xy}}^0 = \frac{1}{n-1} \mat{X}^\trans \mat{Y} \in \real^{p \times q}$.
\FOR{$r=1,2,\dots,R$} 
\STATE Find the $r$-th pair of canonical weight vectors $\mat{\hat{u}}_r$ and $\mat{\hat{v}}_r$ by applying Algorithm \ref{alg:simplified_SCCA} to $\mat{\hat{\Sigma}}_{\mat{xy}}^{r-1}$:
\begin{equation*}
\begin{aligned}
& \underset{\mat{u}_r,\mat{v}_r}{\text{maximize}} & & \mat{u}_r^\trans \mat{\hat{\Sigma}}_{\mat{xy}}^{r-1} \mat{v}_r \\
&                           \text{subject to} & & \twonorm{\mat{u}_r}^2 \leq 1, \normof{\mat{u}_r}{1} \leq c_{1r} \\
&                                             & & \twonorm{\mat{v}_r}^2 \leq 1, \normof{\mat{v}_r}{1} \leq c_{2r}
\end{aligned}
\end{equation*}
\STATE $\mat{\hat{\Sigma}}_{\mat{xy}}^r \gets \mat{\hat{\Sigma}}_{\mat{xy}}^{r-1} - \hat{d}_r \mat{\hat{u}}_r \mat{\hat{v}}_r^\trans$, where $\hat{d}_r = \frac{\mat{\hat{u}}_r^\trans \mat{\hat{\Sigma}}_{\mat{xy}}^{r-1} \mat{\hat{v}}_r}{\twonorm{\mat{\hat{u}}_r}^2 \cdot \twonorm{\mat{\hat{v}}_r}^2}$.
\ENDFOR
	\end{algorithmic}
\end{algorithm}

\begin{algorithm}
	\caption{Sequential computation of $R$ canonical components of simplified SCCA via deflation of the data matrices.}
	\label{alg:simplified_SCCA_for_multiple_canonical_components_deflated_data}
	\begin{algorithmic}[1]
\STATE Let $\mat{X}^0 = \mat{X} \in \real^{n \times p}$, $\mat{Y}^0 = \mat{Y} \in \real^{n \times q}$.
\FOR{$r=1,2,\dots,R$} 
\STATE Find the $r$-th pair of canonical weight vectors $\left(\mat{\hat{u}}_r,\mat{\hat{v}}_r\right)$ by applying Algorithm \ref{alg:simplified_SCCA}:
\begin{equation*}
\begin{aligned}
& \underset{\mat{u},\mat{v}}{\text{maximize}} & & \frac{1}{n-1} \mat{u}_r^\trans {\mat{X}^{r-1}}^\trans \mat{Y}^{r-1} \mat{v}_r \\
& \text{subject to}                           & & \twonorm{\mat{u}_r}^2 \leq 1, \normof{\mat{u}}{1} \leq c_{1r} \\
&                                             & & \twonorm{\mat{v}_r}^2 \leq 1, \normof{\mat{v}}{1} \leq c_{2r} \\
\end{aligned}
\end{equation*}
\STATE Calculate the residual data:
\begin{align}
    \mat{X}^r &\gets \mat{X}^{r-1} \left(\mat{I}_p - \frac{\mat{\hat{u}}_r \mat{\hat{u}}_r^\trans}{\twonorm{\mat{\hat{u}}_r}^2}\right) \label{equ:simplified_SCCA_in_finite_sample_setting_deflated_data_matrix_X_iter_r} \\
    \mat{Y}^r &\gets \mat{Y}^{r-1} \left(\mat{I}_q - \frac{\mat{\hat{v}}_r \mat{\hat{v}}_r^\trans}{\twonorm{\mat{\hat{v}}_r}^2}\right) \label{equ:simplified_SCCA_in_finite_sample_setting_deflated_data_matrix_Y_iter_r}
\end{align}
\ENDFOR
	\end{algorithmic}
\end{algorithm}

\begin{remark}
The deflated data in Eqs. \eqref{equ:simplified_SCCA_in_finite_sample_setting_deflated_data_matrix_X_iter_r}-\eqref{equ:simplified_SCCA_in_finite_sample_setting_deflated_data_matrix_Y_iter_r} can also be interpreted as the residual matrix of linear least squares regression: $\underset{\mat{z} \in \real^n}{\text{minimize}} \; \normof{\mat{X}^{r-1} - \mat{z} \cdot \mat{\hat{u}}_r^\trans}{\rm F}^2$ and $\underset{\mat{\zeta} \in \real^n}{\text{minimize}} \; \normof{\mat{Y}^{r-1} - \mat{\zeta} \cdot \mat{\hat{v}}_r^\trans}{\rm F}^2$, respectively.
\end{remark}


\subsubsection{Sequential calculation of multiple SCCA canonical components in the large-sample-size asymptotic regime}
\label{subsubsec_in_supp:proof_cal_multiple_canonical_components_for_simplified_SCCA}
To compute $R$ canonical components sequentially/greedily, we consider the asymptotic regime of $n\to\infty$ in which case model \eqref{equ:simplified_SCCA_model_multiple_components} becomes
\begin{equation}\label{equ:simplified_SCCA_model_multiple_components_in_asymptotic_regime}
\begin{aligned}
& \underset{\mat{U},\mat{V}}{\text{maximize}} & & \traceof{\mat{U}^\trans \mat{\Sigma}_{\mat{xy}} \mat{V}} \\
& \text{subject to}                           & & \mat{U}^\trans \mat{U} = \mat{I}_R \\
&                                             & & \mat{V}^\trans \mat{V} = \mat{I}_R \\
\end{aligned}
\end{equation}
where $\mat{\Sigma}_{\mat{xy}}$ is the population cross-covariance matrix between \rev{random vectors $\mat{x}$ and $\mat{y}$}. \lsc{Note that in model \eqref{equ:simplified_SCCA_model_multiple_components_in_asymptotic_regime} we have dropped the L1 regularizers: since we have infinite amount of data available for use, the L1 regularizations are no longer necessary.}

The Lagrangian function of problem \eqref{equ:simplified_SCCA_model_multiple_components_in_asymptotic_regime} is defined as
\begin{equation*}
    \mathcal{L} \left(\mat{U}, \mat{V}, \mat{\Psi}, \mat{\Phi}\right) = - \mat{U}^\trans \mat{\Sigma}_{\mat{xy}} \mat{V} + \inp{\mat{\Psi}}{\mat{U}^\trans \mat{U} - \mat{I}_R} + \inp{\mat{\Phi}}{\mat{V}^\trans \mat{V} - \mat{I}_R}
\end{equation*}
where 
$\mat{\Psi} \in \real^{R \times R}$ is a symmetric matrix of Lagrange multipliers for the $R(R+1)/2$ constraints on $\mat{U}$ in problem \eqref{equ:simplified_SCCA_model_multiple_components_in_asymptotic_regime}, and $\mat{\Phi} \in \real^{R \times R}$ is a symmetric matrix of Lagrange multipliers for the $R(R+1)/2$ constraints on $\mat{V}$.

Denote the optimal primal and dual solutions of problem \eqref{equ:simplified_SCCA_model_multiple_components_in_asymptotic_regime} as $\left(\mat{\hat{U}},\mat{\hat{V}}\right)$ and $\left(\mat{\hat{\Psi}},\mat{\hat{\Phi}}\right)$, respectively. According to the KKT conditions, we have
\begin{align}
2 \mat{\hat{U}} \mat{\hat{\Psi}} &= \mat{\Sigma}_{\mat{xy}} \mat{\hat{V}} \label{equ:simplified_SCCA_in_asymptotic_regime_KKT_condition_stationarity_U} \\
2 \mat{\hat{V}} \mat{\hat{\Phi}} &= \mat{\Sigma}_{\mat{xy}}^\trans \mat{\hat{U}} \label{equ:simplified_SCCA_in_asymptotic_regime_KKT_condition_stationarity_V}
\end{align}

Combining Eqs. \eqref{equ:simplified_SCCA_in_asymptotic_regime_KKT_condition_stationarity_U}-\eqref{equ:simplified_SCCA_in_asymptotic_regime_KKT_condition_stationarity_V} with the quadratic constraints in problem \eqref{equ:simplified_SCCA_model_multiple_components_in_asymptotic_regime} yields 
\begin{align*}
    2\mat{\hat{\Psi}} &= \mat{\hat{U}}^\trans \mat{\Sigma}_{\mat{xy}} \mat{\hat{V}} \\
    2\mat{\hat{\Phi}} &= \mat{\hat{V}}^\trans \mat{\Sigma}_{\mat{xy}}^\trans \mat{\hat{U}}
\end{align*}

Note that problem \eqref{equ:simplified_SCCA_model_multiple_components_in_asymptotic_regime} does not have a unique solution due to the rotational ambiguity: if $\left(\mat{\hat{U}},\mat{\hat{V}}\right)$ is an optimal solution of problem \eqref{equ:simplified_SCCA_model_multiple_components_in_asymptotic_regime}, then $\left(\mat{\hat{\hat{U}}},\mat{\hat{\hat{V}}}\right) = \left(\mat{\hat{U}} \mat{Q},\mat{\hat{V}} \mat{Q}\right)$ for any orthogonal matrix $\mat{Q}  \in \real^{R \times R}$ is also an optimal solution.
\lsc{Since $\mat{\hat{\Psi}}$ and thus $\mat{\hat{U}}^\trans \mat{\Sigma}_{\mat{xy}} \mat{\hat{V}}$ is a symmetric matrix}, we can choose the optimal solution $\left(\mat{\hat{U}},\mat{\hat{V}}\right)$ for which $\mat{\hat{U}}^\trans \mat{\Sigma}_{\mat{xy}} \mat{\hat{V}}$ is a diagonal matrix. As a result, 
\begin{equation*}
    2\mat{\hat{\Psi}}=2\mat{\hat{\Phi}} \eqqcolon \mat{D}
\end{equation*}
is a diagonal matrix. Assuming both $\mat{\Sigma}_{\mat{xx}}$ and $\mat{\Sigma}_{\mat{yy}}$ are nonsingular, Eqs. \eqref{equ:simplified_SCCA_in_asymptotic_regime_KKT_condition_stationarity_U}-\eqref{equ:simplified_SCCA_in_asymptotic_regime_KKT_condition_stationarity_V} can be rewritten as
\begin{align}
 \mat{\hat{U}} \mat{D} &= \mat{\Sigma}_{\mat{xy}} \cdot \mat{\hat{V}} \label{equ:simplified_SCCA_in_asymptotic_regime_KKT_condition_stationarity_U_2} \\
 \mat{\hat{V}} \mat{D} &= \mat{\Sigma}_{\mat{xy}}^\trans \cdot \mat{\hat{U}} \label{equ:simplified_SCCA_in_asymptotic_regime_KKT_condition_stationarity_V_2}
\end{align}
Note that the objective of problem \eqref{equ:simplified_SCCA_model_multiple_components_in_asymptotic_regime} is to maximize $\traceof{\mat{D}}$ under the constraints that $\mat{U}$ and $\mat{V}$ both have orthonormal columns. It follows that 
$\mat{D}$ contains the $R$ largest singular values of $\mat{\Sigma}_{\mat{xy}}$, and $\mat{\hat{U}}$ and $\mat{\hat{V}}$ contain the corresponding $R$ left and right singular vectors, respectively. \lsc{According to the Eckart-Young-Mirsky theorem \cite{eckart1936}, the columns of $\mat{\hat{U}}$ and $\mat{\hat{V}}$ can be obtained by successive rank-one SVDs of the residual covariance matrix.} Specifically, let $\mat{\Sigma}_{\mat{xy}}^0 = \mat{\Sigma}_{\mat{xy}} \in \real^{p \times q}$.
For $r=1,2,\dots,R$, we have
\begin{align}
    \left(\hat{d}_r, \mat{\hat{u}}_r,\mat{\hat{v}}_r\right) &= \argmin_{\substack{d_r, \mat{u}_r,\mat{v}_r \\ \twonorm{ \mat{u}_r}=1 \\ \twonorm{ \mat{v}_r}=1}} \normof{\mat{\Sigma}_{\mat{xy}}^{r-1} -  d_r \mat{u}_r \mat{v}_r^\trans }{\rm F}^2 \label{equ:simplified_SCCA_in_asymptotic_regime_estimate_duv} \\
    \mat{\Sigma}_{\mat{xy}}^r &= \mat{\Sigma}_{\mat{xy}}^{r-1} -  \hat{d}_r \mat{\hat{u}}_r \mat{\hat{v}}_r^\trans  \label{equ:simplified_SCCA_in_asymptotic_regime_update_popu_cov_matrix_XY}
\end{align}
Suppose we have obtained the estimate of the $r$-th pair of canonical weight vectors $\left(\mat{\hat{u}}_r,\mat{\hat{v}}_r\right)$. 
We then estimate $d_r$ as
\begin{align*}
    \hat{d}_r = \argmin_{d_r} \normof{\mat{\Sigma}_{\mat{xy}}^{r-1} -  d_r \mat{\hat{u}}_r \mat{\hat{v}}_r^\trans }{\rm F}^2 = \frac{\mat{\hat{u}}_r^\trans  \mat{\Sigma}_{\mat{xy}}^{r-1}  \mat{\hat{v}}_r}{\twonorm{\mat{\hat{u}}_r}^2 \cdot \twonorm{\mat{\hat{v}}_r}^2}
\end{align*}

Taken all together, to compute multiple canonical components sequentially in the large-sample-size asymptotic regime, the residual covariance matrix is updated as below:
\begin{align}
    \mat{\Sigma}_{\mat{xy}}^0 &= \mat{\Sigma}_{\mat{xy}} \label{equ:simplified_SCCA_in_asymptotic_regime_deflated_popu_cov_matrix_XY_iter0} \\
    \mat{\Sigma}_{\mat{xy}}^r &= \mat{\Sigma}_{\mat{xy}}^{r-1} - \frac{\mat{\hat{u}}_r \mat{\hat{u}}_r^\trans \mat{\Sigma}_{\mat{xy}}^{r-1} \mat{\hat{v}}_r \mat{\hat{v}}_r^\trans}{\twonorm{\mat{\hat{u}}_r}^2 \cdot \twonorm{\mat{\hat{v}}_r}^2}, \; r=1,2,\dots,R \label{equ:simplified_SCCA_in_asymptotic_regime_deflated_popu_cov_matrix_XY_iter_r}
\end{align}
This results in Algorithm \ref{alg:simplified_SCCA_for_multiple_canonical_components_in_asymptotic_regime_deflated_popu_cov_matrix_XY}.

For notational simplicity, assume $\expvof{\mat{x}}=\mat{0}$, $\expvof{\mat{y}}=\mat{0}$. It can be shown that the residual covariance matrix update formulas \eqref{equ:simplified_SCCA_in_asymptotic_regime_deflated_popu_cov_matrix_XY_iter0}-\eqref{equ:simplified_SCCA_in_asymptotic_regime_deflated_popu_cov_matrix_XY_iter_r} can be rewritten in terms of \rev{random vectors $\mat{x}$ and $\mat{y}$} as
\begin{align}
    \mat{x}^0 &= \mat{x}, \; \mat{y}^0 = \mat{y} \label{equ:simplified_SCCA_in_asymptotic_regime_deflated_random_vectors_XY_iter0} \\
    \mat{x}^r &= \left(\mat{I}_p - \frac{\mat{\hat{u}}_r \mat{\hat{u}}_r^\trans}{\twonorm{\mat{\hat{u}}_r}^2}\right) \mat{x}^{r-1}
         \label{equ:simplified_SCCA_in_asymptotic_regime_deflated_random_vector_X_iter_r} \\
    \mat{y}^r &= \left(\mat{I}_q - \frac{\mat{\hat{v}}_r \mat{\hat{v}}_r^\trans}{\twonorm{\mat{\hat{v}}_r}^2}\right) \mat{y}^{r-1}
    \label{equ:simplified_SCCA_in_asymptotic_regime_deflated_random_vector_Y_iter_r}
\end{align}
which results in Algorithm \ref{alg:simplified_SCCA_for_multiple_canonical_components_in_asymptotic_regime_deflated_random vectors}.

\clearpage

\begin{algorithm}
	\caption{Sequential computation of $R$ canonical components of simplified SCCA in asymptotic regime via deflation of the population cross-covariance matrix.}
	\label{alg:simplified_SCCA_for_multiple_canonical_components_in_asymptotic_regime_deflated_popu_cov_matrix_XY}
	\begin{algorithmic}[1]
\STATE Let $\mat{\Sigma}_{\mat{xy}}^0 = \expvof{\mat{x} \mat{y}^\trans}$.
\FOR{$r=1,2,\dots,R$} 
\STATE Solve for the $r$-th pair of canonical weight vectors $\mat{\hat{u}}_r$ and $\mat{\hat{v}}_r$:
\begin{equation*}
\begin{aligned}
& \underset{\mat{u}_r,\mat{v}_r}{\text{maximize}} & & \mat{u}_r^\trans \mat{\Sigma}_{\mat{xy}}^{r-1} \mat{v}_r \\
&                           \text{subject to} & & \twonorm{\mat{u}_r}^2 = 1 \\
&                                             & & \twonorm{\mat{v}_r}^2 = 1 \end{aligned}
\end{equation*}
\STATE $\mat{\Sigma}_{\mat{xy}}^r \gets \mat{\Sigma}_{\mat{xy}}^{r-1} - \hat{d}_r \mat{\hat{u}}_r \mat{\hat{v}}_r^\trans$, where $\hat{d}_r = \frac{\mat{\hat{u}}_r^\trans \mat{\Sigma}_{\mat{xy}}^{r-1} \mat{\hat{v}}_r}{\twonorm{\mat{\hat{u}}_r}^2 \cdot \twonorm{\mat{\hat{v}}_r}^2}$.
\ENDFOR
	\end{algorithmic}
\end{algorithm}

\begin{algorithm}
	\caption{Sequential computation of $R$ canonical components of simplified SCCA in asymptotic regime via deflation of random vectors.}
	\label{alg:simplified_SCCA_for_multiple_canonical_components_in_asymptotic_regime_deflated_random vectors}
	\begin{algorithmic}[1]
\STATE Let $\mat{x}^0 = \mat{x} \in \real^{p \times 1}$, $\mat{y}^0 = \mat{y} \in \real^{q \times 1}$.
\FOR{$r=1,2,\dots,R$} 
\STATE Solve for the $r$-th pair of canonical weight vectors $\mat{\hat{u}}_r$ and $\mat{\hat{v}}_r$:
\begin{equation*}
\begin{aligned}
& \underset{\mat{u}_r,\mat{v}_r}{\text{maximize}} & & \mat{u}_r^\trans \expvof{\mat{x}^{r-1}{\mat{y}^{r-1}}^\trans} \mat{v}_r \\
&                           \text{subject to} & & \twonorm{\mat{u}_r}^2 = 1 \\
&                                             & & \twonorm{\mat{v}_r}^2 = 1 \end{aligned}
\end{equation*}
\STATE Calculate the residual random vectors:
\begin{align*}
    \mat{x}^r &\gets \left(\mat{I}_p - \frac{\mat{\hat{u}}_r \mat{\hat{u}}_r^\trans}{\twonorm{\mat{\hat{u}}_r}^2}\right) \mat{x}^{r-1} \\
    \mat{y}^r &\gets \left(\mat{I}_q - \frac{\mat{\hat{v}}_r \mat{\hat{v}}_r^\trans}{\twonorm{\mat{\hat{v}}_r}^2}\right) \mat{y}^{r-1}
\end{align*}
\ENDFOR
	\end{algorithmic}
\end{algorithm}

In finite-sample settings, the covariance matrix deflation based Algorithm \ref{alg:simplified_SCCA_for_multiple_canonical_components_in_asymptotic_regime_deflated_popu_cov_matrix_XY} becomes Algorithm \ref{alg:simplified_SCCA_for_multiple_canonical_components_deflated_sample_cov_matrix_XY} to sequentially compute $R$ canonical components of simplified SCCA, while the random vector deflation based Algorithm \ref{alg:simplified_SCCA_for_multiple_canonical_components_in_asymptotic_regime_deflated_random vectors} becomes Algorithm \ref{alg:simplified_SCCA_for_multiple_canonical_components_deflated_data}.

\section{Supporting Information and additional results for simulation study on synthetic data}
\label{sec_in_supp:simulation}

\subsection{Covariance structure of the synthetic data}
\label{subsec_in_supp:cov_structure_synthetic_data}

The sample cross- and auto-covariance matrices among \rev{random vectors $\mat{x}$ and $\mat{y}$} are defined as
\begin{align}
    \mat{\hat{\Sigma}}_{\mat{xy}} &= \frac{1}{n} \mat{X}^\trans \mat{Y} \label{equ:sample_cov_matrix_XY} \\
    \mat{\hat{\Sigma}}_{\mat{xx}} &= \frac{1}{n} \mat{X}^\trans \mat{X} \label{equ:sample_cov_matrix_X} \\
    \mat{\hat{\Sigma}}_{\mat{yy}} &= \frac{1}{n} \mat{Y}^\trans \mat{Y} \label{equ:sample_cov_matrix_Y}
\end{align}

\subsubsection{Experimental setup 1: uncorrelated variables}
\label{subsubsec_in_supp:uncorrelated_variables}
The population cross- and auto-covariance matrices among \rev{random vectors $\mat{x}$ and $\mat{y}$} are
\begin{equation*}
    \expvof{\mat{x}} = \mat{0}, \quad \expvof{\mat{y}} = \mat{0}
\end{equation*}
\begin{align}
    \mat{\Sigma}_{\mat{xx}} &= \expvof{\mat{x} \mat{x}^\trans} = \mat{I}_p \label{equ:uncorrelated_variables_population_cov_matrix_X} \\
    \mat{\Sigma}_{\mat{yy}} &= \expvof{\mat{y} \mat{y}^\trans} = \twonorm{\mat{c}}^2 \mat{d} \mat{d}^\trans + \sigma^2 \mat{I}_q \label{equ:uncorrelated_variables_population_cov_matrix_Y} \\
    \mat{\Sigma}_{\mat{xy}} &= \expvof{\mat{x} \mat{y}^\trans} = \mat{c} \mat{d}^\trans \label{equ:uncorrelated_variables_population_cov_matrix_XY}
\end{align}

\begin{figure}[tb]
\begin{center}
\includegraphics[width=0.4\linewidth]{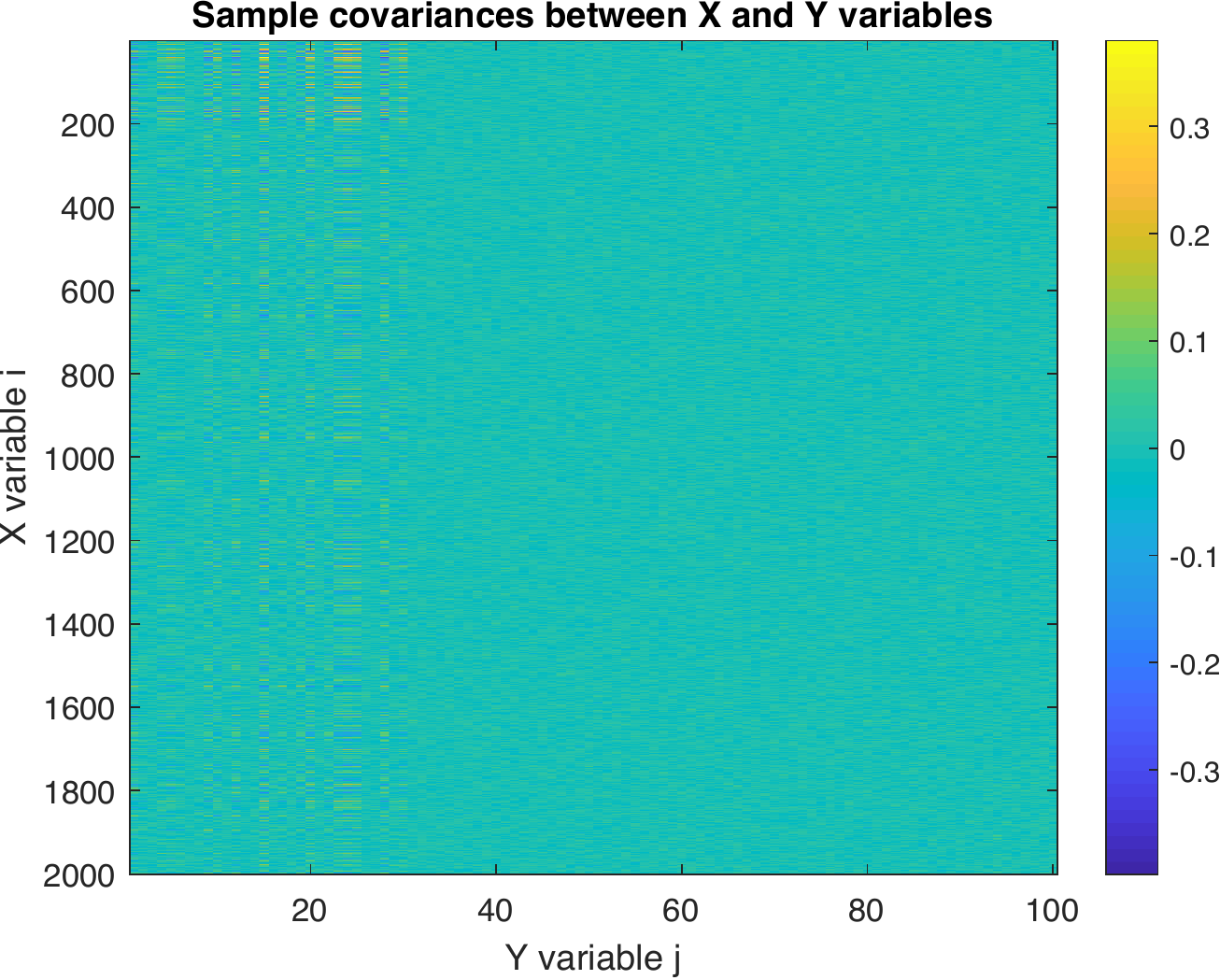}
\hspace{1cm}
\includegraphics[width=0.4\linewidth]{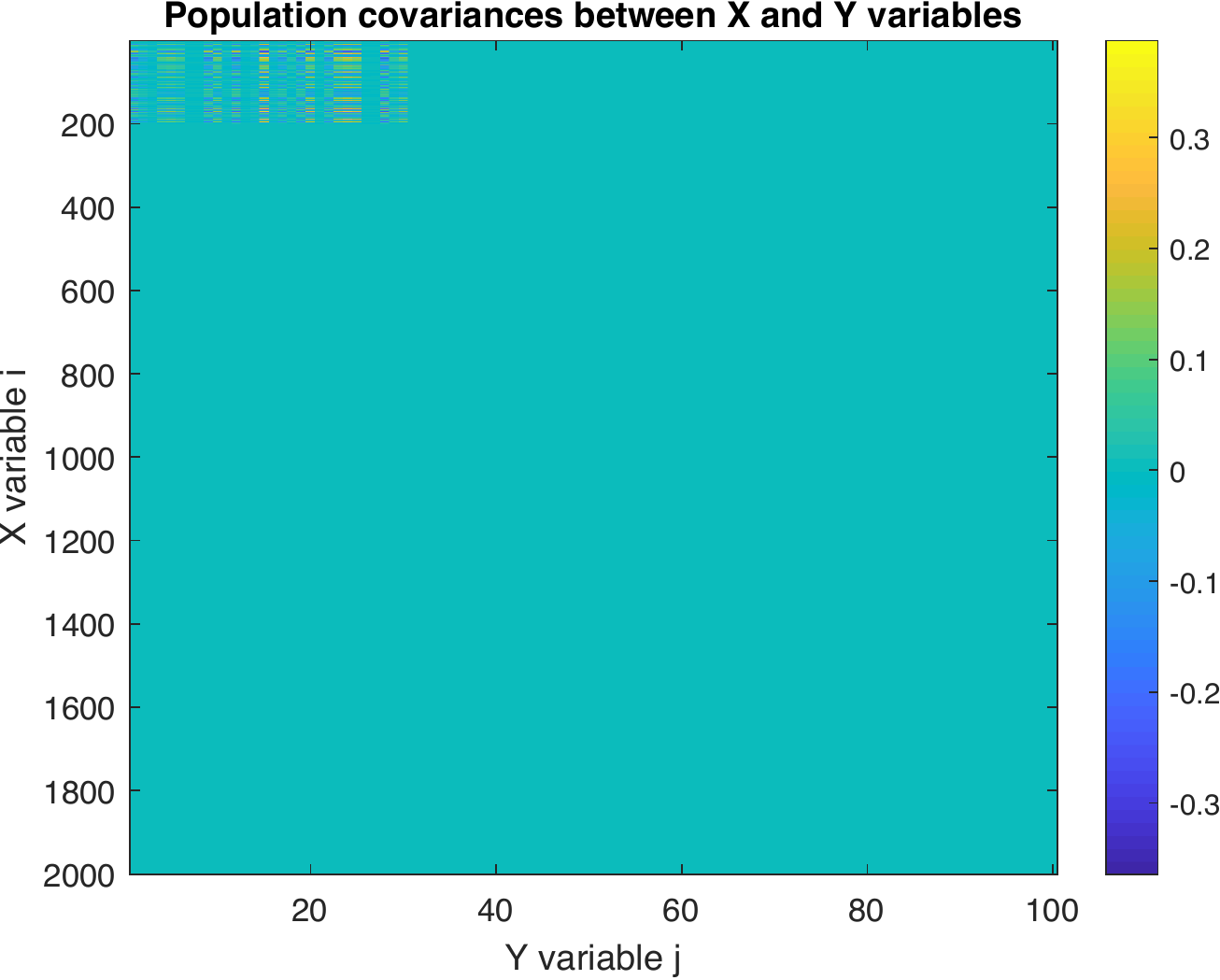} \\
\vspace{0.5cm}
\includegraphics[width=0.4\linewidth]{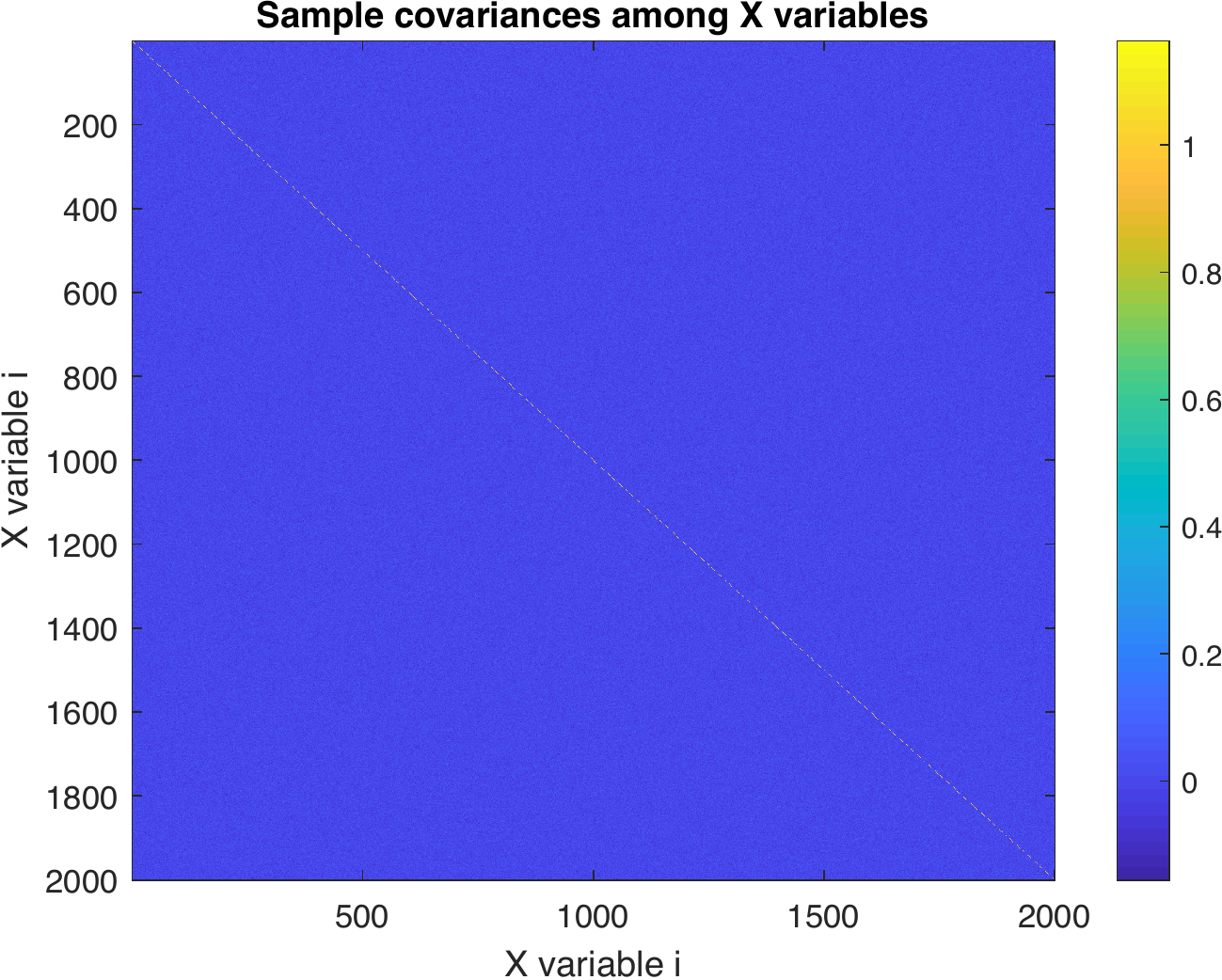}
\hspace{1cm}
\includegraphics[width=0.4\linewidth]{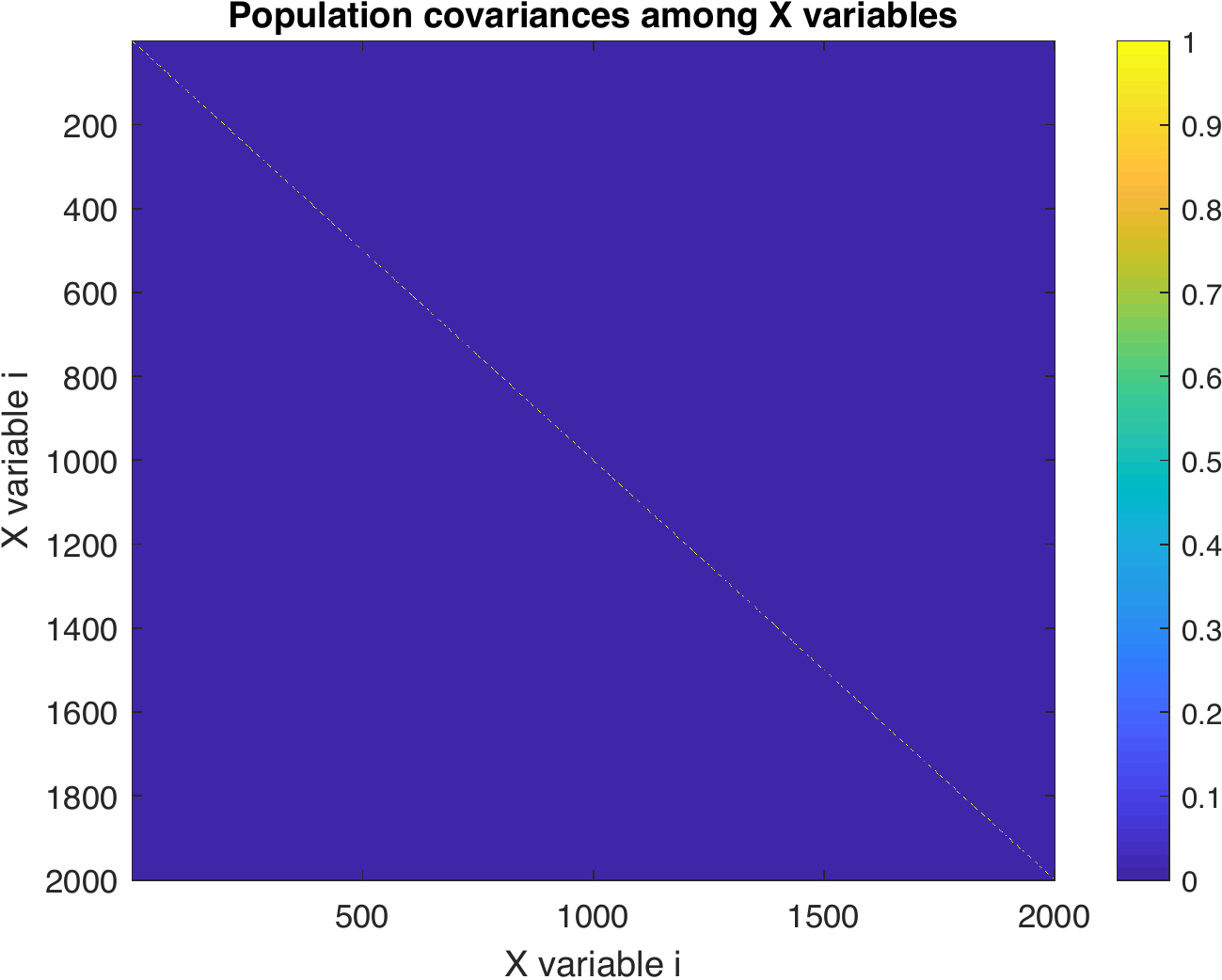} \\
\vspace{0.5cm}
\includegraphics[width=0.4\linewidth]{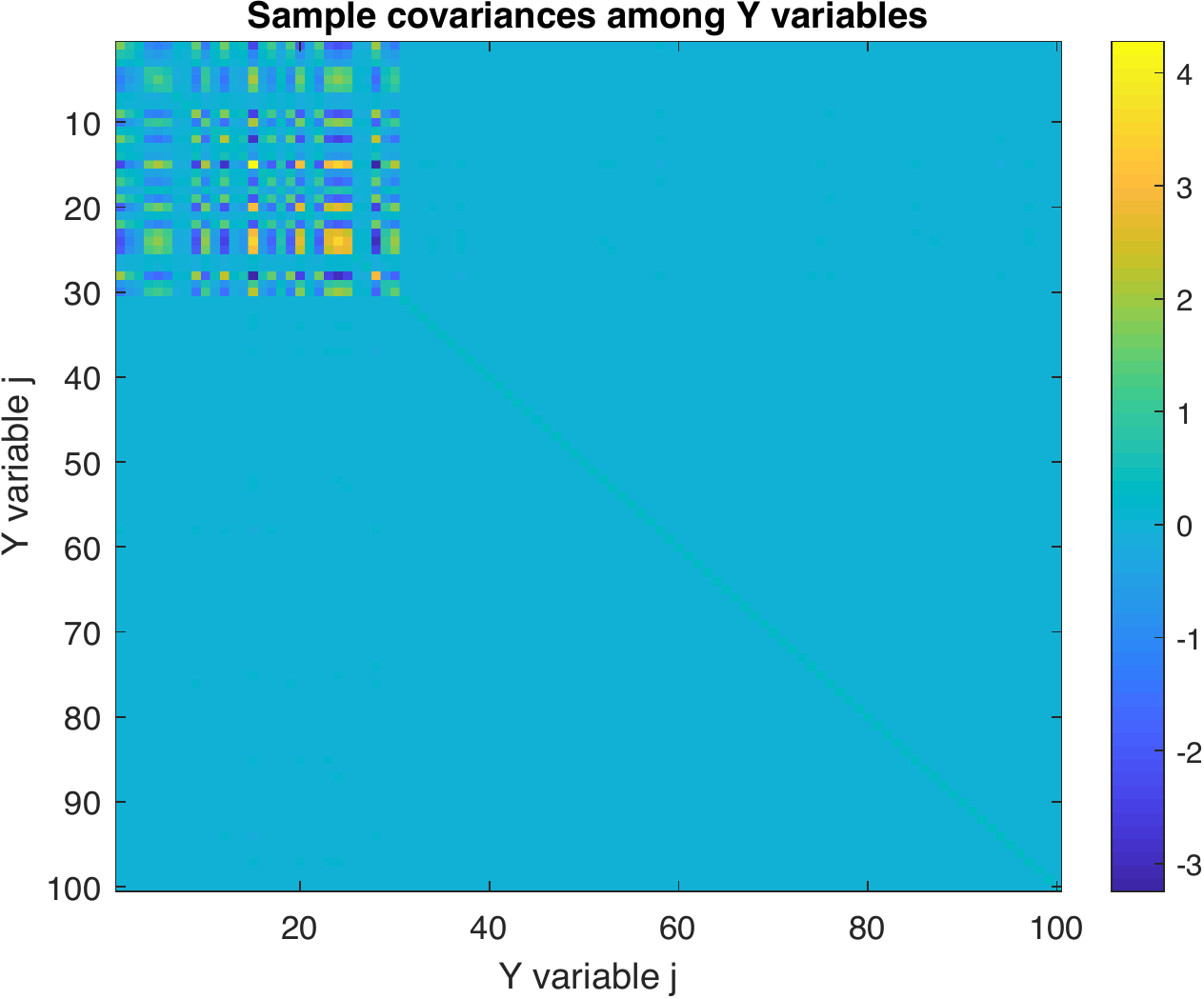}
\hspace{1cm}
\includegraphics[width=0.4\linewidth]{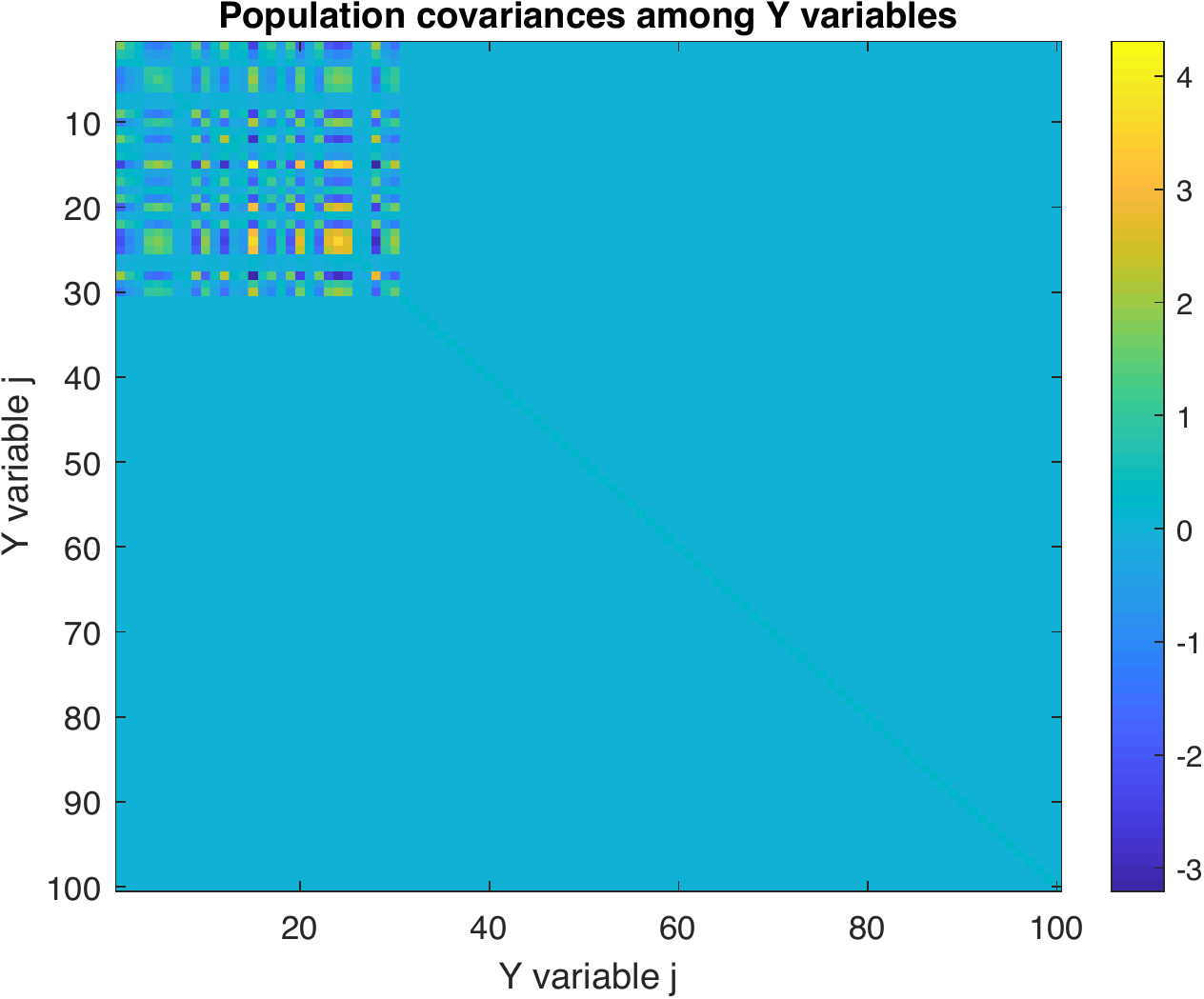}
\caption{Experimental setup 1: Heatmaps showing the sample (left) and population (right) \lsc{cross-covariances between $X$ and $Y$ variables (top), auto-covariances within $X$ variables (middle), and auto-covariances within $Y$ variables (bottom).}}
\label{fig:uncorrelated_variables_cov_XY}
\end{center}
\end{figure}

\subsubsection{Experimental setup 2: grouped variables}
\label{subsubsec_in_supp:grouped_variables}
The population cross- and auto-covariance matrices among \rev{random vectors $\mat{x}$ and $\mat{y}$} are
\begin{equation*}
    \expvof{\mat{x}} = \mat{0}, \quad \expvof{\mat{y}} = \mat{0}
\end{equation*}
\begin{align}
    \mat{\Sigma}_{\mat{xx}} &= \expvof{\mat{x} \mat{x}^\trans} =
  \begin{bmatrix}
    \mat{\Sigma}_1 &  &  &  &  & \\
 & \ddots &  &  &  &  \\
 &  & \mat{\Sigma}_R & &  &  \\
 &  &  & \mat{\Sigma}_{R+1} &  &  \\
 &  &  &  & \ddots  & \\
 &  &  &  &  & \mat{\Sigma}_G \\
  \end{bmatrix} \label{equ:grouped_variables_population_cov_matrix_X} \\
    \mat{\Sigma}_{\mat{yy}} &= \expvof{\mat{y} \mat{y}^\trans} = c^2 \mat{d} \mat{d}^\trans + \sigma^2 \mat{I}_q \label{equ:grouped_variables_population_cov_matrix_Y} \\
    \mat{\Sigma}_{\mat{xy}} &= \expvof{\mat{x} \mat{y}^\trans} = \mat{\Sigma}_{\mat{xx}} \mat{c} \mat{d}^\trans \label{equ:grouped_variables_population_cov_matrix_XY}
\end{align}
where $\mat{\Sigma}_g = (\sigma_{gij}) \in \real^{p_g \times p_g}$, with $\sigma_{gii}=1$ and $\sigma_{gij} \rho_{gi} \rho_{gj}$ for any $i\neq j$ and $g=1,2,\dots,G$, and $c^2 \coloneqq \expvof{z^2} = \mat{c}^\trans \mat{\Sigma}_{\mat{xx}} \mat{c}$. Here $R$ is the number of relevant/informative groups.

\begin{table}[htbp]
  \centering
  \caption{Group sizes of variables in $\mat{x}$}
  \label{tab:group_structure_X_variables}
  \resizebox{0.95\columnwidth}{!}{ 
    \begin{tabular}{rcccccccccccccccccccc}
\hline
Group ID	 & G1	 & G2	 & G3	 & G4	 & G5	 & G6	 & G7	 & G8	 & G9	 & G10	 & G11	 & G12	 & G13	 & G14	 & G15	 & G16	 & G17	 & G18	 & G19	 & G20	 \\
\hline
Group size	 & 89	 & 112	 & 92	 & 88	 & 88	 & 99	 & 130	 & 103	 & 94	 & 91	 & 99	 & 91	 & 90	 & 112	 & 96	 & 100	 & 96	 & 91	 & 103	 & 100	 \\
    \hline
    \end{tabular}
    }

\end{table}

\begin{figure}[tb]
\begin{center}
\includegraphics[width=0.4\linewidth]{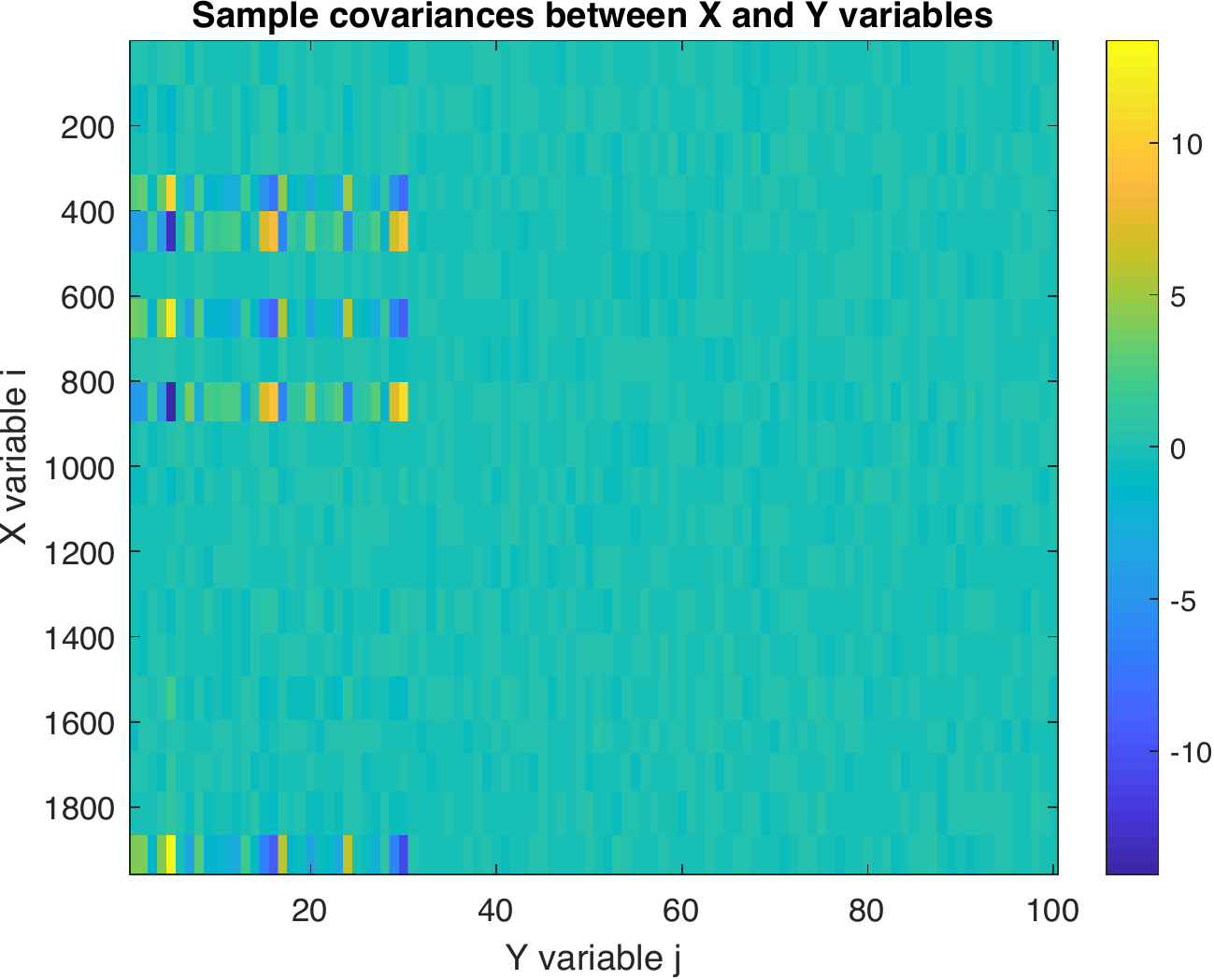}
\hspace{1cm}
\includegraphics[width=0.4\linewidth]{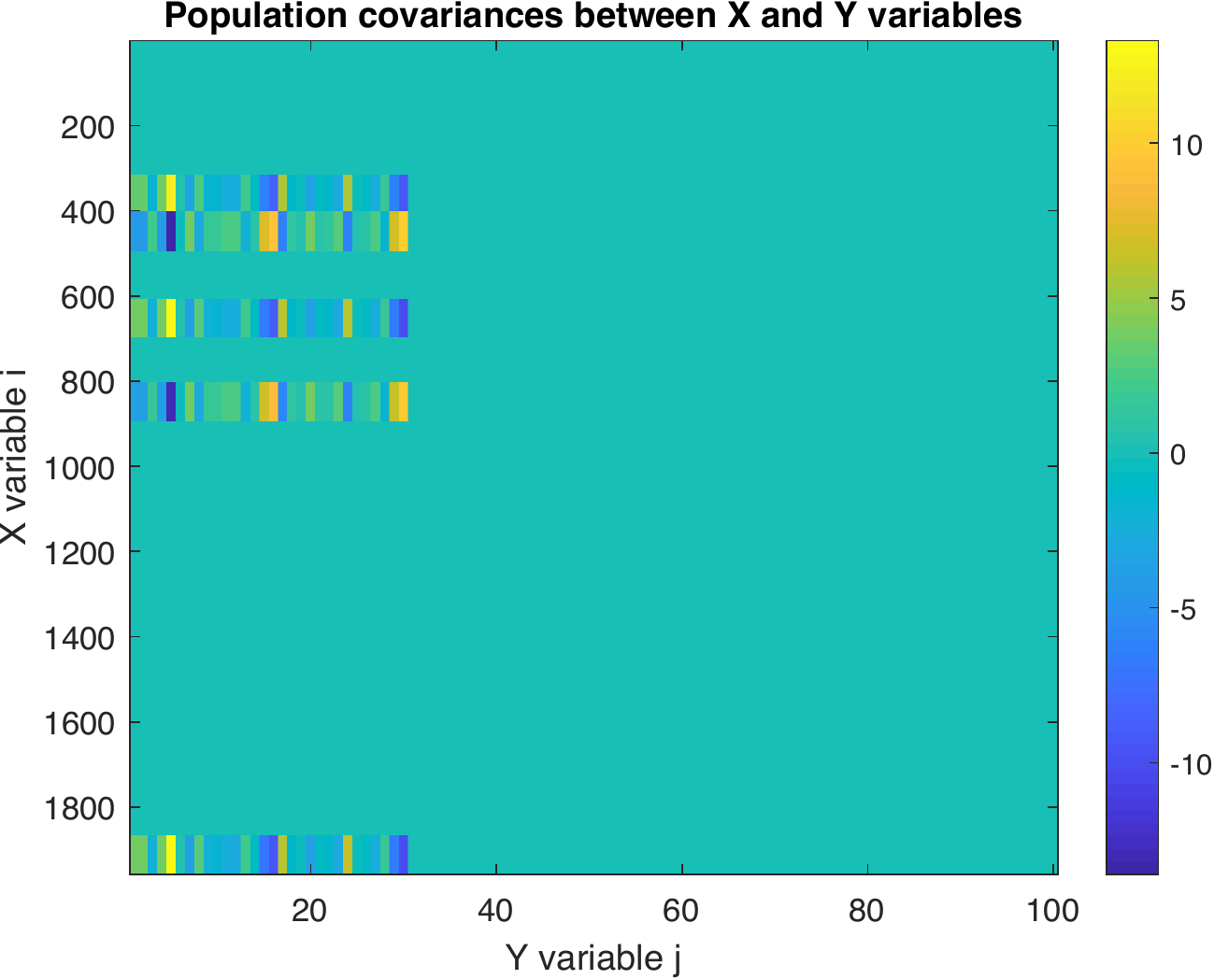} \\
\vspace{0.5cm}
\includegraphics[width=0.4\linewidth]{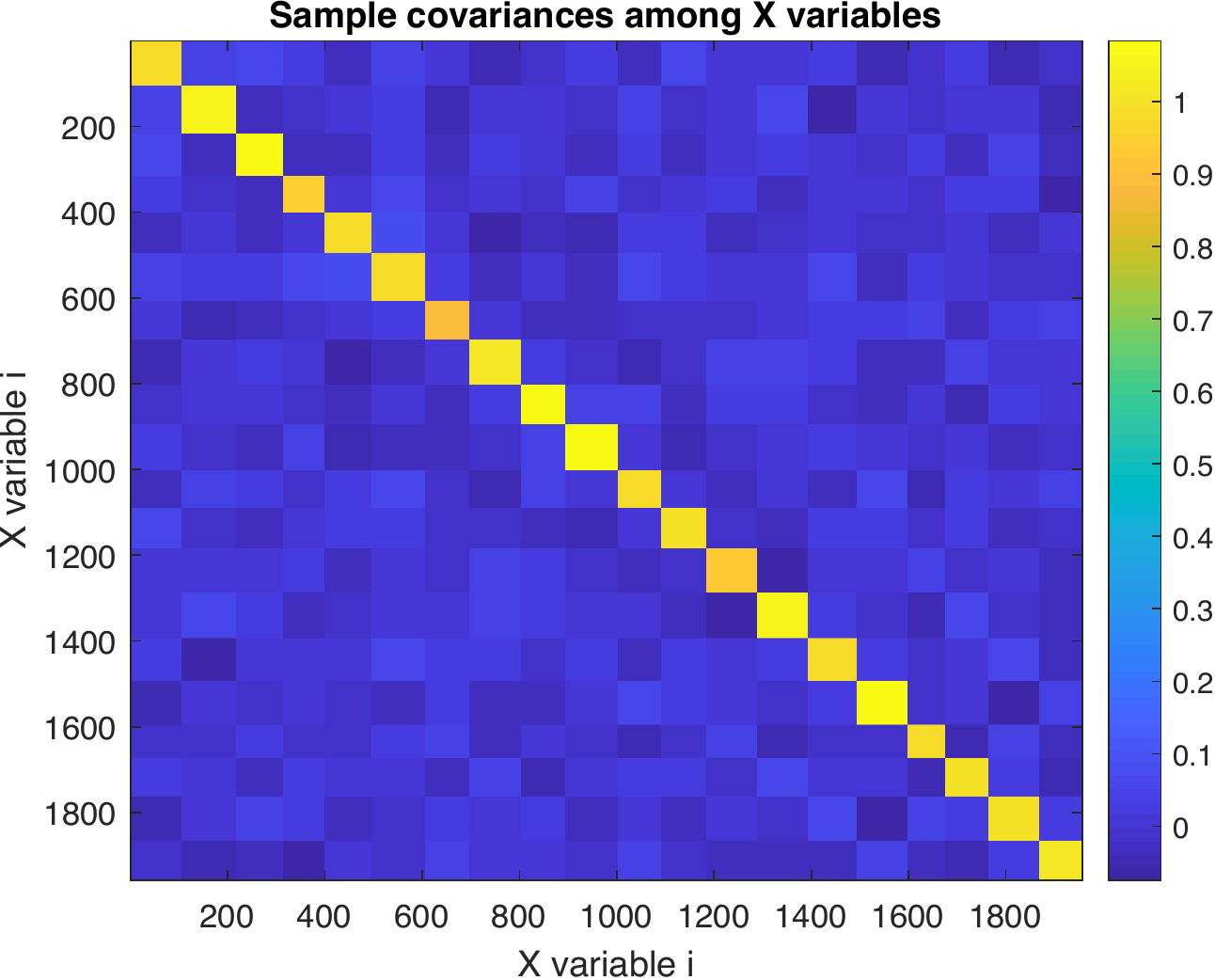}
\hspace{1cm}
\includegraphics[width=0.4\linewidth]{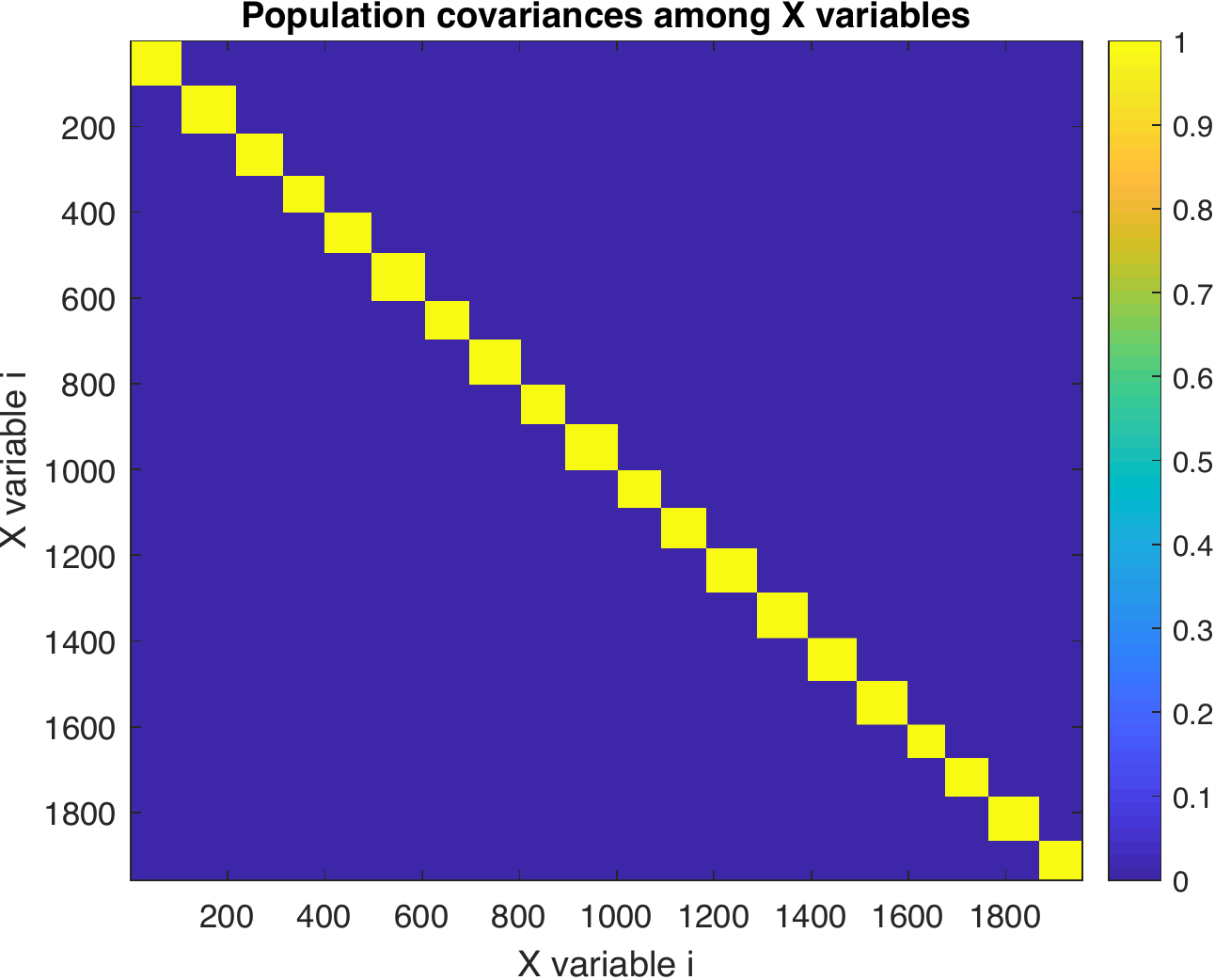} \\
\vspace{0.5cm}
\includegraphics[width=0.4\linewidth]{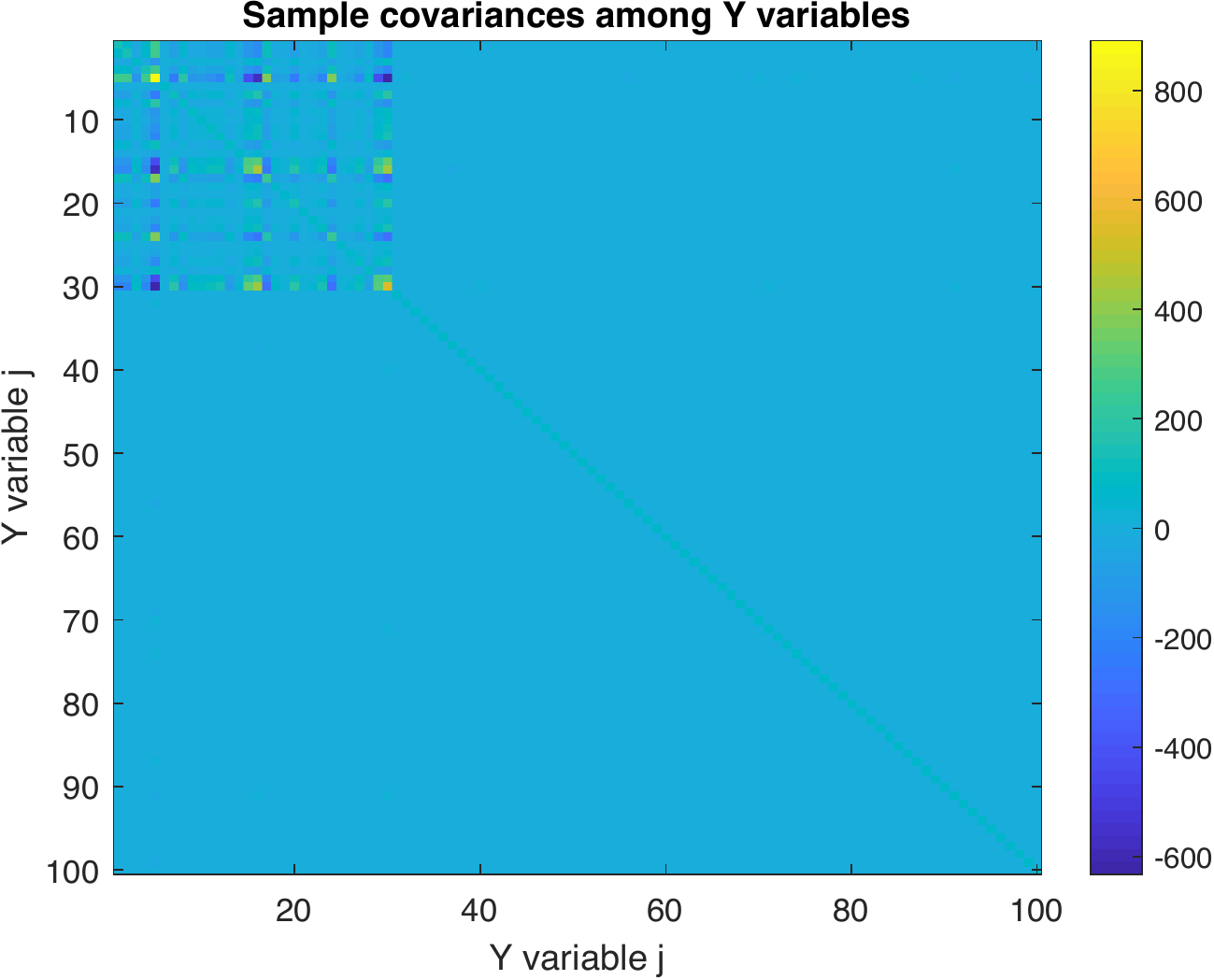}
\hspace{1cm}
\includegraphics[width=0.4\linewidth]{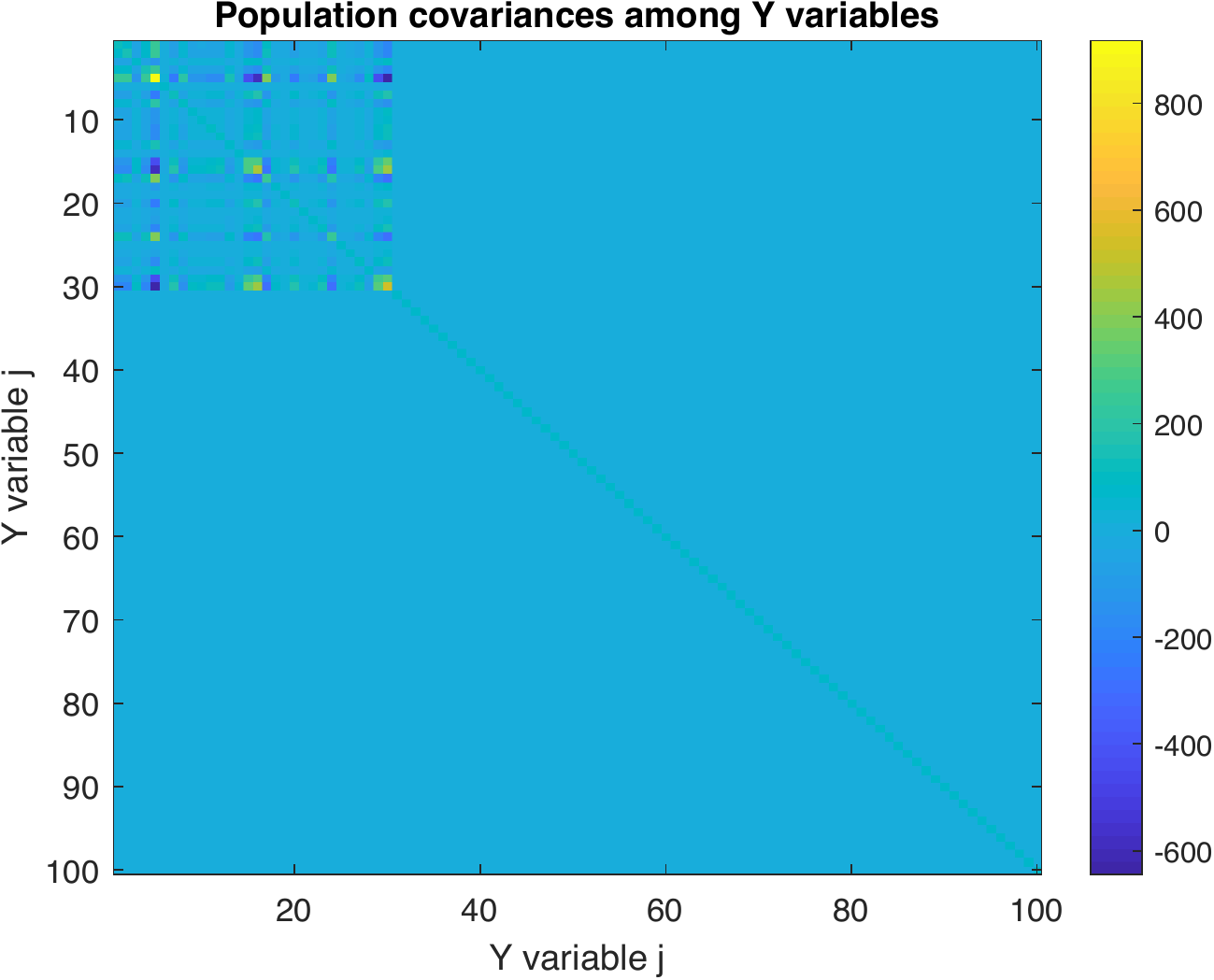}
\caption{Experimental setup 2: Heatmaps showing the sample (left) and population (right) \lsc{cross-covariances between $X$ and $Y$ variables (top), auto-covariances within $X$ variables (middle), and auto-covariances within $Y$ variables (bottom).}}
\label{fig:grouped_variables_cov}
\end{center}
\end{figure}

\clearpage

\subsection{Hyperparameter tuning and performance estimation}
\label{subsec_in_supp:Tuning_parameter_selection+generalization_performance_estimation}

To select the regularization parameters $\left(c_1,c_2\right)$ and estimate the generalization performance, we partition the data into training (50\%, $n_s$ samples), validation (25\%, $n_v$ samples), and testing (25\%, $n_t=n-n_s-n_v$ samples) data sets:
\begin{equation*}
\begin{bmatrix}
\mat{X} & \mat{Y}
\end{bmatrix}
=
\begin{bmatrix}
\mat{X}_{\rm train} & \mat{Y}_{\rm train} \\
\mat{X}_{\rm val} & \mat{Y}_{\rm val} \\
\mat{X}_{\rm test} & \mat{Y}_{\rm test}
\end{bmatrix}
 \in \real^{(n_s+n_v+n_t) \times (p+q)}
\end{equation*}
The training and validation data are used to tune the regularization parameters $\left(c_1,c_2\right)$, and the test data is used to estimate the performance.

To select the regularization parameters $\left(c_1,c_2\right)$, we fit the (simplified) SCCA model on the training data using each candidate value of $\left(c_1,c_2\right)$ as the regularization parameters, where $c_1$ and $c_2$ are chosen from a sequence of values
equally spaced on the log scale: 
$c_1 \in 2.\caret{\left(\floor{\log_2 c_{1,{\rm min}}}:\ceil{\log_2 c_{1,{\rm max}}}\right)}$, $c_2 \in \in 2.\caret{\left(\floor{\log_2 c_{2,{\rm min}}}:\ceil{\log_2 c_{2,{\rm max}}}\right)}$.
Here, $c_{\ell,{\rm min}}$ and $c_{\ell,{\rm max}}$, $\ell=1,2$, are the minimum and maximum value of $c_\ell$ which will be calculated for the standard and simplified SCCA models in Section \ref{subsubsec_in_supp:range_c1c2}. 

Denote the solution of the model fitted with $\left(c_1,c_2\right)$ as $\left(\mat{\hat{u}}_{\rm train}\left(c_1,c_2\right), \mat{\hat{v}}_{\rm train}\left(c_1,c_2\right)\right)$. For the standard SCCA model, the optimal $\left(c_1,c_2\right)$ are chosen as
\begin{align}\label{equ:TPS_criterion_standard_SCCA}
    \left(c_1^{\rm opt},c_2^{\rm opt}\right)
    &= \argmax_{c_1,c_2} \; \Corrof{\mat{X}_{\rm val} \mat{\hat{u}}_{\rm train}, \mat{Y}_{\rm val} \mat{\hat{v}}_{\rm train}} \\
    &= \argmax_{c_1,c_2} \; \frac{\inp{\mat{X}_{\rm val}\mat{\hat{u}}_{\rm train}}{\mat{Y}_{\rm val} \mat{\hat{v}}_{\rm train}}}{\twonorm{\mat{X}_{\rm val}\mat{\hat{u}}_{\rm train}}\twonorm{\mat{Y}_{\rm val}\mat{\hat{v}}_{\rm train}}}
\end{align}

For the simplified SCCA model, the optimal $\left(c_1,c_2\right)$ are chosen as\footnote{The reason the sample covariance matrix has $n_t$ in the denominator rather than $n_t-1$ is that we assume that population mean of $\mat{0}$ is known.}
\begin{align}\label{equ:TPS_criterion_simplified_SCCA}
    \left(c_1^{\rm opt},c_2^{\rm opt}\right)
    &= \argmax_{c_1,c_2} \; \Covof{\mat{X}_{\rm val} \mat{\hat{u}}_{\rm train}/\twonorm{\mat{\hat{u}}_{\rm train}}, \mat{Y}_{\rm val} \mat{\hat{v}}_{\rm train}/\twonorm{\mat{\hat{v}}_{\rm train}}} \\
    &= \argmax_{c_1,c_2} \; \frac{1}{n_t} \frac{\inp{\mat{X}_{\rm val}\mat{\hat{u}}_{\rm train}}{\mat{Y}_{\rm val} \mat{\hat{v}}_{\rm train}}}{\twonorm{\mat{\hat{u}}_{\rm train}}\twonorm{\mat{\hat{v}}_{\rm train}}} 
\end{align}

Then, we refit the SCCA model with $\left(c_1^{\rm opt},c_2^{\rm opt}\right)$ on all training data (combined training and validation data) to get the solution $\left(\mat{\hat{u}}_{\rm trainval}, \mat{\hat{v}}_{\rm trainval}\right)$. The canonical covariance and correlation on the test data are reported as the generalization performance:
\begin{align}\label{equ:generalization_performance_metric_test_cov}
    \Covof{\mat{X}_{\rm test} \mat{\hat{u}}_{\rm trainval}, \mat{Y}_{\rm test} \mat{\hat{v}}_{\rm trainval}} = \frac{\inp{\mat{X}_{\rm test}\mat{\hat{u}}_{\rm trainval}}{\mat{Y}_{\rm test} \mat{\hat{v}}_{\rm trainval}}}{\twonorm{\mat{\hat{u}}_{\rm trainval}}\twonorm{\mat{\hat{v}}_{\rm trainval}}}
\end{align}
\begin{align}\label{equ:generalization_performance_metric_test_corr}
    \Corrof{\mat{X}_{\rm test} \mat{\hat{u}}_{\rm trainval}, \mat{Y}_{\rm test} \mat{\hat{v}}_{\rm trainval}} = \frac{\inp{\mat{X}_{\rm test}\mat{\hat{u}}_{\rm trainval}}{\mat{Y}_{\rm test} \mat{\hat{v}}_{\rm trainval}}}{\twonorm{\mat{X}_{\rm test}\mat{\hat{u}}_{\rm trainval}}\twonorm{\mat{Y}_{\rm test}\mat{\hat{v}}_{\rm trainval}}}
\end{align}

\subsubsection{Effective range of $c_1$ and $c_2$}
\label{subsubsec_in_supp:range_c1c2}

To determine the range for the parameters $\left(c_1, c_2\right)$ for the standard SCCA model \eqref{equ:SCCA_model}, we replace its L2 inequality constraints with the L2 equality constraints:
\begin{equation}\label{equ:SCCA_model_L2_equality_constraints}
\begin{aligned}
& \underset{\mat{u},\mat{v}}{\text{maximize}} & & \mat{u}^\trans \mat{X}^\trans \mat{Y} \mat{v} \\
& \text{subject to}                           & & \mat{u}^\trans \mat{X}^\trans \mat{X} \mat{u} = 1, \normof{\mat{u}}{1} \leq c_1 \\
&                                             & & \mat{v}^\trans \mat{Y}^\trans \mat{Y} \mat{v} = 1, \normof{\mat{v}}{1} \leq c_2 \\
\end{aligned}
\end{equation}

\lsc{We note that for valid L1 regularization,  the L1 inequality constraints needs to be active (i.e., satisfied as equalities) at the optimal solution.} This implies that
\begin{align}
c_1 \geq \underset{\mat{u}}{\text{minimize}} \; \normof{\mat{u}}{1}  \quad  \text{subject to } \twonorm{\mat{X} \mat{u}}^2 = 1 \label{equ:standard_SCCA_feasible_region_u_L2_equality_c1_lower_bound} \\
c_2 \geq \underset{\mat{v}}{\text{minimize}} \; \normof{\mat{v}}{1}  \quad  \text{subject to } \twonorm{\mat{Y} \mat{v}}^2 = 1 \label{equ:standard_SCCA_feasible_region_v_L2_equality_c2_lower_bound}
\end{align}
and
\begin{align}
c_1 \leq \underset{\mat{u}}{\text{maximize}} \; \normof{\mat{u}}{1}  \quad  \text{subject to } \twonorm{\mat{X} \mat{u}}^2 = 1 \label{equ:standard_SCCA_feasible_region_u_L2_equality_c1_upper_bound} \\
c_2 \leq \underset{\mat{v}}{\text{maximize}} \; \normof{\mat{v}}{1}  \quad  \text{subject to } \twonorm{\mat{Y} \mat{v}}^2 = 1 \label{equ:standard_SCCA_feasible_region_v_L2_equality_c2_upper_bound}
\end{align}
Simple analysis shows that \lsc{a sufficient condition for \eqref{equ:standard_SCCA_feasible_region_u_L2_equality_c1_lower_bound}-\eqref{equ:standard_SCCA_feasible_region_v_L2_equality_c2_lower_bound} to hold} is
\begin{align}
c_1 \geq \maxof{\frac{1}{\rev{\sigma_{\rm max}}\left(\mat{X}\right)}, \frac{1}{\sqrt{\sum_{\ell=1}^n \max_{1 \leq i \leq p} x_{\ell i}^2}}} \eqqcolon c_{1,{\rm min}}  \label{equ:standard_SCCA_c1_lower_bound} \\
c_2 \geq \maxof{\frac{1}{\rev{\sigma_{\rm max}}\left(\mat{Y}\right)}, \frac{1}{\sqrt{\sum_{\ell=1}^n \max_{1 \leq j \leq q} y_{\ell j}^2}}} \eqqcolon c_{2,{\rm min}} \label{equ:standard_SCCA_c2_lower_bound}
\end{align}

Note however, that the objective in \eqref{equ:standard_SCCA_feasible_region_u_L2_equality_c1_upper_bound} (resp., \eqref{equ:standard_SCCA_feasible_region_v_L2_equality_c2_upper_bound}) is unbounded above when $n<p$ (resp., $n <q$), and thus it can not be used to find an effective maximum of $c_1$ (resp., $c_2$). To find an effective maximum value of $c_1$ and $c_2$, we solve problem \eqref{equ:SCCA_model_L2_equality_constraints} in the absence of L1 constraints instead:
\begin{equation}\label{equ:SCCA_model_L2_equality_constraints_wo_L1_constraints}
\begin{aligned}
& \underset{\mat{u},\mat{v}}{\text{maximize}} & & \mat{u}^\trans \mat{X}^\trans \mat{Y} \mat{v} \\
& \text{subject to}                           & & \mat{u}^\trans \mat{X}^\trans \mat{X} \mat{u} = 1 \\
&                                             & & \mat{v}^\trans \mat{Y}^\trans \mat{Y} \mat{v} = 1 \\
\end{aligned}
\end{equation}
Denote the optimal solution of problem \ref{equ:SCCA_model_L2_equality_constraints_wo_L1_constraints} as $\left(\mat{u}^*,\mat{v}^*\right)$. We set $c_{1,{\rm max}}=\normof{\mat{u}^*}{1}$ and $c_{2,{\rm max}}=\normof{\mat{v}^*}{1}$.

It can be shown that $\mat{u}^* = \left(\mat{X}^\trans \mat{X}\right)^{-1/2} \mat{u}_1 ,\mat{v}^* = \left(\mat{Y}^\trans \mat{Y}\right)^{-1/2} \mat{v}_1$, where $\mat{u}_1$ and $\mat{v}_1$ are respectively the left and right singular vectors of $\left(\mat{X}^\trans \mat{X}\right)^{-1/2} \mat{X}^\trans \mat{Y} \left(\mat{Y}^\trans \mat{Y}\right)^{-1/2}$ associated with the largest singular value. If $\mat{X}^\trans \mat{X}$ is singular, we can use $\mat{X}^\trans \mat{X} + \epsilon \mat{I}_p$ to approximate it; likewise for $\mat{Y}^\trans \mat{Y}$.

In a similar line of reasoning, to determine the range for the parameters $\left(c_1, c_2\right)$ for the simplified SCCA model \eqref{equ:SCCA_model}, consider
\begin{equation}\label{equ:simplified_SCCA_model_L2_equality_constraints}
\begin{aligned}
& \underset{\mat{u},\mat{v}}{\text{maximize}} & & \mat{u}^\trans \mat{X}^\trans \mat{Y} \mat{v} \\
& \text{subject to}                           & & \twonorm{\mat{u}}^2 = 1, \normof{\mat{u}}{1} \leq c_1 \\
&                                             & & \twonorm{\mat{v}}^2 = 1, \normof{\mat{v}}{1} \leq c_2 \\
\end{aligned}
\end{equation}

\lsc{We note that an effective value of $c_1$ and $c_2$ should be such that the L1 inequality constraints are active (i.e., satisfied as equalities) at the optimal solution.}

To this end, it should satisfy
\begin{align}
c_1 \geq \underset{\mat{u}}{\text{minimize}} \; \normof{\mat{u}}{1}  \quad  \text{subject to } \twonorm{\mat{u}}^2 = 1 \label{equ:simplified_SCCA_feasible_region_u_L2_equality_c1_lower_bound} \\
c_2 \geq \underset{\mat{v}}{\text{minimize}} \; \normof{\mat{v}}{1}  \quad  \text{subject to } \twonorm{\mat{v}}^2 = 1 \label{equ:simplified_SCCA_feasible_region_v_L2_equality_c2_lower_bound}
\end{align}
and
\begin{align}
c_1 \leq \underset{\mat{u}}{\text{maximize}} \; \normof{\mat{u}}{1}  \quad  \text{subject to } \twonorm{\mat{u}}^2 = 1 \label{equ:simplified_SCCA_feasible_region_u_L2_equality_c1_upper_bound} \\
c_2 \leq \underset{\mat{v}}{\text{maximize}} \; \normof{\mat{v}}{1}  \quad  \text{subject to } \twonorm{\mat{v}}^2 = 1 \label{equ:simplified_SCCA_feasible_region_v_L2_equality_c2_upper_bound}
\end{align}
From \eqref{equ:simplified_SCCA_feasible_region_u_L2_equality_c1_lower_bound}-\eqref{equ:simplified_SCCA_feasible_region_v_L2_equality_c2_upper_bound}, it follows that
\begin{align}
c_{1,{\rm min}} \coloneqq 1 \leq c_1 \leq \sqrt{p} \\
c_{2,{\rm min}} \coloneqq 1 \leq c_2 \leq \sqrt{q}
\end{align}

The upper bounds $\sqrt{p}$ for $c_1$ and $\sqrt{q}$ for $c_2$ are too relaxed. To find a tighter bound, we solve problem \eqref{equ:simplified_SCCA_model_L2_equality_constraints} in the absence of L1 constraints instead:
\begin{equation}\label{equ:simplified_SCCA_model_L2_equality_constraints_wo_L1_constraints}
\begin{aligned}
& \underset{\mat{u},\mat{v}}{\text{maximize}} & & \mat{u}^\trans \mat{X}^\trans \mat{Y} \mat{v} \\
& \text{subject to}                           & & \twonorm{\mat{u}}^2 = 1 \\
&                                             & & \twonorm{\mat{v}}^2 = 1 \\
\end{aligned}
\end{equation}
The optimal solution is $\mat{u}^* = \mat{u}_1, \mat{v}^* = \mat{v}_1$, where $\mat{u}_1$ and $\mat{v}_1$ are respectively the left and right singular vectors of $\mat{X}^\trans \mat{Y}$ associated with the largest singular value. We set $c_{1,{\rm max}}=\normof{\mat{u}^*}{1}$ and $c_{2,{\rm max}}=\normof{\mat{v}^*}{1}$.

\subsection{Variable selection performance}
\label{subsec_in_supp:cariable_selection_performance}

The balanced accuracy (bACC) and Matthews correlation coefficient (MCC) are defined as
\begin{equation}\label{def:bACC}
    \mathrm{bACC} = \frac{1}{2} \left(\frac{\mathrm{TP}}{\mathrm{TP} + \mathrm{FN}} + \frac{\mathrm{TN}}{\mathrm{TN} + \mathrm{FP}}\right),
\end{equation}
\begin{equation}\label{def:MCC}
    \mathrm{MCC} =\frac{\mathrm{TP} \times \mathrm{TN} -\mathrm{FP} \times \mathrm{FN}}{\sqrt{(\mathrm{TP} +\mathrm{FP} )(\mathrm{TP} +\mathrm{FN} )(\mathrm{TN} +\mathrm{FP} )(\mathrm{TN} +\mathrm{FN})}},
\end{equation}
where TP, TN, FP, and FN denote the numbers of true positives, true negatives, false positives, and false negatives, respectively. The bACC and MCC are overall measures of variable selection accuracy, and a larger score indicates a better variable selection performance.The relative absolute error (RAE), which for the selection of $X$ variables is defined as
\begin{equation}\label{def:RAE}
    \mathrm{RAE} = \frac{\normof{\mat{\hat{u}}-\mat{u}^*}{1}}{\normof{\mat{u}^*}{1}}
\end{equation}
where $\mat{u}^*$ and $\mat{\hat{u}}$ denote the true and estimated canonical vector, respectively.
Our variable selection performance on the synthetic data is shown in
Table~\ref{tab:X_variable_selection_performance_synthetic_data} and
Table~\ref{tab:Y_variable_selection_performance_synthetic_data}.

\begin{table*}[tb]
\caption{The $X$ variable selection performance of the SCCA and simplified SCCA on whole training data.}
\label{tab:X_variable_selection_performance_synthetic_data}
\vskip 0.15in
\begin{center}
\begin{small}
\begin{sc}
\begin{tabular}{lcccccccr}
\toprule
  Model     & Recall & Precision & F1 score & ACC & bACC & MCC & PR AUC & RAE \\
\midrule
          \multicolumn{9}{c}{Experimental setup 1} \\
SCCA		 &    0.820	 &    0.276	 &    0.413	 &    0.767	 &    0.791	 &    0.382	 &    0.759	 &    0.384	 \\
Simp SCCA	 &    0.430	 &    0.306	 &    0.358	 &    0.846	 &    0.661	 &    0.278	 &    0.429	 &    1.044	 \\
          \multicolumn{9}{c}{Experimental setup 2} \\
SCCA		 &    1.000	 &    0.233	 &    0.378	 &    0.233	 &    0.500	 &      NaN	 &    0.998	 &    0.320	 \\
Simp SCCA	 &    1.000	 &    0.602	 &    0.751	 &    0.846	 &    0.899	 &    0.693	 &    0.800	 &    0.110	 \\
\bottomrule
\end{tabular}
\end{sc}
\end{small}
\end{center}
\vskip -0.1in
\end{table*}

\begin{table*}[tb]
\caption{The $Y$ variable selection performance of the SCCA and simplified SCCA on whole training data.}
\label{tab:Y_variable_selection_performance_synthetic_data}
\vskip 0.15in
\begin{center}
\begin{small}
\begin{sc}
\begin{tabular}{lcccccccr}
\toprule
  Model     & Recall & Precision & F1 score & ACC & bACC & MCC & PR AUC & RAE \\
\midrule
          \multicolumn{9}{c}{Experimental setup 1} \\
SCCA		 &    1.000	 &    0.300	 &    0.462	 &    0.300	 &    0.500	 &      NaN	 &    0.896	 &    0.823	 \\
Simp SCCA	 &    1.000	 &    0.566	 &    0.723	 &    0.770	 &    0.836	 &    0.616	 &    0.957	 &    0.030	 \\
          \multicolumn{9}{c}{Experimental setup 2} \\
SCCA		 &    1.000	 &    0.300	 &    0.462	 &    0.300	 &    0.500	 &      NaN	 &    0.844	 &    0.503	 \\
Simp SCCA	 &    0.967	 &    0.558	 &    0.707	 &    0.760	 &    0.819	 &    0.585	 &    0.948	 &    0.050	 \\
\bottomrule
\end{tabular}
\end{sc}
\end{small}
\end{center}
\vskip -0.1in
\end{table*}


\section{Supporting Information and additional results for imaging genetic data analysis in Section \eqref{subsec:experiment_real_data}}
\label{sec_in_supp:experiment_real_data}

\subsection{Subject characteristics}
\label{subsec_in_supp:subject_characteristics}

\begin{table*}[htbp]
  \centering
  \caption{Subject characteristics}
    \begin{tabular}{rccccc}
    \hline
          & HC    & SMC   &  EMCI  & LMCI   & AD \\
    \hline
Num				 &  183		 &  75		 &  218		 &  184		 &  97		 \\
Gender (M/F)	 &  89/94	 &  29/46	 &  113/105	 &  96/88	 &  54/43	 \\
Handedness (R/L)	 &  163/20	 &  65/10	 &  195/23	 &  165/19	 &  89/8	 \\
Age (mean$\pm$std)	 &  73.96$\pm$5.50	 &  71.77$\pm$5.76	 &  70.56$\pm$7.16	 &  71.89$\pm$7.92	 &  73.99$\pm$8.44	 \\
Edu (mean$\pm$std)	 &  16.44$\pm$2.67	 &  16.87$\pm$2.71	 &  15.95$\pm$2.64	 &  16.14$\pm$2.92	 &  15.60$\pm$2.61	 \\
    \hline
    \end{tabular}%
  \label{tab:demographics}%
\end{table*}

Participant characteristics of our real imaging genetics data from the Alzheimer's Disease Neuroimaging Initiative (ADNI) cohort is shown in Table~\ref{tab:demographics}.

\subsection{Correlation structure of the real imaging genetic data}
\label{subsec_in_supp:corr_structure_real_data}

\begin{figure*}[htb]
\centering
    \includegraphics[width=0.33\textwidth]{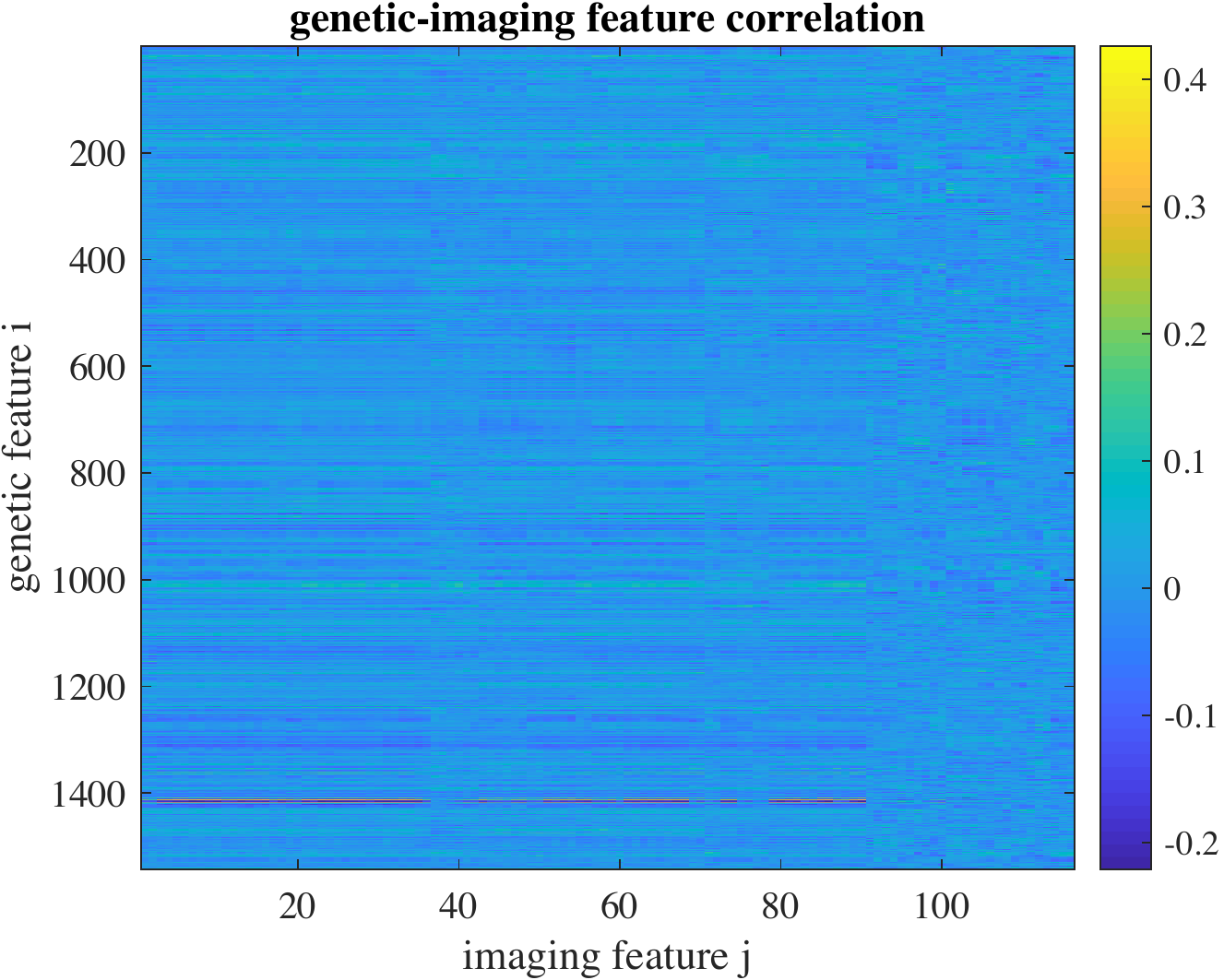}
    \includegraphics[width=0.33\linewidth]{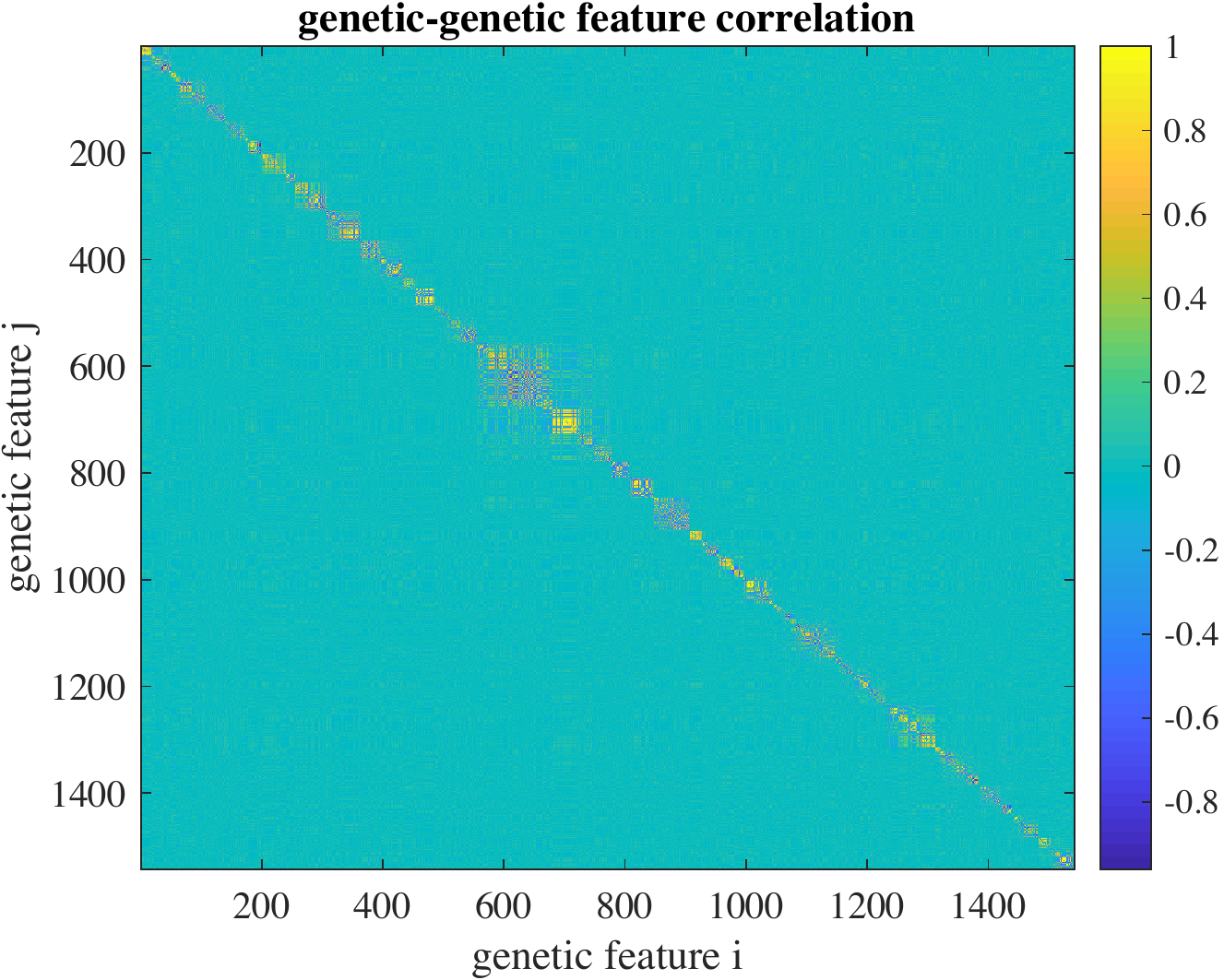}
    \includegraphics[width=0.33\linewidth]{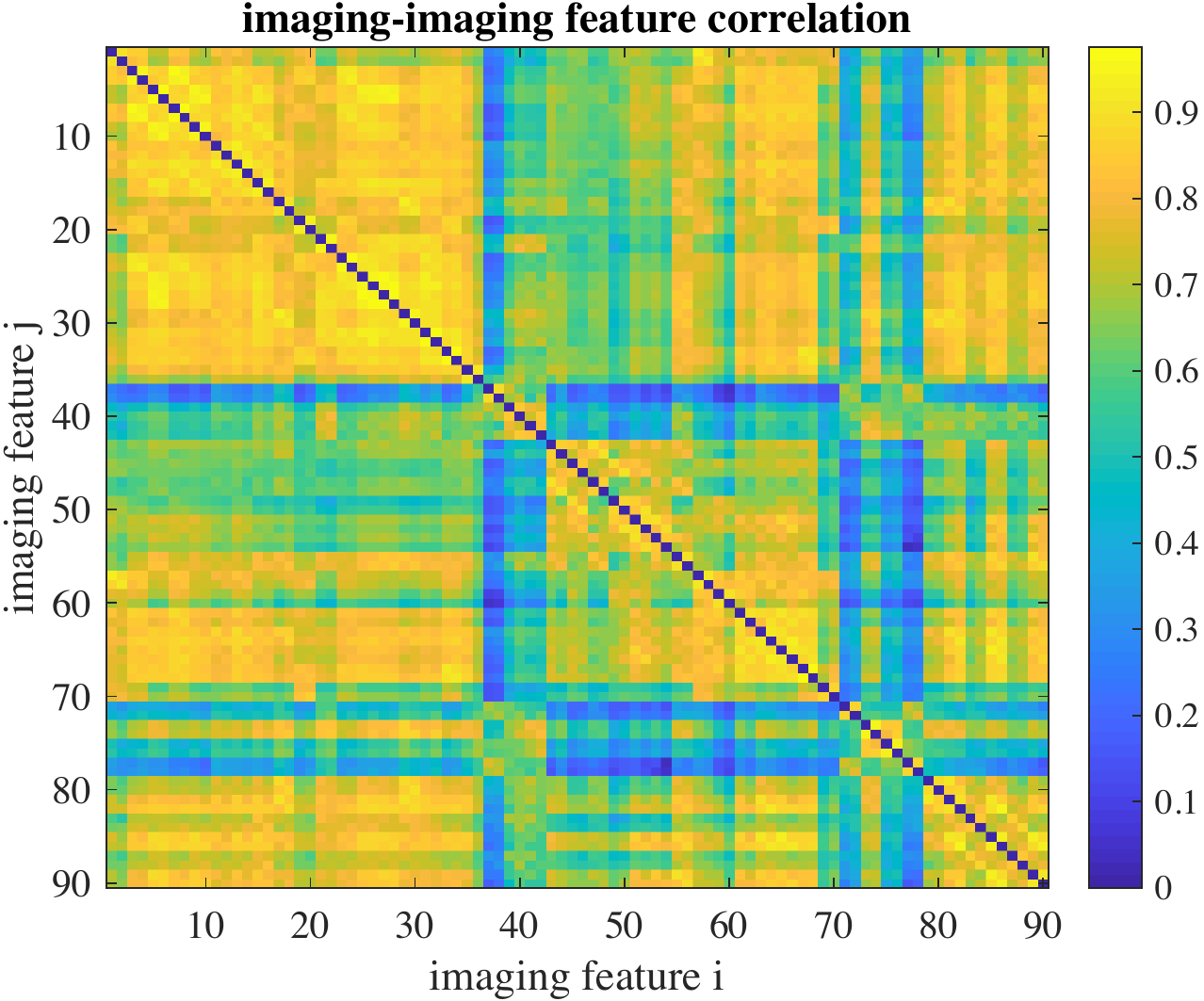}
  \caption{Heatmaps showing the pairwise sample Pearson correlation coefficients between genetic and imaging features (left), within genetic features (middle), and within imaging features (right).}
    \label{fig:corr_SNP_AV45}
\end{figure*}

Correlation structure of the real ADNI imaging genetics data used in this study is shown in Fig.~\ref{fig:corr_SNP_AV45}.

\subsection{Hyperparameter tuning and generalization performance estimation}
\label{Suppsubsec:real_data_Tuning_parameter_selection+generalization_performance_estimation}

We employ the nested cross-validation method which is an extension of the procedure described in Section \ref{subsec_in_supp:Tuning_parameter_selection+generalization_performance_estimation}.
We first \lsc{randomly divide each category of subjects into five roughly equal-sized subgroups and combine the data from each category to form five outer folds}.

We used the first fold for testing and the remaining folds for training/validating the model. Test set data are put aside. The following steps were carried out with the training+validation data:
\begin{description}

  \item[(1)] We employ the stratified cross-validation (CV) method to choose $\left(c_1,c_2\right)$. The samples/subjects from each category are randomly divided into five roughly equal-sized subgroups and then combined to form five folds $\set{I} = \cup_{k=1}^5 \set{I}_k$. \lsc{Denote $\mat{X}_{\set{I}_k}^{\rm trainval}$ and $\mat{Y}_{\set{I}_k}^{\rm trainval}$, $k=1,2,\dots,5$, as the submatrices formed by the rows of $\mat{X}^{\rm trainval}$ and $\mat{Y}^{\rm trainval}$ indexed by $\set{I}_k$, respectively}.

  \item[(2)] The SCCA model is fitted to the normalized $\left(\mat{X}_{\set{I}\setminus\set{I}_1}^{\rm trainval},\mat{Y}_{\set{I}\setminus\set{I}_1}^{\rm trainval}\right)$ to obtain the solution as $\left(\mat{\hat{u}}_{-1}^{\rm trainval}, \mat{\hat{v}}_{-1}^{\rm trainval}\right)$. Then, the performance on the validation data is recorded as $\opof{Corr}{\mat{X}_{\set{I}_1}^{\rm trainval} \mat{\hat{u}}_{-1}^{\rm trainval}, \mat{Y}_{\set{I}_1}^{\rm trainval} \mat{\hat{v}}_{-1}^{\rm trainval}}$. This process is repeated five times with each fold of samples/subjects used once as the validation set.

  \item[(3)] The cross-validation criterion to select the regularization parameters is defined as
\begin{equation}\label{equ:K_fold_CV_TPS_criterion_standard_SCCA}
    \left(c_1^{\rm opt},c_2^{\rm opt}\right) = \argmax_{c_1,c_2} \; \frac{1}{5} \sum_{k=1}^5 \opof{Corr}{\mat{X}_{\set{I}_k}^{\rm trainval} \mat{\hat{u}}_{-k}^{\rm trainval}, \mat{Y}_{\set{I}_k}^{\rm trainval} \mat{\hat{v}}_{-k}^{\rm trainval}}
\end{equation}
where $\opof{Corr}{\cdot,\cdot}$ is the correlation function and $\left(\mat{\hat{u}}_{-k}^{\rm trainval},\mat{\hat{v}}_{-k}^{\rm trainval}\right)$ are the estimates of $\left(\mat{u},\mat{v}\right)$ by the standard SCCA on the training+validation data $\left(\mat{X}_{\set{I}\setminus\set{I}_k}^{\rm trainval},\mat{Y}_{\set{I}\setminus\set{I}_k}^{\rm trainval}\right)$ with $\left(c_1,c_2\right)$ as regularization parameters.

  \item[(4)] The SCCA model was then fit to the entire training set at $\left(c_1^{\rm opt},c_2^{\rm opt}\right)$ to estimate the canonical weights $\left(\mat{\hat{u}}^{\rm opt}, \mat{\hat{v}}^{\rm opt}\right)$.
\end{description}

The canonical correlation on the test data $\Corrof{\mat{X}_{\rm test} \mat{\hat{u}}^{\rm opt}, \mat{Y}_{\rm test} \mat{\hat{v}}^{\rm opt}}$ is reported as the generalization performance. For the simplified SCCA, the canonical covariance is used as the metric to measure the performance and to tune the regularization parameters.

This process is repeated five times with each outer fold of samples/subjects used once as the testing set.

\subsection{Genetic and Imaging Marker Selection}
\label{subsec_in_supp:SNPs}

\begin{figure}[htb]
\begin{center}
\includegraphics[width=0.8\linewidth]{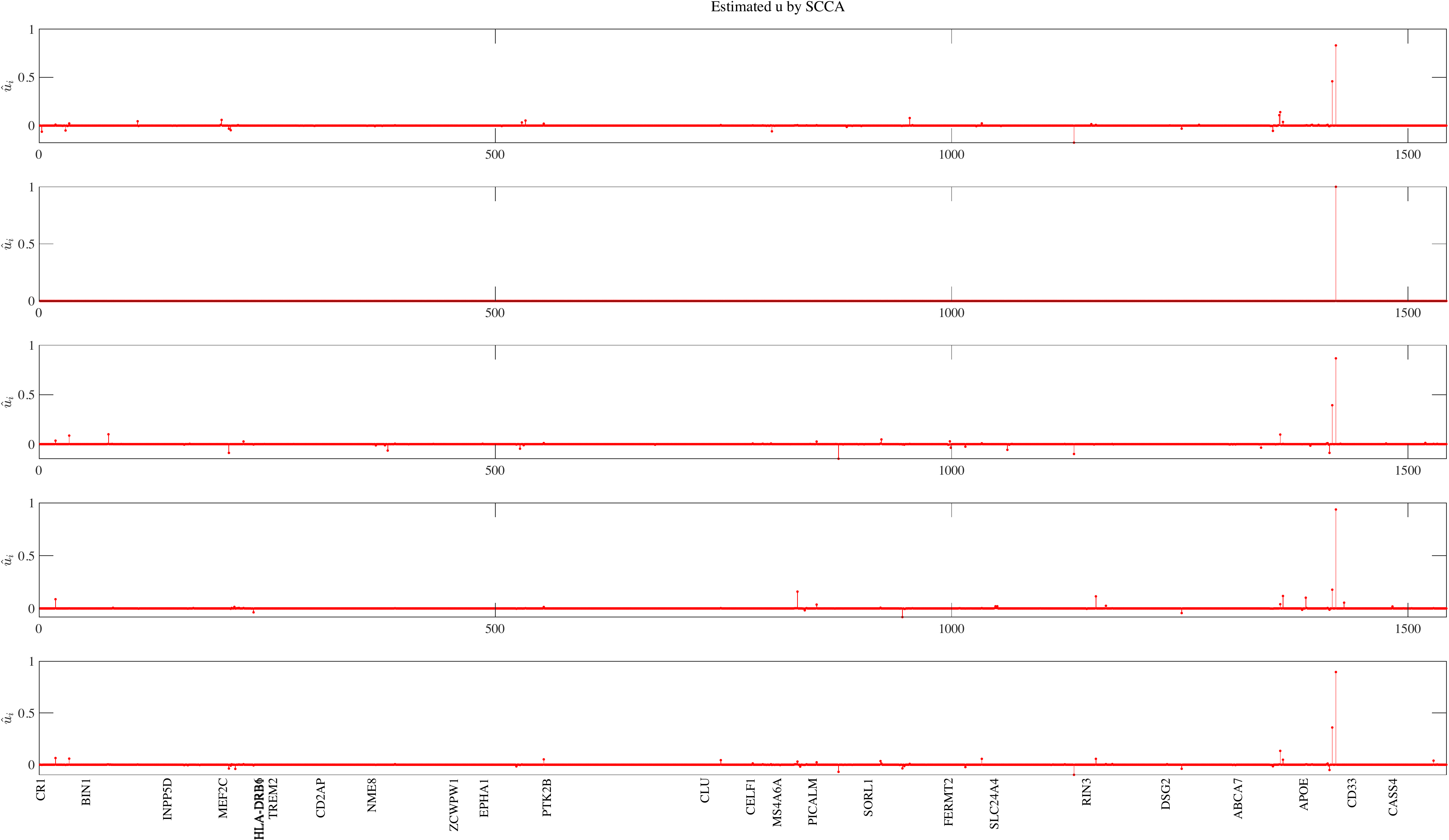} \\
\vspace{0.5cm}
\includegraphics[width=0.8\linewidth]{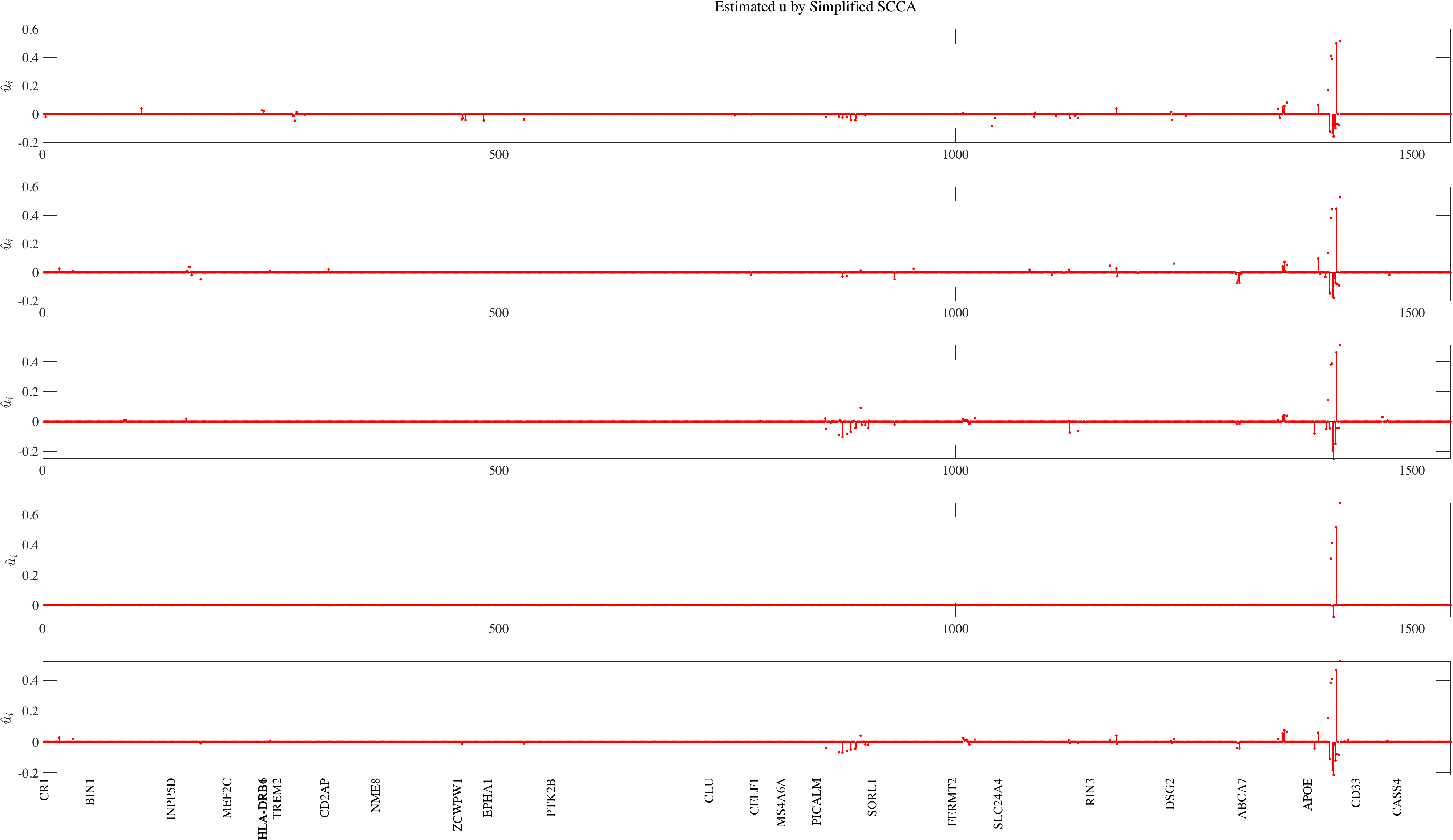}
\caption{Canonical genetic weights estimated by SCCA (top figure) and simplified SCCA (bottom figure). In each figure, the results on each of the four training folds (rows 1-4) and on the entire data (bottom row) are shown.}
\label{fig:genetic_weight_esti}
\end{center}
\end{figure}

\begin{figure}[htb]
\begin{center}
\includegraphics[width=0.8\linewidth]{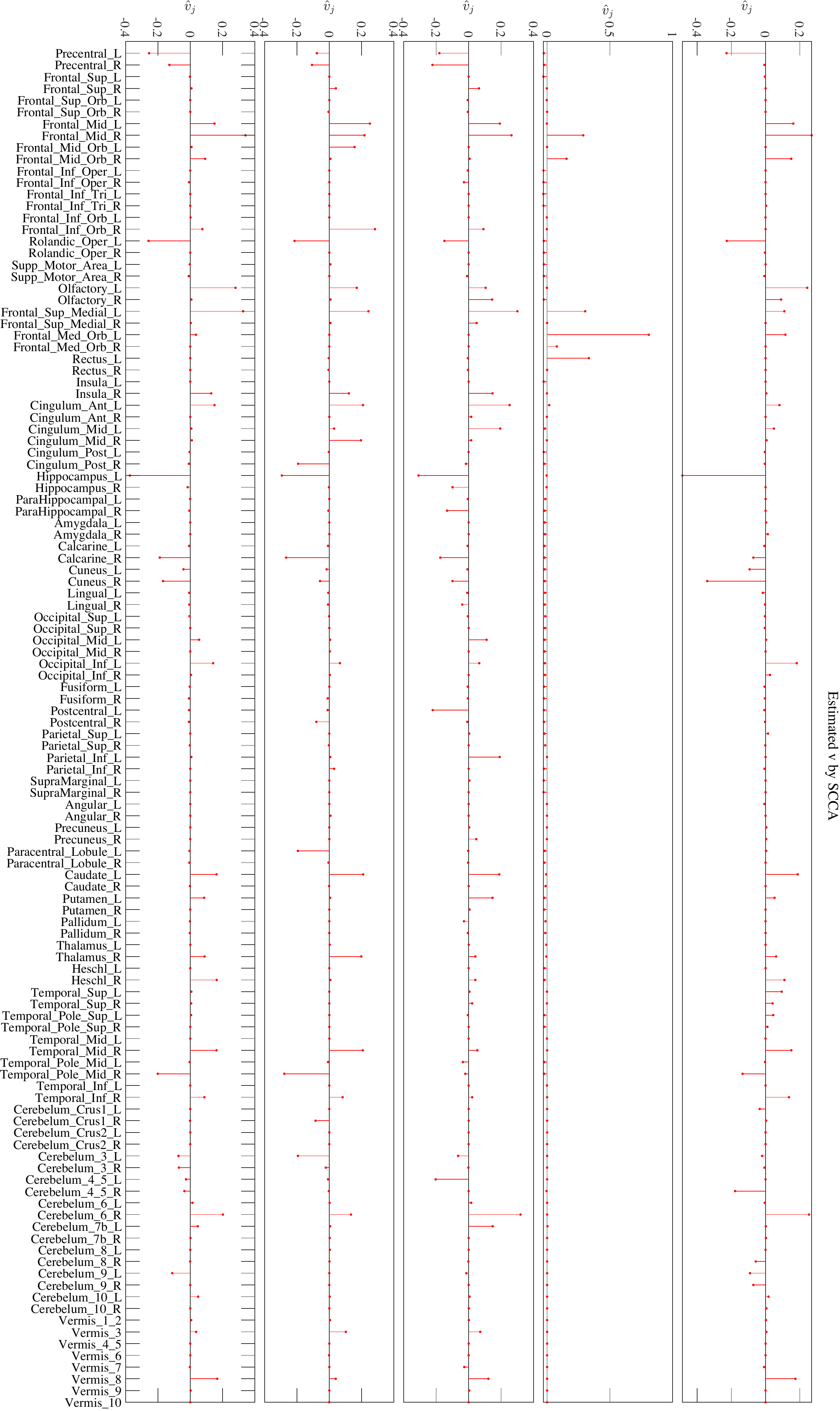}
\caption{Canonical imaging weights estimated by SCCA (top figure) and simplified SCCA (bottom figure). In each figure, the results on each of the four training folds (rows 1-4) and on the entire data (bottom row) are shown.}
\label{fig:imaging_weight_esti}
\end{center}
\end{figure}

\begin{figure}[htb]
\begin{center}
\includegraphics[width=0.8\linewidth]{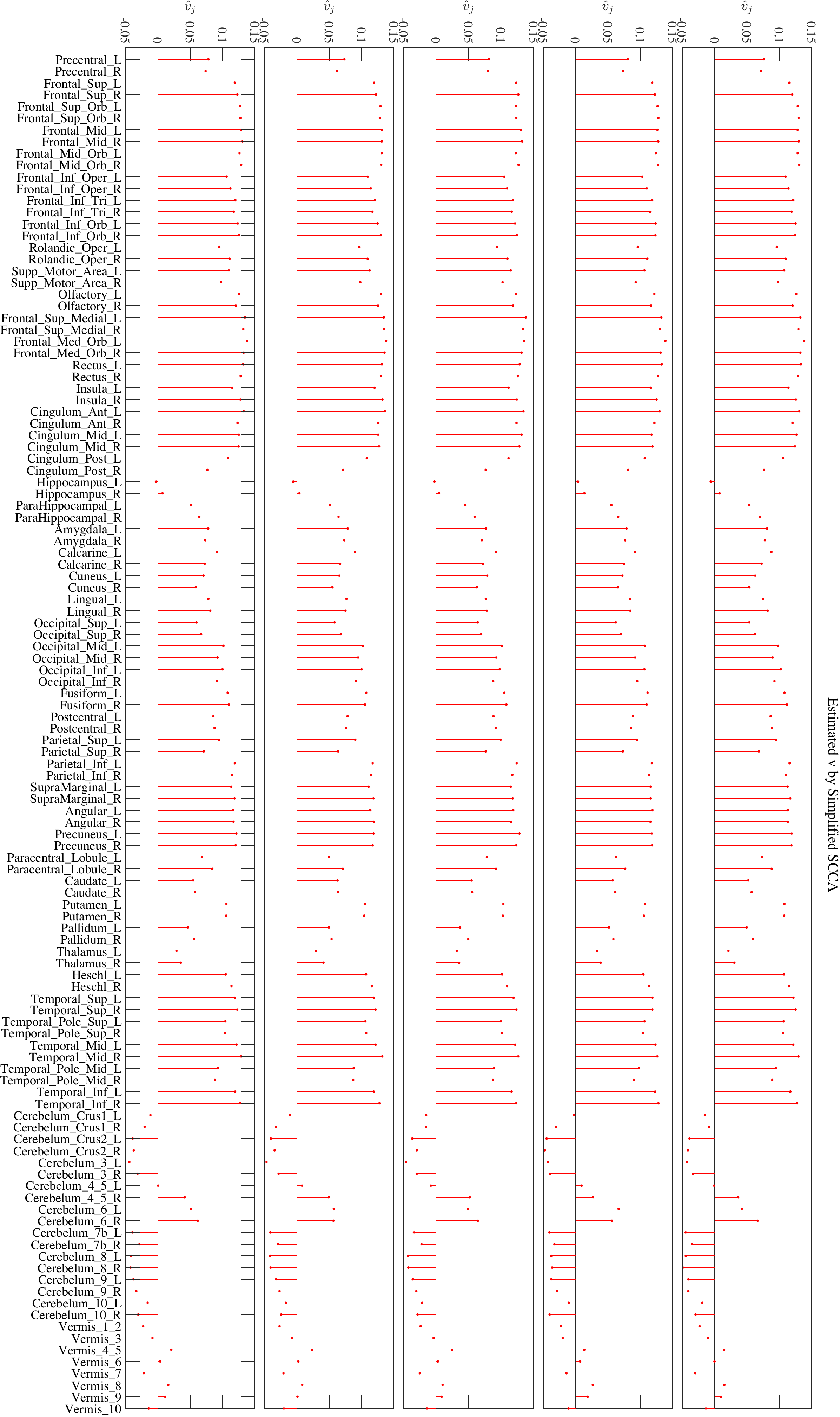}
\caption{Canonical imaging weights estimated by SCCA (top figure) and simplified SCCA (bottom figure). In each figure, the results on each of the four training folds (rows 1-4) and on the entire data (bottom row) are shown.}
\label{fig:imaging_weight_esti}
\end{center}
\end{figure}

\begin{center}
\footnotesize
\begin{longtable}{lccccccr}
\caption{Genetic features (ordered by absolute values of estimated canonical weights) selected by SCCA and simplified SCCA.}
\label{tab:snps_selected} \\
\toprule
\multicolumn{4}{c}{Standard SCCA} & \multicolumn{4}{c}{Simplified SCCA} \\
\cmidrule(lr){1-4}\cmidrule(lr){5-8}
       SNP   &    Closest gene   & $\hat{u}_i$ & p-value &        SNP    &    Closest gene   & $\hat{u}_i$ & p-value \\
\midrule
\endfirsthead

\multicolumn{8}{c}%
{{\bfseries \tablename\ \thetable{} -- continued from previous page}} \\
\toprule
\multicolumn{4}{c}{Standard SCCA} & \multicolumn{4}{c}{Simplified SCCA} \\
\cmidrule(lr){1-4}\cmidrule(lr){5-8}
       SNP   &    Closest gene   & $\hat{u}_i$ & p-value &        SNP    &    Closest gene   & $\hat{u}_i$ & p-value \\
\midrule
\endhead

\hline \multicolumn{8}{r}{{Continued on next page}} \\ \hline
\endfoot

\bottomrule \bottomrule
\endlastfoot

 rs4420638	 &    \textit{APOE}	 & 0.892	 & 8.50e-12	 &  rs4420638	 &    \textit{APOE}	 & 0.522	 & 8.50e-12  \\
  rs769449	 &    \textit{APOE}	 & 0.366	 & 1.60e-12	 &   rs769449	 &    \textit{APOE}	 & 0.466	 & 1.60e-12  \\
rs10404947	 &   \textit{ABCA7}	 & 0.140	 & 5.16e-02	 &   rs157582	 &    \textit{APOE}	 & 0.408	 & 2.37e-05  \\
rs12434016	 & \textit{SLC24A4}	 & -0.102	 & 9.11e-01	 &  rs2075650	 &    \textit{APOE}	 & 0.383	 & 4.46e-07  \\
rs17258982	 &     \textit{CR1}	 & 0.069	 & 5.83e-01	 &  rs1160985	 &    \textit{APOE}	 & -0.213	 & 7.29e-06  \\
  rs609903	 &  \textit{PICALM}	 & -0.065	 & 6.38e-01	 &  rs8106922	 &    \textit{APOE}	 & -0.183	 & 2.27e-03  \\
 rs7141622	 &    \textit{RIN3}	 & 0.058	 & 8.92e-01	 &     rs6859	 &    \textit{APOE}	 & 0.156	 & 9.20e-03  \\
 rs3818361	 &     \textit{CR1}	 & 0.056	 & 8.24e-03	 &   rs405509	 &    \textit{APOE}	 & -0.121	 & 4.20e-03  \\
  rs923892	 &   \textit{SORL1}	 & -0.052	 & 3.82e-01	 &   rs157580	 &    \textit{APOE}	 & -0.111	 & 1.70e-01  \\
 rs2949766	 &   \textit{EPHA1}	 & 0.051	 & 1.58e-01	 &   rs584007	 &    \textit{APOE}	 & -0.084	 & 3.85e-01  \\
rs17126012	 &  \textit{FERMT2}	 & 0.048	 & 4.25e-01	 &   rs439401	 &    \textit{APOE}	 & -0.078	 & 3.26e-01  \\
 rs3087554	 &     \textit{CLU}	 & 0.046	 & 4.43e-01	 & rs10404947	 &   \textit{ABCA7}	 & 0.076	 & 5.16e-02  \\
 rs1160985	 &    \textit{APOE}	 & -0.043	 & 7.29e-06	 &   rs609903	 &  \textit{PICALM}	 & -0.067	 & 6.38e-01  \\
    rs6843	 &   \textit{ABCA7}	 & 0.043	 & 1.61e-01	 &   rs637304	 &  \textit{PICALM}	 & -0.067	 & 3.06e-01  \\
 rs1422189	 &   \textit{MEF2C}	 & -0.042	 & 5.00e-02	 &     rs6843	 &   \textit{ABCA7}	 & 0.066	 & 1.61e-01  \\
 rs2304607	 &   \textit{MEF2C}	 & -0.040	 & 2.05e-01	 &   rs519825	 &    \textit{APOE}	 & 0.060	 & 5.48e-01  \\
rs17660414	 &    \textit{DSG2}	 & -0.038	 & 9.48e-01	 &   rs694011	 &  \textit{PICALM}	 & -0.059	 & 5.21e-01  \\
 rs6064401	 &   \textit{CASS4}	 & 0.035	 & 5.46e-01	 &  rs2074442	 &   \textit{ABCA7}	 & 0.057	 & 1.12e-01  \\
rs11230197	 &  \textit{MS4A6A}	 & 0.034	 & 3.81e-01	 &   rs757232	 &   \textit{ABCA7}	 & 0.053	 & 8.50e-02  \\
   rs93882	 &   \textit{SORL1}	 & 0.022	 & 4.83e-01	 &   rs561655	 &  \textit{PICALM}	 & -0.050	 & 6.54e-01  \\
rs12703526	 &   \textit{EPHA1}	 & -0.022	 & 9.19e-01	 &  rs1237999	 &  \textit{PICALM}	 & -0.043	 & 8.32e-01  \\
  rs611267	 &  \textit{MS4A6A}	 & -0.021	 & 1.19e-01	 & rs34374273	 &    \textit{APOE}	 & -0.041	 & 5.45e-02  \\
 rs7936092	 &  \textit{PICALM}	 & 0.021	 & 2.07e-01	 &  rs1667284	 &    \textit{DSG2}	 & -0.041	 & 6.87e-01  \\
  rs733430	 &   \textit{SORL1}	 & 0.020	 & 2.08e-01	 & rs10898436	 &  \textit{PICALM}	 & 0.040	 & 4.00e-01  \\
 rs2279796	 &   \textit{ABCA7}	 & -0.020	 & 3.68e-01	 & rs11608136	 &  \textit{PICALM}	 & -0.040	 & 7.11e-01  \\
 rs8008270	 &  \textit{FERMT2}	 & -0.013	 & 5.49e-02	 &  rs8013925	 &    \textit{RIN3}	 & 0.040	 & 5.30e-01  \\
 rs8013925	 &    \textit{RIN3}	 & 0.012	 & 5.30e-01	 &  rs1791161	 &    \textit{DSG2}	 & -0.040	 & 6.82e-01  \\
  rs157582	 &    \textit{APOE}	 & 0.011	 & 2.37e-05	 &   rs543293	 &  \textit{PICALM}	 & -0.029	 & 7.78e-01  \\
 rs1667284	 &    \textit{DSG2}	 & -0.009	 & 6.87e-01	 & rs17258982	 &     \textit{CR1}	 & 0.027	 & 5.83e-01  \\
 rs4752856	 &   \textit{CELF1}	 & -0.009	 & 8.21e-01	 &  rs7143400	 &  \textit{FERMT2}	 & 0.026	 & 8.64e-01  \\
  rs558788	 &  \textit{MS4A6A}	 & -0.007	 & 6.15e-01	 &  rs3851179	 &  \textit{PICALM}	 & -0.021	 & 8.13e-01  \\
rs11952384	 &   \textit{MEF2C}	 & -0.007	 & 5.46e-01	 &   rs405697	 &    \textit{APOE}	 & -0.020	 & 2.60e-01  \\
 rs1784927	 &   \textit{SORL1}	 & -0.006	 & 2.66e-01	 &  rs4147932	 &   \textit{ABCA7}	 & 0.017	 & 3.80e-01  \\
 rs4720262	 &    \textit{NME8}	 & 0.006	 & 8.96e-01	 &  rs3818361	 &     \textit{CR1}	 & 0.017	 & 8.24e-03  \\
 rs8106922	 &    \textit{APOE}	 & -0.006	 & 2.27e-03	 & rs12961029	 &    \textit{DSG2}	 & 0.017	 & 1.13e-01  \\
rs12709651	 &    \textit{DSG2}	 & 0.005	 & 8.95e-01	 &  rs8008270	 &  \textit{FERMT2}	 & -0.016	 & 5.49e-02  \\
  rs244749	 &   \textit{MEF2C}	 & 0.005	 & 1.60e-01	 &  rs7941541	 &  \textit{PICALM}	 & -0.016	 & 7.83e-01  \\
  rs753812	 &   \textit{CELF1}	 & 0.005	 & 5.95e-01	 &  rs7160582	 &  \textit{FERMT2}	 & 0.016	 & 8.25e-01  \\
 rs2075650	 &    \textit{APOE}	 & 0.005	 & 4.46e-07	 & rs17125944	 &  \textit{FERMT2}	 & 0.015	 & 4.57e-01  \\
 rs7584458	 &  \textit{INPP5D}	 & -0.005	 & 4.06e-01	 & rs16979595	 &    \textit{APOE}	 & 0.014	 & 5.87e-01  \\
 rs7569827	 &  \textit{INPP5D}	 & -0.005	 & 3.62e-01	 &  rs4904920	 & \textit{SLC24A4}	 & 0.014	 & 8.38e-01  \\
 rs2104239	 &    \textit{RIN3}	 & 0.004	 & 1.72e-02	 &  rs2357947	 &  \textit{FERMT2}	 & 0.014	 & 8.50e-01  \\
rs10742816	 &   \textit{CELF1}	 & 0.004	 & 5.34e-01	 & rs11157933	 &  \textit{FERMT2}	 & 0.014	 & 8.50e-01  \\
 rs4752839	 &   \textit{CELF1}	 & 0.004	 & 5.04e-01	 &  rs6572869	 &  \textit{FERMT2}	 & 0.014	 & 8.50e-01  \\
 rs4663337	 &  \textit{INPP5D}	 & -0.004	 & 3.97e-01	 &  rs2405442	 &  \textit{ZCWPW1}	 & -0.013	 & 2.75e-01  \\
  rs254778	 &   \textit{MEF2C}	 & 0.004	 & 8.06e-01	 & rs11623185	 &    \textit{RIN3}	 & -0.013	 & 5.53e-01  \\
 rs1117067	 &  \textit{MS4A6A}	 & 0.004	 & 4.21e-01	 &  rs2104239	 &    \textit{RIN3}	 & 0.011	 & 1.72e-02  \\
rs11230193	 &  \textit{MS4A6A}	 & 0.004	 & 4.79e-01	 &  rs6951852	 &   \textit{EPHA1}	 & -0.011	 & 1.96e-01  \\
 rs4939319	 &  \textit{MS4A6A}	 & 0.004	 & 4.79e-01	 &  rs7580869	 &  \textit{INPP5D}	 & -0.011	 & 1.28e-01  \\
 rs7929057	 &  \textit{MS4A6A}	 & 0.004	 & 4.79e-01	 & rs10134832	 & \textit{SLC24A4}	 & -0.009	 & 4.92e-01  \\
 rs1866236	 &    \textit{BIN1}	 & 0.003	 & 1.01e-01	 &  rs1026123	 &    \textit{DSG2}	 & -0.009	 & 5.14e-01  \\
rs11218325	 &   \textit{SORL1}	 & 0.003	 & 1.10e-01	 &  rs1667280	 &    \textit{DSG2}	 & -0.009	 & 5.14e-01  \\
 rs1791161	 &    \textit{DSG2}	 & -0.003	 & 6.82e-01	 & rs12434016	 & \textit{SLC24A4}	 & -0.008	 & 9.11e-01  \\
 rs1871045	 &    \textit{APOE}	 & 0.003	 & 9.67e-01	 &   rs273622	 &    \textit{CD33}	 & 0.007	 & 2.38e-01  \\
 rs6069767	 &   \textit{CASS4}	 & 0.003	 & 4.47e-01	 &   rs660895	 & \textit{HLA-DRB1}	 & 0.007	 & 6.42e-01  \\
 rs4662703	 &    \textit{BIN1}	 & 0.003	 & 5.21e-01	 & rs17729233	 &    \textit{DSG2}	 & -0.006	 & 7.17e-01  \\
  rs757232	 &   \textit{ABCA7}	 & 0.003	 & 8.50e-02	 & rs12709651	 &    \textit{DSG2}	 & 0.003	 & 8.95e-01  \\
    rs7026	 &    \textit{APOE}	 & 0.003	 & 9.86e-01	 & rs17660414	 &    \textit{DSG2}	 & -0.003	 & 9.48e-01  \\
rs12476339	 &    \textit{BIN1}	 & 0.003	 & 5.53e-01	 &  rs1710354	 &    \textit{CD33}	 & -0.002	 & 2.95e-01  \\
  rs674747	 &   \textit{MEF2C}	 & 0.002	 & 4.33e-01	 & rs10413089	 &    \textit{APOE}	 & 0.002	 & 1.15e-01  \\
 rs4938933	 &  \textit{MS4A6A}	 & -0.002	 & 3.67e-01	 & rs12539172	 &  \textit{ZCWPW1}	 & -0.002	 & 5.44e-01  \\
rs17186722	 &     \textit{CR1}	 & -0.002	 & 5.20e-01	 & rs13426725	 &    \textit{BIN1}	 & 0.000	 & 1.20e-01  \\
 rs3752243	 &   \textit{ABCA7}	 & -0.002	 & 5.98e-01	 & rs10779277	 &     \textit{CR1}	 & 0.000	 & 2.92e-01  \\
 rs2161228	 &   \textit{MEF2C}	 & -0.002	 & 1.49e-01	 &  rs2490255	 &     \textit{CR1}	 & 0.000	 & 2.65e-01  \\
  rs543293	 &  \textit{PICALM}	 & -0.002	 & 7.78e-01	 & rs17186722	 &     \textit{CR1}	 & 0.000	 & 5.20e-01  \\
 rs3738468	 &     \textit{CR1}	 & -0.002	 & 6.84e-01	 &  rs2940252	 &     \textit{CR1}	 & 0.000	 & 5.18e-01  \\
  rs881768	 &   \textit{ABCA7}	 & -0.002	 & 6.04e-01	 &  rs2661361	 &     \textit{CR1}	 & 0.000	 & 3.76e-01  \\
  rs694011	 &  \textit{PICALM}	 & -0.002	 & 5.21e-01	 &  rs6664001	 &     \textit{CR1}	 & 0.000	 & 2.68e-01  \\
 rs4752845	 &   \textit{CELF1}	 & -0.002	 & 8.76e-01	 & rs17042520	 &     \textit{CR1}	 & 0.000	 & 6.56e-01  \\
rs12798346	 &   \textit{CELF1}	 & -0.002	 & 8.76e-01	 &  rs2135924	 &     \textit{CR1}	 & 0.000	 & 2.68e-01  \\
rs10838738	 &   \textit{CELF1}	 & -0.002	 & 8.76e-01	 &  rs6656123	 &     \textit{CR1}	 & 0.000	 & 3.10e-01  \\
 rs1871047	 &    \textit{APOE}	 & 0.002	 & 7.06e-01	 &   rs311299	 &     \textit{CR1}	 & 0.000	 & 3.73e-01  \\
 rs4726624	 &   \textit{EPHA1}	 & 0.002	 & 6.40e-01	 & rs12734973	 &     \textit{CR1}	 & 0.000	 & 5.06e-01  \\
 rs6951852	 &   \textit{EPHA1}	 & -0.001	 & 1.96e-01	 &  rs1032980	 &     \textit{CR1}	 & 0.000	 & 3.70e-01  \\
rs17014818	 &    \textit{BIN1}	 & 0.001	 & 5.08e-01	 &    rs17615	 &     \textit{CR1}	 & 0.000	 & 2.94e-01  \\
rs12155159	 &    \textit{NME8}	 & -0.001	 & 4.05e-01	 &  rs4308977	 &     \textit{CR1}	 & 0.000	 & 4.17e-01  \\
  rs676759	 &   \textit{SORL1}	 & -0.001	 & 5.51e-01	 &    rs17616	 &     \textit{CR1}	 & 0.000	 & 3.40e-01  \\
 rs8018746	 & \textit{SLC24A4}	 & -0.001	 & 3.49e-01	 &  rs7549152	 &     \textit{CR1}	 & 0.000	 & 5.70e-01  \\
 rs6591559	 &  \textit{MS4A6A}	 & -0.001	 & 3.62e-01	 &  rs2182909	 &     \textit{CR1}	 & 0.000	 & 3.73e-01  \\
 rs1530914	 &  \textit{MS4A6A}	 & -0.001	 & 3.62e-01	 &  rs6540433	 &     \textit{CR1}	 & 0.000	 & 5.06e-01  \\
rs17128308	 & \textit{SLC24A4}	 & 0.001	 & 1.86e-01	 &  rs6690215	 &     \textit{CR1}	 & 0.000	 & 9.38e-02  \\
 rs3752242	 &   \textit{ABCA7}	 & -0.001	 & 6.49e-01	 & rs12021671	 &     \textit{CR1}	 & 0.000	 & 1.26e-01  \\
 rs4904920	 & \textit{SLC24A4}	 & 0.001	 & 8.38e-01	 &  rs2182911	 &     \textit{CR1}	 & 0.000	 & 1.56e-01  \\
 rs3754617	 &    \textit{BIN1}	 & 0.001	 & 6.93e-01	 &  rs4618970	 &     \textit{CR1}	 & 0.000	 & 6.50e-01  \\
 rs2722246	 &    \textit{NME8}	 & -0.001	 & 9.90e-01	 &  rs9429940	 &     \textit{CR1}	 & 0.000	 & 6.50e-01  \\
  rs561655	 &  \textit{PICALM}	 & -0.001	 & 6.54e-01	 & rs11117956	 &     \textit{CR1}	 & 0.000	 & 1.90e-01  \\
  rs412458	 &   \textit{MEF2C}	 & 0.001	 & 3.16e-01	 & rs11117959	 &     \textit{CR1}	 & 0.000	 & 5.21e-01  \\
 rs7580869	 &  \textit{INPP5D}	 & -0.001	 & 1.28e-01	 & rs10127904	 &     \textit{CR1}	 & 0.000	 & 2.85e-02  \\
rs11117959	 &     \textit{CR1}	 & -0.001	 & 5.21e-01	 &  rs2274566	 &     \textit{CR1}	 & 0.000	 & 4.70e-02  \\
rs12883551	 & \textit{SLC24A4}	 & -0.000	 & 2.36e-01	 &  rs3738468	 &     \textit{CR1}	 & 0.000	 & 6.84e-01  \\
 rs2074442	 &   \textit{ABCA7}	 & 0.000	 & 1.12e-01	 & rs17259045	 &     \textit{CR1}	 & 0.000	 & 7.45e-01  \\
 rs4752993	 &   \textit{CELF1}	 & -0.000	 & 8.27e-01	 &  rs6691117	 &     \textit{CR1}	 & 0.000	 & 2.46e-01  \\
   rs12453	 &  \textit{MS4A6A}	 & -0.000	 & 6.72e-02	 & rs12032275	 &     \textit{CR1}	 & 0.000	 & 6.76e-01  \\
 rs1237999	 &  \textit{PICALM}	 & -0.000	 & 8.32e-01	 & rs12734030	 &     \textit{CR1}	 & 0.000	 & 3.09e-01  \\
  rs755553	 &   \textit{CELF1}	 & -0.000	 & 8.66e-01	 & rs12034383	 &     \textit{CR1}	 & 0.000	 & 2.51e-02  \\
rs10426423	 &    \textit{APOE}	 & 0.000	 & 6.63e-01	 & rs10779339	 &     \textit{CR1}	 & 0.000	 & 4.55e-01  \\
 rs7124060	 &   \textit{SORL1}	 & 0.000	 & 2.28e-01	 & rs10494885	 &     \textit{CR1}	 & 0.000	 & 4.65e-01  \\
rs10779277	 &     \textit{CR1}	 & -0.000	 & 2.92e-01	 &  rs6696840	 &     \textit{CR1}	 & 0.000	 & 4.86e-01  \\
 rs2490255	 &     \textit{CR1}	 & -0.000	 & 2.65e-01	 &  rs1323721	 &     \textit{CR1}	 & 0.000	 & 2.43e-01  \\
 rs2940252	 &     \textit{CR1}	 & -0.000	 & 5.18e-01	 & rs10863461	 &     \textit{CR1}	 & 0.000	 & 2.46e-01  \\
\end{longtable}
\end{center}


\begin{center}
\footnotesize
\begin{longtable}{lccccr}
\caption{Imaging features (ordered by absolute values of estimated canonical weights) selected by SCCA and simplified SCCA.}
\label{tab:QTs_selected} \\
\toprule
\multicolumn{3}{c}{Standard SCCA} & \multicolumn{3}{c}{Simplified SCCA} \\
\cmidrule(lr){1-3}\cmidrule(lr){4-6}
             brain ROI   & $\hat{v}_j$ & p-value &               brain ROI   & $\hat{v}_j$ & p-value \\
\midrule
\endfirsthead

\multicolumn{6}{c}%
{{\bfseries \tablename\ \thetable{} -- continued from previous page}} \\
\toprule
\multicolumn{3}{c}{Standard SCCA} & \multicolumn{3}{c}{Simplified SCCA} \\
\cmidrule(lr){1-3}\cmidrule(lr){4-6}
             brain ROI   & $\hat{v}_j$ & p-value &               brain ROI   & $\hat{v}_j$ & p-value \\
\midrule
\endhead

\hline \multicolumn{6}{r}{{Continued on next page}} \\ \hline
\endfoot

\bottomrule \bottomrule
\endlastfoot

       Hippocampus\_L	 & -0.403	 & 1.25e-08	 &    Frontal\_Med\_Orb\_L	 & 0.138	 & 9.65e-26  \\
        Frontal\_Mid\_R	 & 0.279	 & 4.84e-18	 & Frontal\_Sup\_Medial\_L	 & 0.135	 & 8.66e-21  \\
        Frontal\_Mid\_L	 & 0.261	 & 1.67e-18	 &        Cingulum\_Ant\_L	 & 0.133	 & 2.32e-19  \\
          Precentral\_L	 & -0.249	 & 7.67e-07	 &    Frontal\_Med\_Orb\_R	 & 0.133	 & 1.04e-24  \\
      Rolandic\_Oper\_L	 & -0.238	 & 4.63e-10	 & Frontal\_Sup\_Medial\_R	 & 0.132	 & 4.47e-20  \\
Frontal\_Sup\_Medial\_L	 & 0.235	 & 8.66e-21	 &               Rectus\_L	 & 0.132	 & 3.33e-25  \\
        Cerebelum\_6\_R	 & 0.219	 & 5.71e-10	 &         Frontal\_Mid\_R	 & 0.130	 & 4.84e-18  \\
           Calcarine\_R	 & -0.216	 & 5.11e-13	 &    Frontal\_Mid\_Orb\_R	 & 0.129	 & 5.09e-21  \\
              Insula\_R	 & 0.206	 & 1.67e-16	 &         Frontal\_Mid\_L	 & 0.129	 & 1.67e-18  \\
       Cingulum\_Ant\_L	 & 0.188	 & 2.32e-19	 &        Temporal\_Mid\_R	 & 0.128	 & 2.15e-20  \\
 Temporal\_Pole\_Mid\_R	 & -0.187	 & 1.39e-06	 &               Rectus\_R	 & 0.128	 & 4.53e-22  \\
             Caudate\_L	 & 0.185	 & 1.38e-01	 &    Frontal\_Sup\_Orb\_R	 & 0.128	 & 3.23e-20  \\
          Precentral\_R	 & -0.179	 & 1.25e-05	 &               Insula\_R	 & 0.127	 & 1.67e-16  \\
              Vermis\_8	 & 0.171	 & 9.35e-01	 &        Temporal\_Inf\_R	 & 0.127	 & 6.79e-19  \\
       Temporal\_Inf\_R	 & 0.169	 & 6.79e-19	 &    Frontal\_Sup\_Orb\_L	 & 0.127	 & 7.41e-20  \\
              Cuneus\_R	 & -0.165	 & 9.59e-07	 &    Frontal\_Mid\_Orb\_L	 & 0.126	 & 3.34e-21  \\
           Olfactory\_L	 & 0.130	 & 1.46e-13	 &    Frontal\_Inf\_Orb\_R	 & 0.126	 & 2.57e-14  \\
              Heschl\_R	 & 0.130	 & 6.06e-17	 &            Olfactory\_L	 & 0.125	 & 1.46e-13  \\
      Occipital\_Inf\_L	 & 0.112	 & 2.30e-13	 &        Cingulum\_Mid\_L	 & 0.125	 & 8.86e-22  \\
        Cerebelum\_9\_L	 & -0.112	 & 1.18e-03	 &        Cingulum\_Mid\_R	 & 0.125	 & 2.12e-19  \\
            Thalamus\_R	 & 0.108	 & 9.22e-01	 &    Frontal\_Inf\_Orb\_L	 & 0.124	 & 1.64e-17  \\
        Cerebelum\_3\_L	 & -0.107	 & 2.41e-05	 &        Cingulum\_Ant\_R	 & 0.123	 & 6.76e-15  \\
             Putamen\_L	 & 0.105	 & 2.11e-17	 &         Frontal\_Sup\_R	 & 0.123	 & 1.35e-14  \\
   Frontal\_Med\_Orb\_L	 & 0.098	 & 9.65e-26	 &        Temporal\_Sup\_R	 & 0.123	 & 2.06e-20  \\
       Temporal\_Mid\_R	 & 0.097	 & 2.15e-20	 &        Temporal\_Mid\_L	 & 0.121	 & 1.94e-21  \\
      Occipital\_Mid\_L	 & 0.090	 & 1.55e-09	 &            Precuneus\_L	 & 0.121	 & 8.66e-22  \\
   Frontal\_Inf\_Orb\_R	 & 0.081	 & 2.57e-14	 &            Olfactory\_R	 & 0.120	 & 7.48e-11  \\
   Frontal\_Mid\_Orb\_R	 & 0.073	 & 5.09e-21	 &            Precuneus\_R	 & 0.120	 & 8.93e-23  \\
           Olfactory\_R	 & 0.069	 & 7.48e-11	 &    Frontal\_Inf\_Tri\_L	 & 0.120	 & 4.56e-16  \\
              Vermis\_3	 & 0.059	 & 1.05e-01	 &        Temporal\_Inf\_L	 & 0.119	 & 5.09e-19  \\
        Cerebelum\_3\_R	 & -0.053	 & 1.33e-05	 &        Temporal\_Sup\_L	 & 0.119	 & 8.89e-17  \\
     Cerebelum\_4\_5\_R	 & -0.052	 & 5.88e-09	 &         Frontal\_Sup\_L	 & 0.119	 & 6.30e-15  \\
              Cuneus\_L	 & -0.052	 & 5.56e-06	 &        Parietal\_Inf\_L	 & 0.119	 & 2.94e-15  \\
        Frontal\_Sup\_R	 & 0.051	 & 1.35e-14	 &        SupraMarginal\_R	 & 0.118	 & 7.04e-15  \\
       Cerebelum\_10\_L	 & 0.044	 & 1.02e-04	 &    Frontal\_Inf\_Tri\_R	 & 0.117	 & 4.50e-13  \\
       Cerebelum\_7b\_L	 & 0.042	 & 4.67e-07	 &              Angular\_R	 & 0.117	 & 5.56e-16  \\
         Hippocampus\_R	 & -0.034	 & 4.66e-08	 &              Angular\_L	 & 0.116	 & 6.30e-17  \\
     Cerebelum\_4\_5\_L	 & -0.033	 & 1.29e-04	 &        Parietal\_Inf\_R	 & 0.115	 & 7.79e-14  \\
        Cerebelum\_6\_L	 & 0.032	 & 1.69e-09	 &               Insula\_L	 & 0.115	 & 4.36e-14  \\
       Cingulum\_Mid\_R	 & 0.023	 & 2.12e-19	 &               Heschl\_R	 & 0.114	 & 6.06e-17  \\
   Supp\_Motor\_Area\_R	 & -0.010	 & 7.83e-15	 &        SupraMarginal\_L	 & 0.113	 & 2.29e-11  \\
      Cingulum\_Post\_R	 & -0.009	 & 3.18e-04	 &   Frontal\_Inf\_Oper\_R	 & 0.112	 & 2.22e-14  \\
            Fusiform\_R	 & -0.009	 & 1.94e-20	 &       Rolandic\_Oper\_R	 & 0.111	 & 2.84e-13  \\
         Postcentral\_L	 & -0.009	 & 2.91e-09	 &    Supp\_Motor\_Area\_L	 & 0.110	 & 1.13e-15  \\
   Frontal\_Mid\_Orb\_L	 & 0.008	 & 3.34e-21	 &             Fusiform\_R	 & 0.110	 & 1.94e-20  \\
         Postcentral\_R	 & -0.008	 & 1.85e-08	 &       Cingulum\_Post\_L	 & 0.108	 & 5.66e-13  \\
           Calcarine\_L	 & -0.008	 & 3.33e-17	 &             Fusiform\_L	 & 0.108	 & 4.12e-19  \\
  Frontal\_Inf\_Oper\_R	 & -0.008	 & 2.22e-14	 &   Frontal\_Inf\_Oper\_L	 & 0.106	 & 4.89e-12  \\
             Lingual\_L	 & -0.008	 & 2.68e-15	 &              Putamen\_L	 & 0.106	 & 2.11e-17  \\
       Cingulum\_Mid\_L	 & 0.007	 & 8.86e-22	 &              Putamen\_R	 & 0.106	 & 4.16e-15  \\
       Parietal\_Inf\_L	 & 0.007	 & 2.94e-15	 &               Heschl\_L	 & 0.105	 & 1.05e-14  \\
Frontal\_Sup\_Medial\_R	 & 0.007	 & 4.47e-20	 &  Temporal\_Pole\_Sup\_L	 & 0.104	 & 3.30e-08  \\
       Temporal\_Sup\_L	 & 0.007	 & 8.89e-17	 &  Temporal\_Pole\_Sup\_R	 & 0.104	 & 4.33e-09  \\
     ParaHippocampal\_R	 & -0.007	 & 4.46e-01	 &       Occipital\_Mid\_L	 & 0.101	 & 1.55e-09  \\
       Temporal\_Sup\_R	 & 0.006	 & 2.06e-20	 &       Occipital\_Inf\_L	 & 0.100	 & 2.30e-13  \\
 Paracentral\_Lobule\_R	 & -0.006	 & 3.84e-12	 &    Supp\_Motor\_Area\_R	 & 0.098	 & 7.83e-15  \\
             Lingual\_R	 & -0.006	 & 7.85e-17	 &       Rolandic\_Oper\_L	 & 0.095	 & 4.63e-10  \\
 Temporal\_Pole\_Sup\_L	 & 0.005	 & 3.30e-08	 &        Parietal\_Sup\_L	 & 0.094	 & 3.29e-10  \\
 Paracentral\_Lobule\_L	 & -0.005	 & 5.85e-08	 &  Temporal\_Pole\_Mid\_L	 & 0.093	 & 2.11e-07  \\
           Vermis\_1\_2	 & 0.005	 & 6.11e-04	 &       Occipital\_Mid\_R	 & 0.092	 & 1.56e-08  \\
      Occipital\_Sup\_L	 & -0.005	 & 1.43e-02	 &       Occipital\_Inf\_R	 & 0.091	 & 1.26e-09  \\
      Occipital\_Inf\_R	 & 0.005	 & 1.26e-09	 &            Calcarine\_L	 & 0.091	 & 3.33e-17  \\
      Cingulum\_Post\_L	 & -0.004	 & 5.66e-13	 &  Temporal\_Pole\_Mid\_R	 & 0.088	 & 1.39e-06  \\
 Temporal\_Pole\_Mid\_L	 & -0.004	 & 2.11e-07	 &          Postcentral\_R	 & 0.087	 & 1.85e-08  \\
            Fusiform\_L	 & -0.004	 & 4.12e-19	 &          Postcentral\_L	 & 0.086	 & 2.91e-09  \\
            Pallidum\_R	 & -0.004	 & 3.78e-03	 &  Paracentral\_Lobule\_R	 & 0.084	 & 3.84e-12  \\
       Parietal\_Sup\_R	 & -0.003	 & 3.11e-05	 &              Lingual\_R	 & 0.081	 & 7.85e-17  \\
            Pallidum\_L	 & -0.002	 & 4.55e-02	 &           Precentral\_L	 & 0.078	 & 7.67e-07  \\
             Caudate\_R	 & -0.002	 & 5.42e-01	 &              Lingual\_L	 & 0.078	 & 2.68e-15  \\
              Vermis\_9	 & 0.002	 & 9.51e-02	 &             Amygdala\_L	 & 0.078	 & 8.42e-07  \\
              Vermis\_7	 & -0.002	 & 4.63e-02	 &       Cingulum\_Post\_R	 & 0.076	 & 3.18e-04  \\
    Cerebelum\_Crus1\_R	 & -0.002	 & 3.23e-03	 &           Precentral\_R	 & 0.074	 & 1.25e-05  \\
        Cerebelum\_8\_R	 & -0.001	 & 3.11e-06	 &             Amygdala\_R	 & 0.073	 & 1.09e-02  \\
       Cerebelum\_7b\_R	 & 0.001	 & 6.70e-05	 &            Calcarine\_R	 & 0.072	 & 5.11e-13  \\
  Frontal\_Inf\_Oper\_L	 & -0.001	 & 4.89e-12	 &        Parietal\_Sup\_R	 & 0.071	 & 3.11e-05  \\
   Frontal\_Inf\_Orb\_L	 & 0.001	 & 1.64e-17	 &               Cuneus\_L	 & 0.071	 & 5.56e-06  \\
     ParaHippocampal\_L	 & -0.001	 & 4.57e-01	 &  Paracentral\_Lobule\_L	 & 0.068	 & 5.85e-08  \\
            Thalamus\_L	 & 0.001	 & 2.74e-01	 &       Occipital\_Sup\_R	 & 0.067	 & 4.50e-05  \\
   Supp\_Motor\_Area\_L	 & -0.000	 & 1.13e-15	 &      ParaHippocampal\_R	 & 0.064	 & 4.46e-01  \\
        Frontal\_Sup\_L	 & -0.000	 & 6.30e-15	 &         Cerebelum\_6\_R	 & 0.062	 & 5.71e-10  \\
   Frontal\_Sup\_Orb\_L	 & -0.000	 & 7.41e-20	 &       Occipital\_Sup\_L	 & 0.059	 & 1.43e-02  \\
   Frontal\_Sup\_Orb\_R	 & -0.000	 & 3.23e-20	 &               Cuneus\_R	 & 0.059	 & 9.59e-07  \\
   Frontal\_Inf\_Tri\_L	 & -0.000	 & 4.56e-16	 &              Caudate\_R	 & 0.057	 & 5.42e-01  \\
   Frontal\_Inf\_Tri\_R	 & -0.000	 & 4.50e-13	 &             Pallidum\_R	 & 0.056	 & 3.78e-03  \\
      Rolandic\_Oper\_R	 & -0.000	 & 2.84e-13	 &              Caudate\_L	 & 0.054	 & 1.38e-01  \\
   Frontal\_Med\_Orb\_R	 & -0.000	 & 1.04e-24	 &         Cerebelum\_6\_L	 & 0.051	 & 1.69e-09  \\
              Rectus\_L	 & -0.000	 & 3.33e-25	 &      ParaHippocampal\_L	 & 0.051	 & 4.57e-01  \\
              Rectus\_R	 & -0.000	 & 4.53e-22	 &             Pallidum\_L	 & 0.047	 & 4.55e-02  \\
              Insula\_L	 & -0.000	 & 4.36e-14	 &         Cerebelum\_3\_L	 & -0.044	 & 2.41e-05  \\
       Cingulum\_Ant\_R	 & -0.000	 & 6.76e-15	 &         Cerebelum\_8\_R	 & -0.042	 & 3.11e-06  \\
            Amygdala\_L	 & -0.000	 & 8.42e-07	 &         Cerebelum\_8\_L	 & -0.042	 & 2.15e-05  \\
            Amygdala\_R	 & -0.000	 & 1.09e-02	 &      Cerebelum\_4\_5\_R	 & 0.041	 & 5.88e-09  \\
      Occipital\_Sup\_R	 & -0.000	 & 4.50e-05	 &        Cerebelum\_7b\_L	 & -0.040	 & 4.67e-07  \\
      Occipital\_Mid\_R	 & -0.000	 & 1.56e-08	 &     Cerebelum\_Crus2\_L	 & -0.039	 & 3.35e-06  \\
       Parietal\_Sup\_L	 & -0.000	 & 3.29e-10	 &         Cerebelum\_9\_L	 & -0.038	 & 1.18e-03  \\
       Parietal\_Inf\_R	 & -0.000	 & 7.79e-14	 &     Cerebelum\_Crus2\_R	 & -0.038	 & 8.42e-06  \\
       SupraMarginal\_L	 & -0.000	 & 2.29e-11	 &             Thalamus\_R	 & 0.035	 & 9.22e-01  \\
       SupraMarginal\_R	 & -0.000	 & 7.04e-15	 &         Cerebelum\_9\_R	 & -0.033	 & 3.81e-04  \\
             Angular\_L	 & -0.000	 & 6.30e-17	 &         Cerebelum\_3\_R	 & -0.032	 & 1.33e-05  \\
             Angular\_R	 & -0.000	 & 5.56e-16	 &        Cerebelum\_10\_R	 & -0.030	 & 7.99e-05  \\
           Precuneus\_L	 & -0.000	 & 8.66e-22	 &             Thalamus\_L	 & 0.029	 & 2.74e-01  \\
           Precuneus\_R	 & -0.000	 & 8.93e-23	 &        Cerebelum\_7b\_R	 & -0.029	 & 6.70e-05  \\
             Putamen\_R	 & -0.000	 & 4.16e-15	 &            Vermis\_1\_2	 & -0.023	 & 6.11e-04  \\
              Heschl\_L	 & -0.000	 & 1.05e-14	 &               Vermis\_7	 & -0.022	 & 4.63e-02  \\
 Temporal\_Pole\_Sup\_R	 & -0.000	 & 4.33e-09	 &            Vermis\_4\_5	 & 0.021	 & 5.31e-04  \\
       Temporal\_Mid\_L	 & -0.000	 & 1.94e-21	 &     Cerebelum\_Crus1\_R	 & -0.021	 & 3.23e-03  \\
       Temporal\_Inf\_L	 & -0.000	 & 5.09e-19	 &               Vermis\_8	 & 0.016	 & 9.35e-01  \\
    Cerebelum\_Crus1\_L	 & -0.000	 & 7.15e-02	 &        Cerebelum\_10\_L	 & -0.016	 & 1.02e-04  \\
    Cerebelum\_Crus2\_L	 & -0.000	 & 3.35e-06	 &              Vermis\_10	 & -0.014	 & 9.29e-02  \\
    Cerebelum\_Crus2\_R	 & -0.000	 & 8.42e-06	 &     Cerebelum\_Crus1\_L	 & -0.012	 & 7.15e-02  \\
        Cerebelum\_8\_L	 & -0.000	 & 2.15e-05	 &               Vermis\_9	 & 0.011	 & 9.51e-02  \\
        Cerebelum\_9\_R	 & -0.000	 & 3.81e-04	 &               Vermis\_3	 & -0.008	 & 1.05e-01  \\
       Cerebelum\_10\_R	 & -0.000	 & 7.99e-05	 &          Hippocampus\_R	 & 0.007	 & 4.66e-08  \\
           Vermis\_4\_5	 & -0.000	 & 5.31e-04	 &               Vermis\_6	 & 0.003	 & 1.60e-01  \\
              Vermis\_6	 & -0.000	 & 1.60e-01	 &          Hippocampus\_L	 & -0.003	 & 1.25e-08  \\
             Vermis\_10	 & -0.000	 & 9.29e-02	 &      Cerebelum\_4\_5\_L	 & 0.000	 & 1.29e-04  \\
\end{longtable}
\end{center}

\end{document}